\documentclass[12pt]{article}

\usepackage{amsmath}
\usepackage{epsfig, amssymb, amsfonts}
\usepackage{amsthm}
\usepackage{amstext}
\usepackage{amsopn}
\usepackage{enumerate}
\usepackage{mathrsfs}
\usepackage{subfigure}
\usepackage{graphicx}
\usepackage[left=2.3cm,top=2cm,right=2.3cm,bottom=2cm, nohead,foot=1cm]{geometry}
\numberwithin{equation}{section}

\long\def\symbolfootnote[#1]#2{\begingroup%
\def\thefootnote{\fnsymbol{footnote}}\footnote[#1]{#2}\endgroup}

\newtheorem{theorem}{Theorem}[section]
\newtheorem{lemma}[theorem]{Lemma}

\newtheorem{proposition}[theorem]{Proposition}

\newcommand{\R}{{\mathbb R}}

\newcommand{\ii}{{\textup{i}}}
\newcommand{\mcj}{{\mathcal{J}}}
\newcommand{\mcl}{{\mathcal{L}}}
\newcommand{\mce}{{\mathcal{E}}}
\newcommand{\mcg}{{\mathcal{G}}}

\newcommand{\Z}{{\mathbb Z}}

\newcommand{\Tr}{{\textup{Tr}}}

\newcommand{\C}{{\mathbb C}}

\newcommand{\baq}{\begin{eqnarray}}
\newcommand{\eaq}{\end{eqnarray}}
\newcommand{\beq}{\begin{equation}}
\newcommand{\eeq}{\end{equation}}

\title{Diffusive limit for a quantum linear  Boltzmann dynamics}

\author{\textbf{Jeremy Thane Clark} \vspace{.2cm}\\ jclark@mappi.helsinki.fi \\ University of Helsinki, Department of Mathematics\\  Helsinki 00014, Finland }

\begin{document}
\maketitle

\begin{abstract}
In this article, I study the diffusive behavior for a quantum test particle interacting with a dilute background gas.  The model I  begin with is a reduced picture for the test particle dynamics given by a quantum linear Boltzmann equation in which the gas particle scattering is assumed to occur through a hard-sphere interaction.  The state of the particle is represented by a density matrix that evolves according to a translation-covariant Lindblad equation.  The main result is a proof that the particle's position distribution converges to a Gaussian under diffusive rescaling.

\end{abstract}

\section{Introduction}

Boltzmann's transport equation generates an idealized dissipative dynamics for the  one-particle phase space density of a gas of interacting particles~\cite{Boltzmann,Cerc}.  This dynamics is relevant for a dilute gas in which the essential information for the interaction between gas particles is reduced to two-particle scattering data.  Boltzmann's equation effectively governs the convergence of the one-particle phase space density to the distribution corresponding to the thermal equilibrium state.  The linear Boltzmann equation is similar in character but governs the dissipative evolution for the one-particle phase space density of a population of (non self-interacting) test particles interacting with a dilute background gas in equilibrium.  By normalizing the density by the number of test particles, the linear Boltzmann equation serves as a Markovian master equation describing the noisy trajectory of a single test particle in the reservoir.   A mathematical derivation of Boltzmann's equation from a microscopic model for a gas of hard-spheres has appeared in the  article~\cite{Landford}, which verifies a mathematical conjecture in~\cite{Gallavotti} concerning what is referred to as the \textit{Boltzmann-Grad limit}.  A derivation of the linear Boltzmann equation from a microscopic model for a test particle interacting with a background gas follows from the same analysis and is discussed in the review~\cite{Spohn2}. Starting from a quantum version of a Lorentz gas, a linear Boltzmann equation has also been derived in a low density limit in~\cite{Eng}.

The linear Boltzmann equation is a classical model treating the test particle as a point.  A motivation for having a quantum dissipative dynamics rather than a classical one is to have a model in which the respective time scales of dissipation and decoherence for a quantum particle interacting with a gas are both apparent.   The Lindblad equation, serving as a quantum linear Boltzmann equation, should be determined by the essential features of the gas (e.g. spatial density, temperature) and the scattering behavior for the interaction between the test particle and a single gas particle.     

There have been various quantum master equations proposed in the physics literature to model the dissipative/decoherence influence of a background gas acting on a test particle~\cite{Joos,Gallis,Diosi,Vacchini,Hornberger}.  With the microscopic picture for the test particle interacting with an ideal gas as inspiration, the derivations have all been heuristic or  based on an analysis of the effect of a single scattering with a gas particle.  The first heuristic derivations were~\cite{Joos} and~\cite{Gallis} where the authors considered  limits in which the test particle becomes infinitely massive compared to a reservoir particle.  The limiting dynamics in those cases are transient in nature, since  by neglecting frictional effects, they predict linear growth in the mean kinetic energy. The first proposed model  to include energy relaxation was in~\cite{Diosi}.  The model was later refined and given clearer motivation in~\cite{Vacchini,Vacchini2} and  subsequently~\cite{Hornberger,HornVacc}.  A comprehensive review of the  literature can be found in~\cite{VaccHorn}.


In this article, I begin with the quantum linear Boltzmann equation discussed in~\cite{VaccHorn} for the special case in which the interaction between the test particle and a single reservoir particle is a hard-sphere potential and the reservoir particles are distinguishable.  The main result is a proof that the particle diffuses spatially in the limit of large times.  My analysis yields a closed form for the diffusion constant in terms of the scattering amplitude for the hard-sphere potential.  The techniques here are an extension of those in~\cite{Diff}, which require   the Lindblad generator and other operators involved be bounded. 

 I will mention a few recent mathematical results regarding diffusion in quantum models before proceeding to define my model.  The articles~\cite{Salmhofer,Salmhofer2} prove diffusion for a quantum particle in a random potential (quantum Lorentz gas) in a weak coupling limit.  The article~\cite{Schenker} shows diffusion for a quantum particle in a tight-binding approximation (living on $\Z^{d}$) whose wave packet evolves according to a stochastic Schr\"odinger equation. The results of ~\cite{Schenker} are extended in~\cite{Hamza}  to handle a weaker form of stochasticity in which the random potential forcing the particle is periodic rather than evolving independently at each lattice site.  Also for tight-binding models,~\cite{DFP} and~\cite{DF} prove diffusion for microscopic models of a particle interacting with an environment.  The environment in~\cite{DFP} is a lattice of non-interacting reservoirs that are taken to interact with the test particle in a weak coupling limit.  The result in~\cite{DF} treats a quantum particle interacting with a free Bosonic field in the limit that the particle is massive (weak coupling). 

In Sect.~\ref{SecDisc}, I present the model, state the main theorem regarding the diffusive behavior of the model, and give some background information.  Section~\ref{SecCons} discusses the technical questions regarding the existence of the dynamics I consider.  Sections~\ref{SecFiberDec}-\ref{SecDiffConstant} build up the technical tools for the perturbation argument that I employ.   Finally, the proof of the result is contained in Sect.~\ref{SecProofMain}.

\section{Model, results, and discussion}\label{SecDisc}

\subsection{The model and the main theorem}

The  state for the test particle at a given time $t\in \R_{+}$ will be described by a density matrix $\rho_{t}\in \mathcal{B}_{1}\big(L^{2}(\R^{3})\big)$, where $\mathcal{B}_{1}(\mathcal{H})$ denotes  the space of trace class operators over a Hilbert space $\mathcal{H}$.  The state $\rho_{t}$ evolves according to the Lindblad equation  
\begin{align}\label{MainDyn}
\hspace{3cm}\frac{d}{dt}\rho_{t}=\mathcal{L}(\rho_{t})=-\frac{\ii}{\hbar} \big[H,\,\rho_{t} \big]+\Psi(\rho_{t})-\frac{1}{2}\{\Psi^{*}(\textup{I}),\rho_{t}\}, \quad \quad \quad \rho_{0}=\rho,     
\end{align}
where the Hamiltonian $H=H^{*}$   is a function  $H(\vec{P})$ of the vector of momentum operators $\vec{P}$ with $P_{j}=-\ii\,\hbar\frac{\partial}{\partial x_{j}}$ for $j=1,\,2,\,3$, $\Psi$ is a completely positive map defined below, and $\Psi^{*}(\textup{I})$ is the evaluation of the identity for the adjoint map $\Psi^{*}$.  The term $-\frac{\ii}{\hbar} \big[H,\,\cdot \big]$ generates the deterministic part of the test particle dynamics, and $\Psi$ determines the dissipative influence of the gas as an idealized noise reservoir for the evolution of the test particle.   The Hamiltonian $H$ is equal to $\frac{1}{M}\vec{P}^{2}+H_{f}$, where $H_{f}$ is an energy shift related to forward scattering with the gas reservoir particles and  is defined below.

The forms of $\Psi$ and $H$, which I now define, are those suggested in~\cite{VaccHorn} for a quantum linear Boltzmann equation modeling a test particle in dilute gas of distinguishable particles.  The map $\Psi $ is taken to be of the form
\begin{align}\label{Jeez}
\Psi(\rho)=\int_{\R^{3}}d\vec{q}\,e^{\ii\frac{\vec{q}}{\hbar}\vec{X}}\mathbf{T}_{\vec{q}}(\rho)e^{-\ii\frac{\vec{q}}{\hbar}\vec{X} },   
\end{align}
where $\mathbf{T}_{\vec{q}}$ are completely positive maps that act as multiplication in the momentum representation
\beq \label{Ems}
\mathbf{T}_{\vec{q}}(\rho)(\vec{p}_{1},\vec{p}_{2})=\textup{T}_{\vec{q}}(\vec{p}_{1},\vec{p}_{2})\rho(\vec{p}_{1},\vec{p}_{2}),\quad \quad \quad \rho\in\mathcal{B}_{1}\big(L^{2}(\R^{3})\big),  
\eeq
for some functions $\textup{T}_{\vec{q}}:\R^{3}\times \R^{3}\rightarrow \C $ determined below.  To define the maps $\mathbf{T}_{\vec{q}} $, I must introduce some notation.  Let $m$, $M$, $m_{*}=\frac{mM}{m+M}$ refer to the mass of a single reservoir particle, the mass of the test particle, and the \textit{relative mass} respectively.  Let $\eta$ be the density of the gas and  $r(\vec{k})=(\frac{\beta}{2\pi m})^{\frac{3}{2}}  e^{- \frac{\beta}{2m}\vec{k}^{2} } $ be the momentum distribution for a single reservoir particle.  Define $\mathbf{f}:\R_{+}\times [0,\pi]\rightarrow \C$ to be the scattering amplitude determined by the interaction potential between the test particle and a gas particle, and denote $ f(\vec{p}_{2},\vec{p}_{1}):= \mathbf{f}(\mathbf{p},\theta)$, where $\mathbf{p}=\|\vec{p}_{1}\|_{2}=\|\vec{p}_{2}\|_{2}$ and  $\theta$ is the angle between $\vec{p}_{1}$ and $\vec{p}_{2}$. In my case, the scattering amplitude $\mathbf{f}(\mathbf{p},\theta)$ corresponds to the hard-sphere interaction, and I discuss its form in Sect.~\ref{SecDis}.  Finally, the completely positive maps $\mathbf{T}_{\vec{q}}$ have a Kraus form
\begin{align}\label{Kraus}
\mathbf{T}_{\vec{q}}(\rho)=\int_{ (\vec{q})_{\perp} }d\vec{v}\,  L_{\vec{q},\vec{v}}\,\rho\, L_{\vec{q},\vec{v}}^{*},   
\end{align}
where $(\vec{q})_{\perp}$ is the orthogonal complement of the one-dimensional space spanned by $\vec{q}$ and the operators $L_{\vec{q},\vec{v}}=L_{\vec{q},\vec{v}}(\vec{P})$ act as multiplication in the momentum representation as
\begin{multline}\label{ElementKraus}
L_{\vec{q},\vec{v}}(\vec{P})=\eta^{\frac{1}{2}}\frac{m^{\frac{1}{2}} }{m_{*}}\Big(\frac{1}{\,|\vec{q}| }\Big)^{\frac{1}{2}}\Big(r\big(\vec{v}+2^{-1}\frac{m}{m_{*}}\vec{q}+\frac{m}{M}\vec{P}_{\parallel \vec{q}  }    \big) \Big)^{\frac{1}{2}} \\ \times f\Big(\frac{m_{*}}{m}\vec{v}-\frac{m_{*}}{M}\vec{P}_{\perp \vec{q}}-2^{-1}\vec{q} , \,    \frac{m_{*}}{m}\vec{v}-\frac{m_{*}}{M}\vec{P}_{\perp \vec{q}}+2^{-1}\vec{q}\Big).
\end{multline}
In the above, the vector of operators $\vec{P}_{\parallel \vec{q}}$ and $\vec{P}_{\perp \vec{q}}$ have the forms $\vec{P}_{\parallel \vec{q}}= \frac{\vec{q}\cdot\vec{P}}{|\vec{q}|^{2}  }\,\vec{q} $ and  $\vec{P}_{\perp \vec{q}}=\vec{P}-\vec{P}_{\parallel \vec{q}}$.

The Hamiltonian $H_{f}$ will be given by
\beq \label{AltDisp}
H_{f}=H_{f}(\vec{P})=-\frac{2\pi\hbar^{2}\eta}{ m_{*}}\int_{\R^{3}}d\vec{q}\,r(\vec{q})\,\textup{Re}\big[f\big(\frac{m_{*}}{m}\vec{q}-\frac{m_{*}}{M}\vec{P},\, \frac{m_{*}}{m}\vec{q}-\frac{m_{*}}{M}\vec{P}\big)   \big].
\eeq

Theorem~\ref{MainThm} is the main result of this article and states that the particle behaves diffusively for  some diffusion constant $D$.  A closed expression for $D$ is discussed in Sect.~\ref{SecDis}.  The conditions for the theorem are in terms of the initial density matrix $\rho\in \mathcal{B}_{1}\big(L^{2}(\R^{3})\big)$.  The statement that $Y^{*}_{1}\rho Y_{2}$ is trace class for unbounded operators $Y_{1},\, Y_{2}$ with dense domains $\textup{D}(Y_{1})$, $\textup{D}(Y_{2})$ respectively, means that the corresponding quadratic form $Q(\psi_{1},\psi_{2})=\langle Y_{1}\psi_{1}|\rho \,Y_{2}\psi_{2}\rangle$ for $\psi_{1}\in \textup{D}(Y_{1})$ and $\psi_{2}\in \textup{D}(Y_{2})$ is bounded and determines a trace class operator for its kernel. 

\begin{theorem}\label{MainThm}
Assume that $X_{j}\rho$, $j=1,\,2,\,3$ are trace class and let $\mu_{t}(\vec{x}):=\rho_{t}(\vec{x},\vec{x})$ be the position density at time $t$.  In the limit $t\rightarrow \infty$, the rescaled density $t^{\frac{3}{2}}\mu_{t}(t^{\frac{1}{2}}\vec{r})$, $\vec{r}\in \R^{3}$ converges in distribution to a Gaussian with covariance matrix $D\,I_{3}$, $D>0$.  Moreover, if in addition $G_{1}G_{2}\rho$ and  $G_{1}\rho G_{2}$
are trace class for all $G_{1},G_{2}\in\{X_{j},P_{j},\,j=1,2,3 \}$ and $|\vec{P}|^{3-n}\rho|\vec{P}|^{n}$ is trace class for $n=0,\,1,\,2,\,3$, then the second moments satisfy
$$ t^{-1}\int_{\R^{3}}d\vec{x}\,\mu_{t}(\vec{x})\,x_{i}x_{j}=   \delta_{i,j} D+\mathit{O}(t^{-1}) \quad \text{as}\quad t\rightarrow \infty  .$$

\end{theorem}

\subsection{Discussion}\label{SecDis}

\subsubsection{Preliminary mathematical considerations  }
The Lindblad equation~(\ref{MainDyn}) determines a trace-preserving collection of  maps $\Phi_{t}:\mathcal{B}_{1}\big(L^{2}(\R^{3})\big)$ for $t\in \R_{+}$ determined by $\Phi_{t}(\rho)=\rho_{t}$.  The maps $\Phi_{t}$ are completely positive, satisfy the semigroup property $\Phi_{t}\Phi_{s}=\Phi_{t+s}$, and are strongly continuous.  The equation~(\ref{MainDyn}) has the generic form for the time derivative of a norm-continuous dynamical semigroup~\cite{Lindblad} except that  the generator $\mathcal{L}$ I consider is unbounded.  It is therefore necessary to be cautious with respect to the mathematical interpretation of the Lindblad equation~(\ref{MainDyn}).  Analogously to unbounded master equations in classical probability theory~\cite{Feller}, a solution to an unbounded Lindblad equation may be constructible without being unique or preserving the trace the density matrices~\cite{Davies}.  Intuitively, this would correspond to situations in which  infinitely many jumps may occur in a finite time period and there may be ``escape to infinity" in finite time.  In the case of translation-invariant Lindblad equations, this problem can be reduced to a classical one~\cite{Holevo}.  The uniqueness of the solution is implied by the probability preservation for an embedded classical Markovian semigroup that governs the momentum distribution.  In my case, the master equation governing the momentum density $\nu_{t}(\vec{p})=\rho_{t}(\vec{p},\vec{p})$ has the form 
\beq \label{MomDist}
\frac{d}{dt}\nu_{t}(\vec{p})= \mathcal{L}_{0}(\nu_{t})(\vec{p})= \int_{\R^{3}}d\vec{p}_{0} \big(\mcj (\vec{p},\vec{p}_{0}) \nu_{t}(\vec{p}_{0}) -\mcj(\vec{p}_{0},\vec{p}) \nu_{t}(\vec{p})\big), 
\eeq
 where the jump rates  $\mcj(\vec{p}+\vec{q},\vec{p})$ are written in terms of the  differential cross section $\frac{d\sigma}{d\Omega}(\vec{p},\vec{p}_{0}):=|f(\vec{p},\vec{p}_{0})|^{2}$ in the form
\begin{multline}\label{JumpRates}
 \mcj(\vec{p}+\vec{q},\vec{p})= \eta \frac{m}{m_{*}^{2}}\frac{1}{|\vec{q}|}\int_{(\vec{q})_{\perp}}d\vec{v}\, \frac{d\sigma}{d\Omega} \Big(\frac{m_{*}}{m}\vec{v}-\frac{m_{*}}{M}\vec{p}_{\perp \vec{q}}-2^{-1}\vec{q} , \,    \frac{m_{*}}{m}\vec{v}-\frac{m_{*}}{M}\vec{p}_{\perp \vec{q}}+2^{-1}\vec{q}\Big)\\ \times r\big(\vec{v}+2^{-1}\frac{m}{m_{*}}\vec{q}+\frac{m}{M}\vec{p}_{\parallel \vec{q}  }    \big).       
\end{multline}
 
For the differential cross section $\frac{d\sigma}{d\Omega}$ corresponding to a hard-sphere interaction, the solutions $\nu_{t}$ to~(\ref{MomDist}) converge exponentially fast in $L^{1}(\R^{3})$ to the stationary state $\nu_{\infty}(\vec{p})=(2\pi M \beta^{-1})^{-\frac{1}{2}}e^{-\frac{\beta}{2M}\vec{p}^{2}} $.  This puts the dynamics~(\ref{MomDist}) in the ``opposite" category as those in which explosion occurs, since a particle starting with distribution $\nu_{0}(\cdot)=\delta(\cdot-\vec{p}) $ for arbitrarily high momentum $|\vec{p}|\gg M^{\frac{1}{2}}\beta^{-\frac{1}{2} }$ must have distribution $\nu_{t}$ concentrated near a region of radius $M^{\frac{1}{2}}\beta^{-\frac{1}{2} }$ from the origin after a finite time period.   This is not surprising, since it holds for the classical linear Boltzmann dynamics modeling a test particle in a  hard-sphere gas for which the marginal momentum distribution satisfies~(\ref{MomDist}) with $\frac{d\sigma}{d\Omega}(\vec{p}_{2},\vec{p}_{1})=4^{-1}a^{2}$ for spheres with radius $a$.  I return to a discussion of the existence and uniqueness of the semigroup~$\Phi_{t}$ in the next section.  
    
\subsubsection{Hard-sphere scattering   }   
    
Beyond the values $m$, $M$, $\eta$, $\beta$, the critical input determining  the noise term $\Psi$ is the scattering amplitude $f(\vec{p}_{1},\vec{p}_{2})=\mathbf{f}(\mathbf{p},\,\theta)$.  The scattering amplitude carries the essential two-body scattering information between the test particle and a single reservoir particle that  becomes dominant in the low-density regime from which the linear Boltzmann dynamics arises (see~\cite{ReedIII} for a mathematical discussion of the scattering amplitude).  The Hamiltonian for the two-body interaction is   
\begin{align} \label{PertInt}
H_{\textup{tot}} &=-\frac{\hbar^{2}}{2M}\Delta_{\textup{test}}-\frac{\hbar^{2}}{2m}\Delta_{\textup{res}}+V(x_{\textup{test}}-x_{\textup{res}})\nonumber \\ &= -\frac{\hbar^{2}}{2(m+M)}\Delta_{\textup{cm}}-\frac{\hbar^{2}}{2m_{*}}\Delta_{\text{rel}}+V(\vec{x}_{\textup{rel}}), 
\end{align}
where $V(\vec{x}_{test}-\vec{x}_{res})$ is the interaction potential,  and $\Delta_{\textup{cm}}$ and $\Delta_{\text{rel}}$ are the Laplacians with respect to the center-of-mass coordinates $\vec{x}_{\textup{cm}}=\frac{m_{*}}{m}\vec{x}_{\textup{test}}+\frac{m_{*}}{M}\vec{x}_{\textup{res}}   $ and $\vec{x}_{\textup{rel}}= \vec{x}_{\textup{res}}-\vec{x}_{\textup{test}}$. 

   A hard-sphere potential of radius $a$ is formally given by a  potential 
\begin{align*}
V(\vec{x})= \left\{  \begin{array}{cc} \infty &  |\vec{x}|\leq a,    \\  \quad & \quad \\ 0    &  |\vec{x}|>a  .  \end{array} \right.  
\end{align*}
Mathematically, a hard-sphere interaction is not defined as a perturbation of the $V=0$ case of~(\ref{PertInt}), but as a certain self-adjoint extension of the Laplacian corresponding to the kinetic energy. For the region $R=\{|\vec{x}_{text}-\vec{x}_{res}|\geq a   \}= \{|\vec{x}_{rel}|\geq a \}\subset \R^{6}$,  the Hamiltonian for the test particle and a single reservoir particle is   
\beq\label{NonPertInt}
H_{\textup{tot}}=-\frac{\hbar^{2}}{2M}\Delta_{\textup{test}}-\frac{\hbar^{2}}{2m}\Delta_{\textup{res}}= -\frac{\hbar^{2}}{2(m+M)}\Delta_{\textup{cm}}-\frac{\hbar^{2}}{2m_{*}}\Delta_{\text{rel}} 
\eeq
with Dirichlet boundary conditions on $\partial R$.

Since the Hamiltonian~(\ref{NonPertInt}) is not a perturbation of a ``free" Hamiltonian as in~(\ref{PertInt}), the standard scattering formalism in~\cite{ReedIII} does not apply directly for the hard-sphere case.  However, the scattering amplitude for a hard-sphere can be defined using the limits of the scattering amplitude for a  permeable spherical potential $V(\vec{x})=\bar{V}\chi(|\vec{x}|\leq a)$, $\bar{V}>0$.  For the permeable sphere, the scattering amplitude has a partial-wave expansion 
\beq \label{PartWave}
\mathbf{f}(\mathbf{p},\theta)=\frac{\hbar}{2\ii\mathbf{p}}\sum_{\ell =0}^{\infty}(2\ell+1)\big(S_{\ell}\big(\frac{a}{\hbar}\mathbf{p}\big)-1\big)\frak{L}_{\ell}(\cos(\theta)),  
\eeq
where the values $ S_{\ell}(\kappa)\in \C$ are determined by the  equation below and $\frak{L}_{\ell}$ is the $\ell$th Legendre polynomial.  The value $S_{\ell}(\kappa)$ is a phase factor with
$$\hspace{4cm} S_{\ell}(\kappa)=-\frac{h_{\ell}^{(2)}(\kappa)}{h_{\ell}^{(1)}(\kappa)  }\,\frac{\log^{\prime}h_{\ell}^{(2)}(\kappa)-\frak{n}\log^{\prime}j_{\ell}(\frak{n}\kappa)   }{ \log^{\prime}h_{\ell}^{(1)}(\kappa)-\frak{n}\log^{\prime}j_{\ell}(\frak{n}\kappa)   },  \quad \quad \hspace{.8cm} \kappa \in \R_{+},  $$
where $\frak{n}=\big(1-2m_{*}\bar{V}a^{2} \kappa^{-2}\hbar^{-2} \big)^{\frac{1}{2}}$, $\log^{\prime}$ refers to the logarithmic derivative, and $j_{\ell}$,  $h_{\ell}^{(1)}$, and $h_{\ell}^{(2)}$ are the spherical Bessel functions and Hankel functions of the first and second kind respectively.  In the  hard-sphere limit $\bar{V}\rightarrow \infty$, the above becomes  
$$
\hspace{6.5cm}S_{\ell}(\kappa)= -\frac{h_{\ell}^{(2)}(\kappa)}{h_{\ell}^{(1)}(\kappa)  }, \quad \quad \hspace{3cm} \kappa \in \R_{+}.
$$

   Even for an interaction as simple as the hard-sphere, there is no simple formula for $\mathbf{f}(\mathbf{p},\,\theta)$ except  when $\frac{a}{\hbar}\mathbf{p}\ll 1$ and the much more challenging regime $\frac{a}{\hbar}\mathbf{p}\gg  1$.  When $\frac{a}{\hbar}\mathbf{p}\ll 1$, then 
  $$ \mathbf{f}(\mathbf{p},\theta)\approx -a   .   $$
In the regime  $\frac{a}{\hbar}\mathbf{p}\gg 1$, there is a rough approximation~\cite{Nussen2} given by
\beq \label{HighEnergy}
 \mathbf{f}(\mathbf{p},\theta)\approx -\frac{1}{2}a\Big(e^{-2\ii\frac{a}{\hbar} \mathbf{p}\sin(\frac{\theta}{2})}+\textup{i} \frac{1+\cos(\theta)}{\sin(\theta)}J_{1}\big(\frac{a}{\hbar}\mathbf{p}   \sin(\theta)\big)        \Big),  
 \eeq
 where $J_{1}$ is a cylindrical Bessel function of the first kind.  The left term $e^{-2\ii\frac{a}{\hbar} \mathbf{p}\sin(\frac{\theta}{2})}$ is dominant except in a small region in the forward scattering direction $\theta\approx 0$.  The second term in~(\ref{HighEnergy}) becomes dominant for $\theta$ on the order of $(\frac{a}{\hbar} \mathbf{p})^{-1}\ll 1$ and yields a diffraction peak of height proportional to $a( \frac{a}{\hbar} \mathbf{p}) \gg a$ and a series of much smaller peaks forming the classical Airy pattern, which vanish for  $\theta\gg (\frac{a}{\hbar} \mathbf{p})^{-1}$.  The total cross section $\sigma_{\textup{tot}}(\mathbf{p})=\int_{\Omega}|\mathbf{f}(\mathbf{p},\theta)|^{2}   $,
 which characterizes the effective target area with respect to scattering, has the limiting values  $\sigma_{\textup{tot}}(\mathbf{p})\approx 4\pi a^2$ and $\sigma_{\textup{tot}}(\mathbf{p})\approx 2\pi a^{2}$ for 
 $\frac{a}{\hbar}\mathbf{p}\ll 1$ and  $\frac{a}{\hbar}\mathbf{p}\gg  1$, respectively.  The fact that $\sigma_{\textup{tot}}(\mathbf{p})$ limits to twice the classical cross-sectional area $\pi a^{2}$ in the large $\mathbf{p}$ limit is known as the {\it extinction paradox}, since  classical behavior would be expected in that regime (i.e. yielding the cross-sectional area $\pi a^{2}$).    
 
The study of $\mathbf{f}(\mathbf{p},\theta)$ through the partial-wave expansion when $\frac{a}{\hbar}\mathbf{p}\gg  1$ is well-known to be difficult, since the number of non-negligible terms to be summed in~(\ref{PartWave}) grows on the order of $\frac{a}{\hbar}\mathbf{p}$.  There is a successful approach based on exchanging the sum~(\ref{PartWave}) for a complex integration over a contour for an analytic extension of the index $\ell$: 
\beq \label{PartWaveInt}
\mathbf{f}(\mathbf{p},\theta)=\frac{\hbar}{2i\mathbf{p}}\int_{\textup{C}}\frac{d\lambda\,\lambda}{\cos(\pi \lambda)}\,\frac{h_{\lambda}^{(2)}(\frac{a}{\hbar}\mathbf{p} )}{h_{\lambda}^{(1)}(\frac{a}{\hbar}{\mathbf{p}} )  }   \frak{L}_{\lambda-\frac{1}{2}}\big(\cos(\theta)   \big).  
\eeq
 The contour $\textup{C}$ is is $u$-shaped around the positive real axis, although it must avoid enclosing any poles of $h_{\lambda}^{(1)}(\frac{a}{\hbar}{\mathbf{p}})$ as a complex function of $\lambda$.  The rewriting of the sum for the complex integration is called \textit{Watson's transformation}~\cite{Watson}.  The complex integration formulation allows the application of the method of stationary phase and the calculus of residues in the evaluation of~(\ref{PartWaveInt}).  This method was introduced by Watson and further developed in~\cite{Fock,Franz}.   Mathematical detail can be found in~\cite{Nussen,Nussen2}, and~\cite{Grandy} is an introductory book on the subject of scattering of waves from spherical objects.  In addition to the quantum hard-sphere, the same formalism~(\ref{PartWave}) applies for acoustical scattering and electromagnetic scattering from spherical targets.  For the purpose of the current article, I essentially only require a small corollary of the above results, which is  that the differential cross section $\frac{d\sigma}{d\Omega}(\mathbf{p},\theta)= |\mathbf{f}(\mathbf{p},\theta)|^{2}$ does not vanish away from the region of the forward direction $\theta\approx 0$ as $\mathbf{p}\rightarrow \infty$.

 \subsubsection{The diffusion constant }

The diffusion constant $D$ in Thm.~\ref{MainThm} is a sum of two parts $D_{\textup{kin}},   D_{\textup{jps} }\geq 0$ that can be interpreted as driven by an averaged kinetic motion of the test particle and by spatial jumps made by the test particle.  Let $\frak{p}_{t}\in \R^{3}$ be the classical Markov process with jump rates given by~(\ref{JumpRates})  and starting in the stationary distribution
 $\nu_{\infty}$.  The constants $D_{\textup{kin}}$ and $D_{\textup{jps}}$ have the form
\begin{eqnarray*}  
D_{\textup{kin}}&=&3^{-1}\int_{0}^{\infty}dt\,\mathbb{E}_{\nu_{\infty}} \big[ \frak{v}(\frak{p}_{t})\cdot \frak{v}(\frak{p}_{0}) \big], \\
D_{\textup{jps} }&=& 3^{-1}\hbar^{2}\int_{\R^{3}}d\vec{p}\,\nu_{\infty}(\vec{p})\,\int_{\R^{3}}d\vec{q} \sum_{n=1}^{3}  (\partial_{1,n}\partial_{2,n}   \textup{T}_{\vec{q}})\big(\vec{p},\vec{p}  \big),   \\ 
\end{eqnarray*}
where $\partial_{1,n},\partial_{2,n}$, $n=1,2,3$ are the partial derivatives with respect to the first and second arguments of $\textup{T}_{\vec{q}}\big(\vec{p}_{1},\vec{p}_{2}  \big)$ respectively, and $\frak{v}:\R^{3}\rightarrow \R^{3}$  has the form
\begin{align}\label{VelocityFunction}
 \frak{v}(\vec{p})= &  (\nabla H)(\vec{p})+\frak{u}(\vec{p}),\quad  \text{for}  \\ \frak{u}(\vec{p})= &\frac{2 \eta}{M}\int_{\R^{3}}d\vec{p}_{\textup{rel}}\,|\vec{p}_{\textup{rel}}|\,r\big(\frac{m}{m_{*}} \vec{p}_{\textup{rel}}+\frac{m}{M}\vec{p}   \big)   \int_{\Omega}\frac{ \vec{p}_{\textup{rel}}+|\vec{p}_{\textup{rel}}|\widehat{\theta}    }{\big|\vec{p}_{\textup{rel}}+|\vec{p}_{\textup{rel}}|\widehat{\theta}  \big| } \,\textup{Im}\Big[ (\partial_{\mathbf{z}}\mathbf{f})\big( \vec{p}_{\textup{rel}},\theta) \overline{\mathbf{f}}\big( \vec{p}_{\textup{rel}},\theta)  \Big]\nonumber . 
 \end{align}
The integration $\int_{\Omega}$ is over unit vectors $\widehat{\theta}\in \R^{3}$, $\theta$ is the angle between  $\vec{p}_{\textup{rel}} $ and $\widehat{\theta}$, and $\partial_{\mathbf{z}}$ is the mixed derivative  $\partial_{\mathbf{z}}=\cos(2^{-1}\theta )\partial_{\mathbf{p}}-2\mathbf{p}^{-1}\sin(2^{-1}\theta)\partial_{\theta}$.


It is worth formulating classical analogies of the dynamics to  better understand  the contributions to the diffusion constant.   Consider a classical dynamics in which the spatial variable $X_{t}\in \R^{3}$ for the particle  is driven ballistically by the momentum process $\frak{p}_{t}$ through the dispersion relation $(\nabla H)\big(\vec{p}\big)$:
\begin{align*}
X_{t}=X_{0}+\int_{0}^{t}dr\,(\nabla H)\big(\frak{p}_{r}\big).
\end{align*}
The Markov process  $\frak{p}_{t}$ is time-reversible, since the jump rates~(\ref{JumpRates}) satisfy the  detailed balance property that the kernel $A(\vec{p}_{1},\vec{p}_{2}):= \mcj(\vec{p}_{2},\vec{p}_{1})e^{\frac{\beta}{4M}(\vec{p}_{2}^{2}-\vec{p}_{1}^{2}  ) } $ is symmetric in $\vec{p}_{1}$ and $\vec{p}_{2}$.  This follows, since
\begin{multline}\label{JumpRatesAgain}
A(\vec{p}_{1},\vec{p}_{2}) = \eta \big(\frac{\beta}{2\pi m}    \big)^{\frac{3}{2}}\frac{m}{m_{*}^{2}}\frac{1}{|\vec{p}_{2}-\vec{p}_{1}|}e^{-\frac{\beta}{8m}(\vec{p}_{1}-\vec{p}_{2})^{2}          }\\  \times\int_{(\vec{p}_{2}-\vec{p}_{1} )_{\perp}}d\vec{v}\, \frac{d\sigma}{d\Omega} \Big(\frac{m_{*}}{m}\vec{v}-2^{-1}(\vec{p}_{2}-\vec{p}_{1}) , \,    \frac{m_{*}}{m}\vec{v}+2^{-1}(\vec{p}_{2}-\vec{p}_{1})\Big) e^{-\frac{\beta}{2m}(\vec{v}+\frac{m}{2M}(\vec{p}_{1}+\vec{p}_{2}))^{2}         }       
\end{multline}
and  $\frac{d\sigma}{d\Omega}(\vec{p}_{1},\vec{p}_{2}) $ is a function of the norm $|\vec{p}_{1}|=|\vec{p}_{2}|$ and the angle between the vectors $\vec{p}_{1}$ and $\vec{p}_{2}$.  
Central limit theorems for integral functionals of time-reversible Markov processes were treated in~\cite{Kipnis}. A more general discussion  on the subject of integral functionals of Markov processes can be found in the recent book~\cite{Landim}.  This theory would suffice to show  the renormalized processes $N^{-\frac{1}{2}}X_{Nt}$ converge in law as $N\rightarrow \infty$ to a Brownian motion $\mathbf{B}_{t}$ in $\R^{3}$ with diffusion constant $D'I_{s}$ for $D'>0$ with the Green-Kubo form
$$  D'=3^{-1}\int_{0}^{\infty}dt\,\mathbb{E}_{\nu_{\infty}} \big[ (\nabla H)(\frak{p}_{t})\cdot (\nabla H)(\frak{p}_{0}) \big] .  $$

Our expression for the diffusion constant $D$ differs from the  expression $D'$ by the substitution of $(\nabla H)\big(\vec{p}\big)$ with $\frak{v}(\vec{p})$ and the additional contribution $D_{\textup{jps}}$.  Technically, the quantum dynamics is more closely related to a classical Markov process that moves with velocity $(\nabla H)\big(\vec{p}\big)$ and makes joint spatial/momentum jumps with Poisson rates that depend on the  current state of the momentum (i.e. translation invariant rates).  Suppose jumps $(\vec{x},\vec{q})\in \R^{3}\times \R^{3}$ occur with Poisson rate density $ \textup{R}_{\vec{p}}(\vec{x},\vec{q}) $ when the current momentum is $\vec{p}\in \R^{3}$, and the following integration formulas hold:
$$ \mcj (\vec{p}+\vec{q},\vec{p})=\int_{\R^{3}}d\vec{x}\,\textup{R}_{\vec{p}}(\vec{x},\vec{q}) \hspace{1cm}\text{and}\hspace{1cm} \frak{u}(\vec{p})=\int_{\R^{3}\times \R^{3}}d\vec{x}\,d\vec{q}\, \textup{R}_{\vec{p}}(\vec{x},\vec{q})\,\vec{x} .    $$
In other words, the marginal jump rates in momentum agree with the process $\frak{p}_{t}$ and the averaged velocity generated by position jumps is $\frak{u}(\vec{p})$.  Suppose also  the rates $ \textup{R}_{\vec{p}}(\vec{x},\vec{q}) $ satisfy rotation invariance  $ \textup{R}_{\vec{p}}(\vec{x},\vec{q})=\textup{R}_{S\vec{p}}(S\vec{x},S\vec{q}) $ for $S\in \textup{SO}_{3}$ and the detailed balance condition 
\begin{align}\label{ReDetailBalance}
\nu_{\infty}(\vec{p}) \textup{R}_{\vec{p}}(\vec{x},\vec{q})=\nu_{\infty}(\vec{p}+\vec{q}) \textup{R}_{\vec{p}+\vec{q} }(-\vec{x},-\vec{q} ).
\end{align}
Then, under some assumptions of the existence of moments for the spatial jumps, the diffusion matrix is  $D''I_{3}$ for 
$$D''= 3^{-1}\int_{0}^{\infty}dt\,\mathbb{E}_{\nu_{\infty}} \big[ \frak{v}(\frak{p}_{t})\cdot \frak{v}(\frak{p}_{0}) \big]  + 3^{-1}\int_{\R^{3}}d\vec{p}\,\nu_{\infty}(\vec{p})\,\int_{\R^{3}\times \R^{3}}d\vec{x}\,d\vec{q}\,\textup{R}_{\vec{p}}(\vec{x},\vec{q})\,|\vec{x}|^{2} .     $$
The detailed balance condition~(\ref{ReDetailBalance}) eliminates an extra term from appearing in the expression for $D''$, and  there is an analogous condition satisfied by the quantum dynamics considered in this article, which is discussed in Sect.~\ref{RemarkDetailedBalance}.  

Some similarity between $D_{\textup{jps} }$ and the second term in the expression for $D''$ can be seen by a comparison of $\textup{R}_{\vec{p}}(\vec{x},\vec{q})$ with a Wigner transform of $\textup{M}_{\vec{q}}\big( \vec{p}_{1},\vec{p}_{2} \big)$:
 $$\textup{R}_{\vec{p}}'(\vec{x},\vec{q}):=(2\pi)^{-3}\int_{\R^{3}}d\vec{k}\,e^{\textup{i} \vec{x}\vec{k}}   \textup{M}_{\vec{q}}\big( \vec{p}-2^{-1}\hbar \vec{k} ,\vec{p}+2^{-1}\hbar \vec{k} \big).  $$
 Although the classical analogy is not perfect, it gives some intuition for both the symmetries and complications in the analysis of the model.

In the limit of large mass $M$ for the test particle, the $D_{\textup{kin}}$ component of the diffusion constant dominates.  Let  $\lambda=\frac{m}{M}$, then as $M\rightarrow \infty$   
\begin{align*} D_{\textup{jps} }=  & \frac{4 \eta \hbar^{2}m_{*}}{3M^{2}}\big(\frac{2m_{*} }{\pi\beta }   \big)^{\frac{1}{2}}\int_{\R_{+}}d\mathbf{q} \,\mathbf{q}^{3}    e^{- \mathbf{q}^{2}}   \Big(  \frac{\beta }{4m_{*}}\big( 2\mathbf{q}^{2}+3\frac{m}{M} \big)\, \sigma_{\textup{tot}}\big((\frac{2m_{*} }{\beta} )^{\frac{1}{2}}\mathbf{q} \big)+  \sigma_{\mathbf{z}}\big((\frac{2m_{*} }{\beta} )^{\frac{1}{2}}\mathbf{q}\big)    \Big) \\
= &
 \lambda^{2}\frac{4 \eta \hbar^{2}}{3m}\big(\frac{2m }{\pi\beta }   \big)^{\frac{1}{2}}\int_{\R_{+}}d\mathbf{q} \,\mathbf{q}^{3}    e^{- \mathbf{q}^{2}}   \Big(  \frac{\beta }{2m} \mathbf{q}^{2}\, \sigma_{\textup{tot}}\big((\frac{2m }{\beta} )^{\frac{1}{2}}\mathbf{q} \big)+  \sigma_{\mathbf{z}}\big((\frac{2m }{\beta} )^{\frac{1}{2}}\mathbf{q}\big)    \Big)+\mathit{O}(\lambda^{3}),   
 \end{align*}
where  $\sigma_{\textup{tot}}(\mathbf{p})$ is the total cross section defined above and $\sigma_{\mathbf{z}}(\mathbf{p})= \int_{\Omega}|\partial_{\mathbf{z}}\mathbf{f}(\mathbf{p},\theta)|^{2} $.
In the regime $\lambda\ll 1$,  $D_{\textup{jps} }$ is small compared to the classical part $D_{\textup{kin}}$, which has order $\lambda$.

I remark that the leading term for $ D_{\textup{jps}}$ in the limit of small $\lambda$ is not the same as the \textit{quantum diffusion coefficient} $D_{\textup{qdc}}$ discussed in~\cite[Sec. 5  ]{VaccHorn}, which has the form
$$D_{\textup{qdc}}=\lambda^{2}\frac{\eta \beta \hbar^{2}}{6 m^{2}}   \big(\frac{2m  }{\pi\beta  }   \big)^{\frac{1}{2}}\int_{\R_{+}}d\mathbf{q} \,\mathbf{q}^{5}    e^{- \mathbf{q}^{2}}\, \sigma_{0}\big((\frac{2m }{\beta} )^{\frac{1}{2}}\mathbf{q} \big) ,  $$
where $\sigma_{0}(\mathbf{p})=\int_{\Omega}\big(1-\cos(\theta)\big)\frac{d \sigma}{d\Omega}(\mathbf{p},\theta)  $.  
The constant  $D_{\textup{qdc}}$ arises in a Brownian limit of the quantum linear Boltzmann dynamics~(\ref{MainDyn}).  The difference between the leading order of $D_{\textup{jps}}$ and $D_{\textup{qdc}}$ is a non-commutation between the diffusive and the Brownian limits.  Most notably,  the leading order for $D_{\textup{jps}}$ depends on the derivatives of the scattering amplitude, but $D_{\textup{qdc}}$ depends only on the values of the differential cross section $\frac{d \sigma}{d\Omega}(\mathbf{p},\theta)=|\mathbf{f}(\mathbf{p},\theta)|^{2}   $.

\subsubsection{Remarks on the proof strategy}

The basic strategy   I use for the proof of Thm.~\ref{MainThm} is a standard one, which is formulated using a decomposition arising from the translation invariance of the model.  The articles~\cite{Schenker,Diff,DF} are similarly based on a translation symmetry.   The translation invariance of the model implies a ``decomposition" of the state space into  fibers $\mathcal{B}_{1}\big(L^{2}(\R^{3})\big)\approx \int_{\R^{3} }^{\oplus}L^{1}(\R^{3})  $\symbolfootnote[2]{This equivalence is not meant to be precise.} that evolve independently according to the dynamics.  In other words, there are a family of maps $\rho\rightarrow [\rho]_{\vec{k}}  $ from $\mathcal{B}_{1}\big(L^{2}(\R^{3})\big)$ to $L^{1}(\R^{3})$ for $\vec{k}\in \R^{3}$ and a family of semigroups $ \widetilde{\Phi}_{t}^{(\vec{k})}:L^{1}(\R^{3})  $   such that
$$ [\Phi_{t}(\rho)]_{\vec{k}}=  \widetilde{\Phi}_{t}^{(\vec{k})}[\rho]_{\vec{k}} .  $$
The longterm diffusive behavior for the test particle is determined by information contained in the semigroups $ \widetilde{\Phi}_{t}^{(\vec{k})}$ for an infinitesimal neighborhood around $\vec{k}=0$.   In particular, if $\mathcal{L}_{\vec{k}}$ is the infinitesimal generator for $\widetilde{\Phi}_{t}^{(\vec{k})}$, then the diffusion matrix is determined by the second derivative at zero for the eigenvalue $\epsilon (\vec{k})\in \R_{-}+\textup{i}\,\R  $  of $\mathcal{L}_{\vec{k}}$ having leading real part in the spectrum: $D=-\nabla^{\otimes^2}\epsilon_{\vec{k}}|_{\vec{k}=0}$.

The analysis in the proof uses that the first few derivatives of $\mathcal{L}_{\vec{k}}$ in a neighborhood of $\vec{k}=0$ are relatively bounded to 
$\mathcal{L}_{0}$.  The contribution to $\mathcal{L}_{\vec{k}}$ coming from the kinetic Hamiltonian term $\frac{1}{2M}\vec{P}^{2}$ is a multiplication operator with function $g_{\vec{k}}(\vec{p}):= \frac{\textup{i}}{M}\vec{k}\cdot \vec{p}   $.   The choice of hard-sphere interaction for the model has the advantage that the Hamiltonian term is relatively bounded to $\mathcal{L}_{0}$.  Softer short-range interaction potentials will not have this property, and the argument would in the least have to be modified to cover other cases.

\section{Existence and conservativity of the dynamics  }\label{SecCons}

In this section, I provide some technical discussion regarding the construction of the semigroup $\Phi_{t}$.  The reader who is unconcerned by this issue may pass to the next section without any notations missed.  A theory specific to translationally invariant Lindblad dynamics was developed in~\cite{Holevo}.  In Appendix~\ref{FourFormal},  I present a Fourier transform relating the noise map $\Psi$ to the maps~$\mathbf{T}_{\vec{q}}$ in~(\ref{Jeez}).   See~\cite{Fagnola,Cov} and the review~\cite{Unbounded} for technical work regarding Lindblad equations with unbounded generators.     

 I define the semigroup $\Phi_{t}: \mathcal{B}_{1}\big(L^{2}(\R^{3})\big)$  as the pre-adjoint maps corresponding  to a semigroup of completely positive maps $\Phi_{t}^{*}$ acting on the collection of bounded operators $ \mathcal{B}\big(L^{2}(\R^{3})\big)$ over the Hilbert space $L^{2}(\R^{3})$ and satisfying the integral equation below:  
\beq \label{IntEq}
\langle \phi |\Phi_{t}^{*}(G)\, \psi\rangle = \langle \phi |G\, \psi\rangle+ \int_{0}^{t}dr\,  \mathcal{L}\big(\phi,\,\Phi_{r}^{*}(G)\, \psi \big), \quad \quad \phi,\psi \in \textup{D}(\vec{P}^{2}),\,\,G\in \mathcal{B}\big(L^{2}(\R^{3})\big),
\eeq
where the form generator $\mathcal{L}:\textup{D}(\vec{P}^{2})\times \mathcal{B}\big(L^{2}(\R^{3})\big)\times \textup{D}(\vec{P}^{2})\rightarrow \C$ is defined as
$$ \mathcal{L}(\phi ;\,G;\, \psi )=  \langle  A  \phi \big| G \psi\rangle  + \langle  \phi \big| G A \psi\rangle +\int_{\R^{3}}d\vec{q} \int_{ (\vec{q})_{\perp} }d\vec{v}\,  \big\langle e^{\ii\frac{\vec{q}}{\hbar}\vec{X}}L_{\vec{q},\vec{v}}\, \phi \big| G \, e^{\ii\frac{\vec{q}}{\hbar}\vec{X}} L_{\vec{q},\vec{v}}\psi \big\rangle,    $$
for $A=-\frac{i}{\hbar} H -2^{-1}\Psi^{*}(\textup{I})$.  The domain $\textup{D}(\vec{P}^{2})\subset L^{2}(\R^{3})$ is the collection of all vectors $\phi$ whose density in the momentum representation satisfies  $\int_{\R^{3}}d\vec{p}\,|\vec{p}|^{4}|\phi(\vec{p})|^{2}<\infty$.    The translation invariance of the dynamics is characterized in the form generator through  the covariance relation
\beq \label{CovGen}
\hspace{4.2cm} \mathcal{L}\big(\phi ;\,e^{-\ii\vec{b}\vec{P}}Ge^{\ii\vec{b}\vec{P}};\, \psi \big)=   \mathcal{L}\big(e^{\ii\vec{b}\vec{P}}\phi ;\,G;\,e^{\ii\vec{b}\vec{P}}\psi \big),  \hspace{1.8cm} \vec{b}\in \R^{3}.
\eeq
This relation follows since $A$ and $L_{\vec{q},\vec{v}}$ are functions of $\vec{P}$ and by Weyl intertwining relations between  $e^{\ii\vec{b}\vec{P}}$ and  $e^{\ii\frac{\vec{q}}{\hbar}\vec{X}}$.

A semigroup $\Phi_{t}:\mathcal{B}_{1}\big(L^{2}(\R^{3})\big)$ is said to be \textit{conservative} if it preserves trace.  This is equivalent to the adjoint  maps sending the identity operator to itself (i.e. $\Phi_{t}^{*}(\textup{I})=\textup{I}$), since
$$ \hspace{3cm}\Tr[ \Phi_{t}(\rho)]=\Tr[ \Phi_{t}^{*}(\textup{I})\rho]= \Tr[\rho], \hspace{2cm} \rho\in\mathcal{B}_{1}\big(L^{2}(\R^{3})\big).   $$
Conservativity is analogous to being \textit{stochastic} for a Markovian semigroup~\cite[Ch.X]{Feller}.
 The weak*-continuity for $\Phi_{t}^{*}$ stated in Thm.~\ref{ExistenceLemma} is equivalent to strong continuity for $\Phi_{t}$, by a general result for semigroups on Banach spaces~\cite[Cor.3.18]{Bratteli}.  The translation covariance property~(\ref{CovGroup}) implies the same covariance for $\Phi_{t}$.

 Denote $\mathcal{T}_{m}=\{y\in L^{1}(\R^{3})\,|\,\int_{\R^{3}}d\vec{p}|\vec{p}|^{m}|y(\vec{p})|<\infty\} $.  In the proof of Thm.~\ref{ExistenceLemma}, I use results from Sect.~\ref{SecMarMom} that the operator $\mathcal{L}_{0}$ (defined in~(\ref{MomDist})) is closed with domain $\mathcal{T}_{1}$  and generates an ergodic Markovian semigroup.  The proof of Thm.~\ref{ExistenceLemma} reiterates some of the reasoning from~\cite{Holevo} and concludes with an argument specific to my model verifying that $\mathcal{L}_{0}$ is equal to the generator of a certain Markovian semigroup defined through the quantum dynamics.

\begin{theorem}\label{ExistenceLemma}
There exists a unique  semigroup of completely positive maps $\Phi_{t}^{*}:\mathcal{B}\big(L^{2}(\R^{3})\big)$ satisfying the integral equation~(\ref{IntEq}).  The semigroup is weak*-continuous and conservative.  Finally, the maps $\Phi_{t}^{*}$ have the translation covariance property 
\begin{align}\label{CovGroup}
\hspace{4.2cm} \Phi_{t}^{*}\big(e^{-\ii\vec{b}\vec{P}}Ge^{\ii\vec{b}\vec{P}} \big)=   e^{-\ii\vec{b}\vec{P}} \Phi_{t}^{*}(G)e^{\ii\vec{b}\vec{P}}, \hspace{1.8cm} \vec{b}\in \R^{3},
\end{align}
for all $G\in \mathcal{B}\big(L^{2}(\R^{3})\big)$.

\end{theorem}

\begin{proof} By~\cite{Davies}, a weak*-continuous semigroup of completely positive maps $\frak{T}_{t}^{(min),*}$ satisfying~(\ref{IntEq}) can be constructed if $A$ generates a strongly continuous contraction semigroup on $L^{2}(\R^{3})$, and the domain $\textup{D}(A)$ of $A$ is contained in the domain of $L_{\vec{q},\vec{v}}$ for all index values $\vec{q},\vec{v}$.  In my case, $A$ is a perturbation of $ -\frac{\ii}{\hbar}\frac{1}{2M}\vec{P}^{2}$ by $-\frac{\ii}{\hbar}H_{f}-\frac{1}{2}\Psi^{*}(\textup{I})$.  The operators $\Psi^{*}(\textup{I})= \mathcal{E}(\vec{P})$ and $H_{f}$ are both functions of the vector of momentum operators and are relatively bounded to $|\vec{P}|$ (see Lemmas~\ref{Horses} and~\ref{Leftover}).  Thus, $A$ is m-accretive with domain $\textup{D}(\vec{P}^{2})$ and generates a      
contractive semigroup.  Finally, the operators $L_{\vec{q},\vec{v}}$ are relatively bounded to $|\vec{P}|$ by Lem.~\ref{Leftover} and thus include $\textup{D}(\vec{P}^{2})$ in their domains.  It follows by~\cite{Davies}, there exists a solution  
$\frak{T}_{t}^{(min),*}$ to~(\ref{IntEq}) such that for any second  weak*-continuous collection of completely positive maps  $\frak{T}_{t}(G)$ solving~(\ref{IntEq}) that
$$\hspace{5cm} \frak{T}_{t}^{(min),*}(G)\leq \frak{T}_{t}(G)\leq \|G\|\textup{I}, \quad  \hspace{.7cm} \quad G\in \mathcal{B}\big(L^{2}(\R^{3})\big),\quad \, G\geq 0.
   $$    
The semigroup $\frak{T}_{t}^{(min),*}$ is referred to as the {\it minimal solution}. From~\cite{Holevo}, the relation~(\ref{CovGen}) has the consequence that the minimal solution $\frak{T}_{t}^{(min),*}$ satisfies the translation covariance relation~(\ref{CovGroup}).

  With the construction of $\frak{T}_{t}^{(min),*}$, the next question is of its uniqueness as a solution to~(\ref{IntEq}).  By~\cite{Davies}, the uniqueness of the solution $\frak{T}_{t}^{(min),*}$ is implied by  $\frak{T}_{t}^{(min),*}(\textup{I})=\textup{I}$ for all $t\geq 0$ (i.e. $\frak{T}_{t}^{(min),*}$ is conservative). 
As observed in~\cite{Holevo}, the translational covariance of $\frak{T}_{t}^{(min),*}$ reduces the problem of checking conservativity for the Lindblad dynamics to checking stochasticity for a classical Kolmogorov semigroup.  I prove below that the Kolmogorov semigroup is $e^{t\mathcal{L}_{0}^{*}}$.   By Prop.~\ref{RelBnds},  $e^{t\mathcal{L}_{0}}$ is ergodic to the stationary state $(2\pi M\beta^{-1}   )^{-\frac{1}{2}}   e^{-\frac{\beta}{2M}\vec{p}^{2}}$, and it follows that  $e^{t\mathcal{L}_{0}^{*}}(1_{\R^{3}})=1_{\R^{3}}$.  This would entail that $\Phi_{t}^{*} :=\frak{T}_{t}^{(min),*}$ is the unique weak*-continuous solution to~(\ref{IntEq}).

 The algebra $\mathcal{A}=\{ g(\vec{P})\,\big|\,g\in L^{\infty}(\R^{3})\}$ is invariant under the action of $\frak{T}_{t}^{(min),*}$, since bounded functions of $\vec{P}$ are the only elements in $\mathcal{B}\big(L^{2}(\R^{3})\big)$ that commute with the Weyl operators $e^{\ii\vec{b}\vec{P}}$ for all $\vec{b}\in \R^{3}$.  Hence, I can define a semigroup of maps $S_{t}^{*}: L^{\infty}(\R^{3})$ through 
$g_{t}=S_{t}^{*}(g)$ for
$$\hspace{5cm}g_{t}(\vec{P}):=\frak{T}_{t}^{(min),*}\big(g(\vec{P})\big),\hspace{3cm}  g\in L^{\infty}(\R^{3}).$$  The semigroup $S_{t}^{*}$ is contractive and maps positive functions to positive functions. Moreover, $S_{t}^{*}(1_{\R^{3}})=1_{\R^{3}}$ if and only if $\frak{T}_{t}^{(min),*}(\textup{I})=\textup{I}$.  Let  $L^{*}$ be the  generator for $S_{t}^{*}$ and $L$ be the adjoint  generating $S_{t}$.  Denote the dense domain of $L$ by $\textup{D}(L)\subset L^{1}(\R^{3})$.  I will show that $L=\mathcal{L}_{0}$. 

 By differentiating~(\ref{IntEq}) at $t=0$ for $G= g(\vec{P})$ with  $g\in L^{\infty}(\R^{3}) $,
\begin{align}\label{LindToMast}
  \frac{d}{dt}\Big|_{t=0}\int_{\R^{3}}d\vec{p} &\,y(\vec{p}) S_{t}^{*}(g)(\vec{p})\nonumber  \\ &= \int_{\R^{3}}d\vec{q}\int_{\R^{3}}d\vec{p}\,y(\vec{p})\mcj (\vec{p}+\vec{q},\vec{p})\,g(\vec{p}+\vec{q})-\int_{\R^{3}}d\vec{p}\, y(\vec{p})\int_{\R^{3}}d\vec{q}\,\mcj (\vec{p}+\vec{q},\vec{p})\,g(\vec{p})\nonumber \\ &=\int_{\R^{3}}d\vec{p}\,\mathcal{L}_{0}(y)(\vec{p})g(\vec{p}), 
\end{align}
where $y(\vec{p})= \phi(\vec{p})\psi(\vec{p})$.  The set of possible $y(\vec{p})$ with $\phi, \psi\in \textup{D}(\vec{P}^{2})$ is $\mathcal{T}_{4}$.  As a consequence of Part (2) of Lem.~\ref{RelBnds}, the generator $\mathcal{L}_{0}$ and the multiplication operator $(\frak{M}y)(\vec{p})=  |\vec{p}|  y(\vec{p})$ (having domain $\mathcal{T}_{1}$) are relatively bounded to each other.  In other words, there is a $c>0$ such that for all $y\in \mathcal{T}_{1}$ the graph norms for $\mathcal{L}_{0}$ and $\frak{M}$ satisfy  
\beq c^{-1}\big(\|y\|_{1} +\|\frak{M}(y)\|_{1}       \big)\leq \|y\|_{1}+\|\mathcal{L}_{0}(y)\|_{1}         \leq c\big( \|y\|_{1} + \|\frak{M}(y)\|_{1} \big) .    \eeq
This implies  the right side of~(\ref{LindToMast}) is bounded for $y\in  \mathcal{T}_{1}$.  Moreover, any closed operator agreeing with $\mathcal{L}_{0}$ on a subset of $\mathcal{T}_{1}$ that is dense in $\mathcal{T}_{1}$ with respect to the graph norm of $\frak{M}$ must be equal to $\mathcal{L}_{0}$.  Continuing with~(\ref{LindToMast}),  for all $y\in \mathcal{T}_{4},\, g\in L^{\infty}(\R^{3})$    
\begin{align} \label{LindToMast2}
\frac{d}{dt}\Big|_{t=0}\int_{\R^{3}}d\vec{p}\, y(\vec{p}) S_{t}^{*}(g)(\vec{p})&=\frac{d}{dt}\Big|_{t=0}\int_{\R^{3}}d\vec{p}\, S_{t}(y)(\vec{p}) g(\vec{p})\nonumber  \\  &= \int_{\R^{3}}d\vec{p}\, L(y)(\vec{p}) g(\vec{p})
\end{align}
Hence, $L^{*}(y)$ is the weak limit of $t^{-1}\big( S_{t}^{*}(y)-y\big)$ as $t\rightarrow 0$ for all $y\in \mathcal{T}_{4}$.  Consequently, by the remark~\cite[IV.1.5]{Kato}, $\mathcal{T}_{4}$ must be contained in the domain $\textup{D}(L)$.  Moreover (\ref{LindToMast}) and~(\ref{LindToMast2}) imply $L$ agrees with $\mathcal{L}_{0}$ for elements in $\mathcal{T}_{4}$.   Since $L$ generates a contractive semigroup, it must be closed.  Finally, $\mathcal{T}_{4}$ is dense in $\mathcal{T}_{1}$ with respect to the graph norm of $\frak{M}$, and therefore $\mathcal{L}_{0}=L$ by the remark above.

\end{proof}

\section{The fiber decomposition and detailed balance}\label{SecFiberDec}

In this section, I discuss a few symmetries and decompositions for the dynamics.  The dynamics also satisfies a rotation covariance, which I do not explain.

\subsection{Fiber maps }
The Lindblad generator $\mathcal{L}$ satisfies translation covariance with respect to position shifts $\vec{a}\in \R^{3}$
\begin{align}\label{EllCov}   
 \mathcal{L}(e^{-\ii\frac{\vec{a}}{\hbar}\vec{P}}\rho e^{\ii\frac{\vec{a}}{\hbar}\vec{P}})= e^{-\ii\frac{\vec{a}}{\hbar}\vec{P}}\mathcal{L}(\rho) e^{\ii\frac{\vec{a}}{\hbar}\vec{P}}.
 \end{align}
  A consequence of this symmetry is that the dynamics admits a fiber decomposition in which individual fibers obey autonomous dynamics.  Similar fiber decompositions can, for instance, be found in the study of periodic Schr\"odinger potentials~\cite[Ch.XIII.16]{ReedIV} and in the study of certain translation invariant noisy dynamics for a quantum particle living on $\Z^{d}$,~\cite{Schenker,Diff,DF}.  Also, the autonomous Markovian  evolution for the momentum distribution~(\ref{MomDist}) is a special case of the fiber decomposition and provides a useful tool for proving  a solution of an unbounded Lindblad equation is unique when the Lindbladian is translation-covariant~\cite{Holevo} (as discussed in the previous section). In this and future sections, I will denote  $\mathcal{T}_{n}=\{y\in L^{1}(\R^{3})\,\big|\,  \int_{\R^{3}} d\vec{p}\,|\vec{p}|^{n}|y(\vec{p})| <\infty \}$ and use the following notations for norms 
\begin{itemize}
\item    $\|y\|_{p}$ is the standard   $L^{p}(\R^{3})$ norm for $p\in[1,\infty]$.    

\item $\|\rho\|_{\mathbf{1}}:=\textup{Tr}[|\rho|]    $ is the trace norm  for elements in $\mathcal{B}_{1}\big(L^{2}(\R^{3})\big)$.

\item $\|A\|$ is the operator norm for elements $A\in\mathcal{B}\big(L^{p}(\R^{3})\big)$.  

\item  $\|y\|_{A}$ is the graph norm for a densely defined operator $A$ on $L^{1}(\R^{3})$: 
$$ \|y\|_{A}:=\|y\|_{1}+\|A y\|_{1}.    $$
\end{itemize}

Now, I define the fiber maps.  For $\vec{k}\in \R^{3}$, there is a contractive map $\rho\rightarrow [\rho]_{\vec{k}}$ that sends an element 
$\rho\in \mathcal{B}_{1}\big(L^{2}(\R^{3})\big)$ to an element $ L^{1}(\R^{3})$ satisfying the formal relation
 \begin{align}\label{DeftOnes}
 [\rho]_{\vec{k}}(\vec{p})=\rho(\vec{p}-2^{-1}\hbar\vec{k},\,\vec{p}+2^{-1}\hbar\vec{k}),
 \end{align} 
where the kernel of $\rho$ on the right side is in the momentum representation. A rigorous definition of $[\rho]_{\vec{k}}$ requires something more than the integral kernel of $\rho$, which is only defined a.e. $\R^{3}\times \R^{3}$.  However, I can define the contractions $[\cdot]_{\vec{k}}:\mathcal{B}_{1}\big(L^{2}(\R^{3})\big)\rightarrow L^{1}(\R^{3})$ through the Riesz representation theorem.  Let $g\in L^{\infty}(\R^{3})$, then  by standard properties of the trace
\begin{align} \label{BndFib}
\big|\Tr\big[ e^{\ii\vec{k}\vec{X}} g(\vec{P}+2^{-1}\hbar\vec{k})\rho\big]\Big| &\leq \|g(\vec{P})\|\,\|\rho\|_{\mathbf{1}}\nonumber \\ &=\|g\|_{\infty}\|\rho\|_{\mathbf{1}}.\end{align}
By the Riesz representation theorem, there is a unique complex-valued measure $\mu_{\vec{k},\rho}$ on $\R^{3}$ with  total variation $\leq \|\rho\|_{\mathbf{1}}$  and satisfying the first equality below for all $g \in  L^{\infty}(\R^{3})$.
\begin{align} \label{Reisz}
\Tr[e^{\ii\vec{k}\vec{X}} g(\vec{P}+2^{-1}\hbar\vec{k})\rho]= & \int_{\R^{3}}\mu_{\vec{k},\rho}(d\vec{p})\,g(\vec{p})\nonumber \\ = & \int_{\R^{3}}d\vec{p}\,[\rho]_{\vec{k}}(\vec{p})\,g(\vec{p})
\end{align}
  The continuous spectrum of the family $P_{j}$, $j=1,2,3$ implies the measure $\mu_{\vec{k},\rho}$ must be absolutely continuous with respect to Lebesgue measure, and I denote its Radon-Nikodym derivative by $[\rho]_{\vec{k}}\in L^{1}(\R^{3})$ in~(\ref{Reisz}).  It is clear from~(\ref{DeftOnes})  the functions $[\rho]_{\vec{k}}$, $\vec{k}\in \R^{3}$ also determine the density matrix $\rho$.

The fiber decomposition for the dynamics $\Phi_{t}$ can be characterized by the existence of semigroups $\widetilde{\Phi}_{t}^{(\vec{k})}:L^{1}(\R^{3})$ such that 
\beq \label{ForFib}
[\Phi_{t}(\rho)]_{\vec{k}}=\widetilde{\Phi}_{t}^{(\vec{k})}[\rho]_{\vec{k}}.   
\eeq
The relation~(\ref{ForFib}) is equivalent  to noticing that $\Phi_{t}^{*}$ maps the Banach space
 $B_{\vec{k}}=\{e^{\ii\vec{k}\vec{X}  }g(\vec{P}),\, g\in L^{\infty}(\R^{3}) \}\subset \mathcal{B}_{\infty}\big(L^{2}(\R^{3})\big) $ into itself for each $\vec{k}$ with   
 \beq  \label{SailorMoon} \Phi_{t}^{*}\big( e^{\ii\vec{k}\vec{X}  }g(\vec{P}+2^{-1}\hbar\vec{k} ) \big)=   \widetilde{\Phi}_{t}^{(\vec{k}),*}(g)(\vec{P}+2^{-1}\hbar\vec{k}).    
 \eeq 
where $  \widetilde{\Phi}_{t}^{(\vec{k}),*}:L^{\infty}(\R^{3})$ is the adjoint semigroup of $\widetilde{\Phi}_{t}^{(\vec{k})}$.
 
 \begin{lemma}\label{TrivialContract}
For all $\vec{k}\in \R^{3}$, the semigroups $\widetilde{\Phi}_{t}^{(\vec{k})}:L^{1}(\R^{3})$ are strongly continuous and contractive.  Consequently, the generator   $\mathcal{L}_{\vec{k}}$ of the semigroup $\widetilde{\Phi}_{t}^{(\vec{k})}$ is a closed operator with dense domain $\textup{D}( \mathcal{L}_{\vec{k}})\subset L^{1}(\R^{3})$, and all $z\in \mathbb{C}$ with $\textup{Re}(z)>0$ belong to the resolvent set of $\mathcal{L}_{\vec{k}}$ with
\begin{align}\label{ResProp}
\big\|\big( \mathcal{L}_{\vec{k}} -z    \big)^{-1} \big\|\leq \frac{1}{\textup{Re}(z)   }. 
\end{align}
 
\end{lemma} 
 \begin{proof}
 It is equivalent to establish the contraction property for the adjoint semigroup $\widetilde{\Phi}_{t}^{(\vec{k}),*}$.  The supremum norm of an element $f\in L^{\infty}(\R)$ is equal to the operator norm of $f(\vec{P})\in \mathcal{B}\big(L^{2}(\R^{3})\big)$, so I have the first equality below for  $ g\in L^{\infty}(\R) $:
\begin{align*}
\|\widetilde{\Phi}_{t}^{(\vec{k}),*}g \|_{\infty}&=\|\widetilde{\Phi}_{t}^{(\vec{k}),*}(g)(\vec{P}) \|\\ &= \| \Phi_{t}^{*}\big(e^{\ii\vec{k}\vec{X}  }g(\vec{P}+2^{-1}\hbar \vec{k})\big)\|\\  &\leq  \|g(\vec{P}+2^{-1}\hbar \vec{k})\|   =\|g\|_{\infty}.  
\end{align*}
The second equality above is by the defining relation~(\ref{SailorMoon}), and the inequality is due to $\Phi_{t}^{*}$ being a contraction.  By similar reasoning, the weak*-continuity of $\Phi_{t}^{*}$ implies weak*-continuity for $\widetilde{\Phi}_{t}^{(\vec{k}),*}$, and consequently the pre-adjoint semigroup $\widetilde{\Phi}_{t}^{(\vec{k})}$ is strongly continuous.

The consequences for the generator $\mathcal{L}_{\vec{k}} $ and its resolvents are standard and can be found in~\cite[Sec.IX.2-3]{Kato} or~\cite[Sec.XIII.9-10]{Feller}.

\end{proof}

By the same reasoning as for $\mathcal{L}_{0}$ in the proof of Thm.~\ref{ExistenceLemma}, the domain $\textup{D}(\mathcal{L}_{\vec{k}})$ must include $\mathcal{T}_{4}$. An implication of Prop.~\ref{RelBnds} is that the domain contains $\mathcal{T}_{1}$ for all $\vec{k}$.    The operator  $\mathcal{L}_{\vec{k}}$  is a sum of terms
\beq \label{FiberGen0}
\mathcal{L}_{\vec{k}}= -\frac{\ii}{\hbar}h_{\vec{k}}+        \mathcal{J}_{\vec{k}}-\mathcal{E}_{\vec{k}}
\eeq
where the terms $h_{\vec{k}},\,\mathcal{J}_{\vec{k}},\,\mathcal{E}_{\vec{k}}:\mathcal{T}_{1}\rightarrow L^{2}(\R^{3})$ arise from the Hamiltonian, and  gain and loss parts of the Lindblad generator~(\ref{MainDyn}), respectively.  The terms above are determined analogously to $\widetilde{\Phi}_{t}^{(\vec{k})}$ in~(\ref{ForFib}) such that for $\rho$ in a dense domain of $\mathcal{B}_{1}\big( L^{2}(\R^{3})   \big)$,
$$
 h_{\vec{k}}[\rho]_{\vec{k}}= \big[  [H,\rho]\big]_{\vec{k}}, \quad  \quad \mathcal{J}_{\vec{k}}[\rho]_{\vec{k}}= \big[  \Psi(\rho)  \big]_{\vec{k}},\quad \quad \mathcal{E}_{\vec{k}}[\rho]_{\vec{k}}=\big[ 2^{-1}\{\Psi^{*}(\textup{I}),\rho\}  \big]_{\vec{k}}.
 $$
   Formally, I have 
\begin{eqnarray}
  h_{\vec{k}}(f)(\vec{p}) &=& \big(H(\vec{p}-2^{-1}\hbar\vec{k})- H(\vec{p}+2^{-1}\hbar\vec{k})       \big)f(\vec{p}) \nonumber \\  \mathcal{J}_{\vec{k}}(f)(\vec{p})&=&\int_{\R^{3}}d\vec{p}_{0}\,\mcj_{\vec{k}}(\vec{p},\vec{p}_{0})f(\vec{p}_{0}) \label{FiberGen1} \\
\mathcal{E}_{\vec{k}}(f)(\vec{p})&=& 2^{-1}\big(\mathcal{E}(\vec{p}-2^{-1}\hbar\vec{k} ) + \mathcal{E}(\vec{p}+2^{-1}\hbar\vec{k} )    \big)f(\vec{p}) \nonumber
\end{eqnarray}
where $\mathcal{E}(\vec{p})=\int_{\R^{3}}d\vec{p}_{0}\,\mathcal{J}_{0}(\vec{p}_{0},\vec{p})$, and   $\mcj_{\vec{k}}(\vec{p},\vec{p}_{0})$ is complex-valued  pseudo jump kernel having the form
$$\mcj_{\vec{k}}(\vec{p},\vec{p}_{0})= M_{\vec{p}-\vec{p}_{0} }\big(\vec{p}_{0}-2^{-1}\hbar\vec{k},\,\vec{p}_{0}+2^{-1}\hbar\vec{k}\big).    $$

The jump ``rates" $\mcj_{\vec{k}}(\vec{p},\vec{p}_{0})$ can be written more explicitly as
\begin{align}\label{OffRates}
\mcj_{\vec{k}}(&\vec{p}+\vec{q}  , \vec{p} )\nonumber   \\ &=M_{\vec{q}}\big(\vec{p}-2^{-1}\hbar\vec{k},\,\vec{p}+2^{-1}\hbar\vec{k}\big)\nonumber   \\ &= \int_{(\vec{q})_{\perp}}d\vec{v}\,L_{\vec{q},\vec{v}}(\vec{p}-2^{-1}\hbar\vec{k})\overline{L}_{\vec{q},\vec{v}}(\vec{p}+2^{-1}\hbar\vec{k})   \, \nonumber   \\ &= \eta\frac{m}{m_{*}^{2}}|\vec{q}|^{-1}\,e^{-\frac{\beta m \hbar^{2} }{8M^{2}}\vec{k}_{\parallel \vec{q}}^{2}   }\int_{(\vec{q})_{\perp}}d\vec{v}\,   r\big(\vec{v}+2^{-1}\frac{m}{m_{*}}\vec{q}+\frac{m}{M}\vec{p}_{\parallel \vec{q}}    \big)\mathbf{f} \big(\mathbf{Q}_{-},\,\Theta_{-}\big)     \,\overline{\mathbf{f}}\big(\mathbf{Q}_{+},\,\Theta_{+}\big)  
\end{align}
where $\mathbf{Q}_{\pm}$ and  $\Theta_{\pm}$ are defined as 
\begin{eqnarray*}
\mathbf{Q}_{\pm}&=&\mathbf{Q}_{\pm}(\vec{q},\vec{v},\vec{p},\vec{k}  )=\big( \big(\frac{m_{*}}{m}\vec{v}-\frac{m_{*}}{M}(\vec{p}_{\perp q}\pm 2^{-1}\hbar\vec{k}_{\perp q}   )   \big)^{2}   + 4^{-1}\vec{q}^{2}     \big)^{\frac{1}{2}},  \\  \Theta_{\pm} &=&\Theta_{\pm}(\vec{q},\vec{v},\vec{p},\vec{k}  )=2\tan^{-1}\Big(\frac{\big|2^{-1}\vec{q}  \big|}{\big|\frac{m_{*}}{m}\vec{v}-\frac{m_{*}}{M}(\vec{p}_{\perp q}\pm 2^{-1}\hbar\vec{k}_{\perp q}   )  \big|}    \Big), \\
\vec{z}_{\pm}&=&\vec{z}_{\pm}(\vec{q},\vec{v},\vec{p},\vec{k}  )= \vec{v}+2^{-1}\frac{m}{m_{*}}\vec{q}+\frac{m}{M}(\vec{p}_{\parallel \vec{q}}   \pm 2^{-1}\hbar\vec{k}_{\parallel \vec{q}}),
\end{eqnarray*}
and I have used the simple relation
\begin{align}\label{ARelate}
r(\vec{z}_{-} )^{\frac{1}{2}}r\big(\vec{z}_{+}   \big)^{\frac{1}{2}}  =e^{-\frac{\beta m \hbar^{2}}{8M^{2}}\vec{k}_{\parallel \vec{q}}^{2}   }\, r\big(\vec{v}+2^{-1}\frac{m}{m_{*}}\vec{q}+\frac{m}{M}\vec{p}_{\parallel \vec{q}}    \big).
\end{align}

\subsection{Detailed balance}\label{RemarkDetailedBalance}

  There is  also a detailed balance symmetry for the noise map $\Psi$.  The  function  $\textup{T}_{\vec{q}}(p_{1},p_{2})$ for the multiplication maps $\mathbf{T}_{\vec{q}}:\mathcal{B}_{1}\big(L^{2}(\R^{3})\big)$ (defined in~(\ref{Kraus})) can be written in the form
 \begin{align}\label{DetailedEms}
  \textup{T}_{\vec{q}}(\vec{p}_{1},\vec{p}_{2})=   \frac{ \upsilon_{\infty}^{\frac{1}{4}}(\vec{p}_{1}+\vec{q})\upsilon_{\infty}^{\frac{1}{4}}(\vec{p}_{2}+\vec{q})    }{ \upsilon_{\infty}^{\frac{1}{4}}(\vec{p}_{1})\upsilon_{\infty}^{\frac{1}{4}}(\vec{p}_{2})   } \gamma_{\vec{q}}(\vec{p}_{1},\vec{p}_{2}) ,
  \end{align}
 where  $\gamma_{\vec{q}}(\vec{p}_{1},\vec{p}_{2})$ is the kernel for a positive operator on the Hilbert space $L^{2}(\R^{3})$ for each $\vec{q}$ and satisfies the symmetry
 \begin{align}\label{GammaSymmetry}
   \gamma_{\vec{q}}(\vec{p}_{1},\vec{p}_{2})= \gamma_{-\vec{q}}(\vec{p}_{1}+\vec{q},\vec{p}_{2}+\vec{q}).
  \end{align}
The explicit form for $\gamma_{\vec{q}}(\vec{p}_{1},\vec{p}_{2})$ can easily be worked out, and along the diagonal, it reduces to $  \gamma_{\vec{q}}(\vec{p},\vec{p})= A(\vec{p}+\vec{q},\vec{p})  $ for the kernel $A$ defined in~(\ref{JumpRatesAgain}).
  
 The above   translates into a decomposition for the noise map $\Psi$.   Let $\Upsilon:\mathcal{B}_{1}\big(L^{2}(\R^{3})\big) $ by given by
  $\Upsilon (\rho)  :=  \upsilon_{\infty}^{\frac{1}{4}}(\vec{P})\,\rho \,\upsilon_{\infty}^{\frac{1}{4}}(\vec{P})   $, and formally define 
  $\Gamma=  \Upsilon^{-1}\Psi \Upsilon  $.  The above decomposition for the maps $\mathbf{T}_{\vec{q}}$ implies that $\Gamma$  satisfies $$\hspace{4cm}\overline{\Gamma^{*}(\overline{\rho} )}=\Gamma(\rho),  \hspace{2cm}\rho\in \mathcal{B}_{1}\big(L^{2}(\R^{3})\big), $$ where the bars refer to  complex conjugation is in the momentum representation.  We do use this formal relation, although the symmetry~(\ref{GammaSymmetry}) is useful for seeing a cancellation  in the proof of Prop.~\ref{Perturbation}.

\section{ The ``marginal" momentum dynamics    }\label{SecMarMom}

For a translation-invariant linear Boltzmann dynamics, the  momentum is a Markovian process.  Thus, the marginal probability density for the momentum obtained by integrating out the spacial degrees of freedom is governed by an autonomous master equation.  The quantum situation is similar in this respect as seen in the classical Kolmogorov equation~(\ref{MomDist}) having jump kernel $\mathcal{J}(\vec{p}_{2},\vec{p}_{1})$, which corresponds to the fiber decomposition from the last section for the $\vec{k}=0$ case (I drop the index for $\vec{k}=0$).  I denote the gain and loss terms for the Markov generator $\mathcal{L}_{0}$ by $\mathcal{J}$ and $\mathcal{E}$, respectively. In this section, I will characterize the ergodicity for the classical stochastic semigroup $e^{t\mathcal{L}_{0}}$ and collect some analytic facts about the generator $\mathcal{L}_{0}$.  As mentioned in Sect.~\ref{SecDis}, the momentum distribution is exponentially ergodic to a Gaussian with variance $\frac{M}{\beta}$, and that statement is made more precise below in Prop.~\ref{RelBnds}.  In the analysis, I rely on soft consequences of Nussenzveig's results in~\cite{Nussen2} to ensure  the scattering cross section $\frac{d\sigma}{d\Omega}(\mathbf{p},\theta)$ does not vanish for $\theta\in (0,\pi)$ bounded away from the endpoints    as $\textup{p}\nearrow \infty$.  
 
  The following lemma gives upper and lower bounds for the escape rates $\mathcal{E}(\vec{p})= \int_{\R^{3}}d\vec{p}_{0} \mathcal{J}(\vec{p}_{0},\vec{p})$.  An implication is that the operator $\Psi^{*}(\textup{I})=\mathcal{E}(\vec{P})$ acting on the Hilbert space $L^{2}(\R^{3})$ is relatively bounded to $|\vec{P}|$.  As I mentioned in the introduction, the total cross section $\sigma_{\textup{tot}}(\mathbf{p})>0$ appearing in the bounds of Lem.~\ref{Horses} approaches $4\pi a^{2}$ for $\mathbf{p}\ll 1$ and $2\pi a^{2}$ for $\mathbf{p}\gg 1$.

\begin{lemma}\label{Horses}
The escape rate $\mathcal{E}(p)$ has the following upper and lower bounds
$$ \underline{\mathcal{E}}\vee (\underline{\mathcal{F}}\,|\vec{p}|)    \leq \mathcal{E}(\vec{p})\leq \overline{\mathcal{E}}+\overline{\mathcal{F}}\,|\vec{p}|, $$
where the constants $\underline{\mathcal{E}},\overline{\mathcal{E}},\underline{\mathcal{F}},\overline{\mathcal{F}}$  are defined as
\begin{eqnarray*}
\underline{\mathcal{E}}&=& \big(\inf_{\mathbf{p}\in \R_{+} }\sigma_{\textup{tot}}(\mathbf{p})\big)\frac{2\eta}{m_{*}}\big(\frac{m_{*}  }{m  }\big)^{4} \big(\frac{2m}{\pi\beta}   \big)^{\frac{1}{2} },  \quad  \underline{\mathcal{F}}= \big(\inf_{\mathbf{p}\in \R_{+}}\sigma_{\textup{tot}}(\mathbf{p})\big)\frac{\eta}{M}\big(\frac{m_{*}  }{m  }\big)^{3},  \\  \overline{\mathcal{E}}&=& \big(\sup_{\mathbf{p}\in \R_{+}}\sigma_{\textup{tot}}(\mathbf{p})\big)\frac{2\eta}{m_{*}}\big(\frac{m_{*}  }{m  }\big)^{4} \big(\frac{2m}{\pi\beta}   \big)^{\frac{1}{2} } , \quad \overline{\mathcal{F}}= \big(\sup_{\mathbf{p}\in \R_{+}}\sigma_{\textup{tot}}(\mathbf{p})\big)\frac{\eta}{M}\big(\frac{m_{*}  }{m  }\big)^{3} .
\end{eqnarray*}

\end{lemma}

\begin{proof}
  The escape rate $\mathcal{E}(\vec{p})$ can be written
\begin{align}\label{HereToFore}
\mathcal{E}(\vec{p})&= \eta \frac{m}{m_{*}^{2}}\int_{\R^{3}} d\vec{q}\frac{1}{|\vec{q}|}\int_{(\vec{q})_{\perp} }d\vec{v}\, r\big(\vec{v}+2^{-1}\frac{m}{m_{*}}\vec{q}+\frac{m}{M}\vec{p}_{\parallel \vec{q}} \big)  \,\frac{d\sigma}{d\Omega}\big(|\vec{p}_{rel}(\vec{q},\vec{v},\vec{p}_{\perp \vec{q}})|,\,\theta(\vec{q},\vec{v},\vec{p}_{\perp \vec{q}})  \big)\nonumber  \\ &=  \frac{ \eta}{m_{*} } \int_{\R^{3}} d\vec{p}_{\textup{rel}}\big| \vec{p}_{\textup{rel}}  \big|\,r\big(\frac{m}{m_{*}}\vec{p}_{\textup{rel} }+\frac{m}{M}\vec{p}  \big)\int_{\Omega}\frac{d\sigma}{d\Omega}\big( \big|\vec{p}_{\textup{rel} }\big|,\theta \big) \nonumber     \\    &= \eta\frac{ m^{2}_{*}}{m^{3} } \int_{\R^{3}} d\vec{p}_{0}\big|\frac{m_{*}}{m}\vec{p}_{0}-\frac{m_{*}}{M}\vec{p}   \big|\,r(\vec{p}_{0} )\int_{\Omega}\frac{d\sigma}{d\Omega}\big( \big|\frac{m_{*}}{m}\vec{p}_{0}-\frac{m_{*}}{M}\vec{p}\big|,\theta \big),
\end{align}
where $\vec{p}_{rel}=\frac{m_{*}}{m}\vec{v}-\frac{m_{*}}{M}\vec{p}_{\perp \vec{q}}+2^{-1}\vec{q}$ and $\theta=2\tan^{-1}\big(\frac{2^{-1}|\vec{q}|}{|\frac{m_{*}}{m}\vec{v}-\frac{m_{*}}{M}\vec{p}_{\perp \vec{q}}|   }\big) $.  
The integral $\int_{\Omega}$ is over all unit vectors $\hat{\theta}\in \R^{3}$ and $\theta$ is the angle between $\vec{p}_{\textup{rel} }$ and $\hat{\theta}$.  
The second and third equalities in~(\ref{HereToFore}) are changes of integration variables, where the third equality is simply due to $\vec{p}_{\textup{rel} }\rightarrow \vec{p}_{0}= \frac{m}{m_{*}}\vec{p}_{\textup{rel} }+\frac{m}{M}\vec{p} $.   
For the change of integration in the second inequality, I consider a smooth extension of $\frac{d\sigma}{d\Omega}(|\vec{p}_{\textup{rel}}|,\theta)$   to  a function $\frac{d\sigma}{d\Omega}\big(\vec{p}_{\textup{rel}},\,\vec{p}_{\textup{rel},0}\big)$, $\vec{p}_{\textup{rel},0}\in \R^{3}$,  where  $\frac{d\sigma}{d\Omega}(|\vec{p}_{\textup{rel}}|,\theta)=\frac{d\sigma}{d\Omega}\big(\vec{p}_{\textup{rel}},\,\vec{p}_{\textup{rel},0}\big)$
when $|\vec{p}_{\textup{rel}}|=|\vec{p}_{\textup{rel},0}|$.    
I relate the variables $\vec{p}_{\textup{rel},0}$ and $\vec{p}_{\textup{rel}}$  
to $\vec{v},\vec{q}\in \R^{3}$ as  $\vec{p}_{\textup{rel},0}=\frac{m_{*}}{m}\vec{v}-\frac{m_{*}}{M}\vec{p}_{\perp \vec{q}}-2^{-1}\vec{q}$ and as above for $\vec{p}_{\textup{rel}}$.  I can go from  integration over $\vec{p}_{\textup{rel}}$, $\theta$ to integration over $\vec{p}_{\textup{rel}},\vec{p}_{\textup{rel},0}$ and then to integration over $\vec{q},\vec{v}$ by the same reasoning as in~\cite[App.A]{VaccHorn}: 
\begin{align}
\int_{\R^{3}} d\vec{p}_{\textup{rel}}\big|\vec{p}_{\textup{rel}}| \int_{\Omega}&=\int_{\R^{3}} d\vec{p}_{\textup{rel}}\int_{\R^{3}} d\vec{p}_{\textup{rel},0}\,\delta\big( 2^{-1}|\vec{p}_{\textup{rel}}|^{2}-2^{-1}|\vec{p}_{\textup{rel},0}|^{2}\big) \nonumber   \\ &=\int_{\R^{3}}d\vec{q}\int_{\R^{3}}d\vec{v}\,\delta\big(\frac{m_{*}}{m}\vec{q}\cdot \vec{v}   \big) \nonumber \\ &  =\frac{m}{m_{*}}\int_{\R^{3}}d\vec{q}\frac{1}{|\vec{q}|}\int_{\R^{3}}d\vec{v}\,\delta\big(\frac{\vec{q}}{|\vec{q}|}\cdot \vec{v}   \big)\nonumber  \\ & =  \frac{m}{m_{*}}\int_{\R^{3}}d\vec{q}\frac{1}{|\vec{q}|}\int_{(\vec{q})_{\perp} }d\vec{v}.
\end{align}

With~(\ref{HereToFore}), I have the following upper bound for the quantum escape rates:
\begin{align*}
\mathcal{E}(\vec{p}) &\leq \big(\sup_{\mathbf{p}\in \mathbb{R}_{+} }\sigma_{\textup{tot}}(\mathbf{p})\big)\eta \frac{m^{2}_{*}}{m^{3}} \int_{\R^{3}} d\vec{p}_{0}\big|\frac{m_{*}}{m}\vec{p}_{0}-\frac{m_{*}}{M}\vec{p}\big| \frac{e^{-\frac{\beta}{2m}\vec{p}_{0}^{2}   }   }{(2\pi m\beta^{-1})^{\frac{3}{2}}    }\\ & \leq \big(\sup_{\mathbf{p}\in \mathbb{R}_{+} }\sigma_{\textup{tot}}(\mathbf{p})\big)\eta \frac{m^{2}_{*}}{m^{3}}\big( 2\frac{m_{*}}{m}(\frac{2m}{\pi\beta})^{\frac{1}{2}}  + \frac{m_{*}}{M}|\vec{p}| \big).  
\end{align*}
This gives the coefficients for an upper bound.  For a lower bound, I have
$$\mathcal{E}(\vec{p})\geq \big(\inf_{\mathbf{p}\in \mathbb{R}_{+} }\sigma_{\textup{tot}}(\mathbf{p})\big)\eta \frac{m^{2}_{*}}{m^{3}} \int_{\R^{3}} d\vec{p}_{0}\big|\frac{m_{*}}{m}\vec{p}_{0}-\frac{m_{*}}{M}\vec{p}\big| \frac{e^{-\frac{\beta}{2m}\vec{p}_{0}^{2}   }   }{(2\pi m\beta^{-1})^{\frac{3}{2}}    }.$$
The right side is minimized for $p=0$ with the value $\underline{\mathcal{E}}$.  Moreover, since the integral is a weighted average with mean zero of the convex function $F(\vec{p}_{0})=\big|\frac{m_{*}}{m}\vec{p}_{0}-\frac{m_{*}}{M}\vec{p}\big| $, the integral is smaller than the value $F(0)=\frac{m_{*}}{M}|\vec{p}|$.

\end{proof}

The following theorem summarizes a small part of the results from~\cite{Nussen2}.  To prove $\mathcal{L}_{0}$ has a spectral gap in Prop.~\ref{RelBnds}, it is imperative that the scattering does not all become absorbed in the forward direction when the test particle is traveling with high momentum.  The lemma states that the differential cross section $\frac{d\sigma}{d\Omega}(\mathbf{p},\theta)$ approaches $4^{-1}a^{2}$ as $\mathbf{p}\rightarrow \infty$ for $\theta$ bounded away from the endpoints  $0,\pi$.

\begin{theorem}[Nussenzveig]\label{Nussen}
As $ \frac{a}{\hbar} \mathbf{p}\rightarrow \infty $, there is uniform convergence 
$$\hspace{4cm}\frac{d\sigma}{d\Omega}=|\mathbf{f}(\mathbf{p},\,\theta  )|^{2}\longrightarrow \frac{a^{2}}{4}, \quad \quad \quad  \delta \leq  \theta \leq \pi -\delta,    $$
for any fixed $\delta >0$.   

\end{theorem}

\begin{proof}
By~\cite[Sec. 3]{Nussen2}, there is uniform convergence
$$ \mathbf{f}(\mathbf{p},\,\theta  )\longrightarrow -\frac{a}{2}e^{-2\ii\frac{a}{\hbar}\mathbf{p}  \sin(\frac{\theta}{2}) }  $$
as $\frac{a}{\hbar}\mathbf{p}\rightarrow \infty$ for $\theta$ bounded away from the endpoints of the interval $[0,\pi]$.  The estimates in~\cite{Nussen2} even allow for uniform convergence over intervals $ [(\frac{a}{\hbar}\mathbf{p})^{-\gamma},\,\pi-(\frac{a}{\hbar}\mathbf{p})^{-\gamma}]$ for $0 <\gamma < 3^{-1} $.

\end{proof}

For $y\in L^{1}(\R^{3})$, define the multiplication operator $\frak{M}(y)(\vec{p})=|\vec{p}|y(\vec{p})$.  Part (5) of the proposition below states that the graph norm for $\frak{M}^{n}$ is equivalent to the graph norm for the operator $\mathcal{L}_{0}^{n}$.  This equivalence is useful, since  $\|\cdot\|_{\frak{M}^{n}}$ is easy to compute in some situations and the norm $\| \cdot \|_{\mathcal{L}_{0}^{n}}$ is contractive under the dynamical evolution $e^{t\mathcal{L}_{0}}$.  By the discussion in Sect.~\ref{SecCons}, the domain $\textup{D}(\mathcal{L}_{0})$ contains $\mathcal{T}_{4}$, and the domain for $\textup{D}(\mathcal{L}_{\vec{k}})$ also includes $\mathcal{T}_{4}$ by the same argument.  Since $\mathcal{L}_{\vec{k}}$ is closed, Part (5) of Prop.~\ref{RelBnds} for $n=1$ implies  $\textup{D}(\mathcal{L}_{ 0})=\mathcal{T}_{1}$, and Part (6) implies  $\textup{D}(\mathcal{L}_{\vec{k}})$ contains $\mathcal{T}_{1}$.  Moreover, Part (5) also implies  $\mathcal{L}_{ 0}^{n}$ maps elements in $\mathcal{T}_{N}$ into $\mathcal{T}_{N-n}$ for $N\geq n$.

\begin{proposition}\label{RelBnds}
Let $\mathcal{L}_{\vec{k}}$ be the generator for the contractive semigroup $\widetilde{\Phi}_{t}^{(\vec{k})}$ acting on the $\vec{k}$-fiber.   Define the projection $\mathbf{P}$  to operate as $\mathbf{P}y=\big(\int_{\R^{3}}d\vec{p}\,y(\vec{p})     \big)\nu_{\infty}$ for $y\in L^{1}(\R^{3})$.

\begin{enumerate}

\item Zero is a non-degenerate eigenvalue for $\mathcal{L}_{0}$ with corresponding spectral projection  $\mathbf{P}$.

\item The remainder of spectrum satisfies   $\sum(\mathcal{L}_{0})-\{0\}\subset (-\infty,\,-\frak{g}]+\ii\,\R$ for some $\frak{g}>0$.  

\item Given $\delta>0$, there is a $C_{\delta}>0$ such that $ \| \widetilde{\Phi}_{t}^{(0)}(\textup{I}-\mathbf{P})  \|\leq C_{\delta} e^{(\delta-\frak{g})t } $ for all $t>0$. 

\item  For $C_{\delta}>0$ defined above and  $z\in \C$ with $\textup{Re}(z)>\delta-\frak{g}$,
$$ \big\| \frac{1}{\mathcal{L}_{0}-z} \big\| \leq |z|^{-1}+\frac{C_{\delta}}{\textup{Re}(z)-\delta+\frak{g} }\quad \text{and} \quad \big\| \frac{\mathcal{L}_{0}}{\mathcal{L}_{0}-z} \big\| \leq 2+\frac{C_{\delta}|z|}{\textup{Re}(z)-\delta+\frak{g} }  .   $$ 

\item The graph norms $\| \cdot \|_{\mathcal{L}_{0}^{n}}$ and  $\|\cdot\|_{\frak{M}^{n}}$ are equivalent.  

\item $\mathcal{L}_{\vec{k}}$ is relatively bounded to $\mathcal{L}_{0}$.

\end{enumerate}

\end{proposition}

\begin{proof}\text{  }\\
\noindent Part (1):\\

It can be verified that the state $\nu_{\infty}(\vec{p})=\big(2\pi M\beta^{-1}\big)^{-\frac{1}{2}}e^{-\frac{\beta}{2M}\vec{p}^{2}}$ satisfies $\mathcal{L}_{0}\nu_{\infty}=0$ by a direct computation.  This also follows from the detailed balance form of the jump rates~(\ref{JumpRatesAgain}).   Since the jump rates $\mathcal{J}(\vec{p}_{2},\vec{p}_{1})$ are strictly positive,  the process is ergodic and the only stationary state is  $\nu_{\infty}$.  Thus, $0$ is a non-degenerate eigenvalue of $\mathcal{L}_{0}$.

\vspace{.5cm}

\noindent Part (2):\\

 Let $\textup{I}_{r}:L^{1}(\R^{3})\rightarrow L^{1}(\R^{3})$ be the  projection onto the ball of radius $r>0$ (i.e. multiplication by the indicator function $\chi(|\vec{p}|\leq r)$) and define the operator $\mcg_{r}=\mathcal{L}_{0}(\textup{I}-\textup{I}_{r})-\mce\,\textup{I}_{r}$ with domain $\mathcal{T}_{1}$.  I will argue the following points: 
\begin{enumerate}[(i).]
\item $\mcj \textup{I}_{r}$ is compact.

\item For sufficiently large $r$, $\mcg_{r}$ has spectrum in  $(-\infty,\,-h]+\ii\,\R$ for some $h>0$.  In fact, 
\begin{align}\label{FryCook}
 \| e^{-t\mcg_{r}}\|\leq 2 e^{-th},   
\end{align}
where $\|\cdot\|$ is the operator norm for maps on $L^{1}(\R^{3})$.
\end{enumerate}
Assuming the above two points, then since $\mcl_{0}=\mcg_{r}+\mcj \textup{I}_{r}$, it follows by Weyl's theorem~\cite{Kato} that $\mcl_{0}$ and $\mcj \textup{I}_{r}$ have the same essential spectrum.  Since $\mcl_{0}$ generates a contractive semigroup $\sum(\mcl_{0})\subset -\R_{+}+ \ii\, \R$, the set $\sum(\mcl_{0}) \cap \big[(-h,0]+\ii\,\R\big] $ may only contain elements in the point spectrum  of $\mcl_{0}$ and no accumulation points.  In fact, there will only be finitely many elements in $\sum(\mcl_{0}) \cap \big[ (-h+\epsilon,0]+\ii\,\R  \big]$ for any $\epsilon>0$.  To see this, let us suppose there were an infinite sequence of normalized eigenvectors $e_{n}\in L^{1}(\R^{3})$ for $\mathcal{L}_{0}$ with eigenvalues $v_{n}\in (-h+\epsilon,0]+\ii\,\R  $.  Since there are no accumulation points, the $v_{n}$ must have $|\textup{Im}(v_{n})|\rightarrow \infty$.  Notice that     
\beq \label{Finite} 
\mathcal{J}\textup{I}_{r}e_{n}= \big(\mathcal{L}_{0}-\mathcal{G}_{r}\big)e_{n}=\big(v_{n}-\mathcal{G}_{r}\big)e_{n}.    \eeq
The left side of~(\ref{Finite}) tends to zero in $L^{1}(\R^{3})$ as $n\rightarrow \infty$.  To see this, first observe 
$$\big\|\mathcal{J}\textup{I}_{r}e_{n}\big\|_{1}\leq  \Big(\sup_{|\vec{p} |\leq r}\| l_{\vec{p}}\|_{1}\Big)\int_{|\vec{p}|\leq r} |e_{n}(\vec{p})| ,      $$
where $l_{\vec{p}_{0}}(\vec{p})=\mcj(\vec{p},\vec{p}_{0})$.  By using that $\mcl_{0}e_{n}=v_{n}e_{n}$ and $\|e_{n}\|_{1}=1$, there is the inequality
$$\sup_{|\vec{p}|\leq r}\big|e_{n}(\vec{p})\big|\leq   \frac{  \sup_{|\vec{p}|\leq r ,\,\vec{p}_{0}\in \R^{3}  } \mcj(\vec{p},\vec{p}_{0})  }{ |\textup{Im}(v_{n})| }\longrightarrow 0\quad \text{as}\quad |\textup{Im}(v_{n})|\longrightarrow \infty .  $$
 Moreover, the right side of~(\ref{Finite}) must have $L^{1}$ norm $\geq 2^{-1}\epsilon$ for each $n$, since $\big(v_{n}-\mathcal{G}_{r}\big)^{-1}=\int_{0}^{\infty}dt e^{t(\mathcal{G}_{r}-v_{n})}$ has operator norm $\leq 2\epsilon^{-1}$ by~(\ref{FryCook}).  This contradiction shows  there can be only finite many eigenvalue in the region  $(-h+\epsilon,0]+\ii\,\R  $.  All of these  eigenvalues must have real part strictly smaller than $0$ by the mixing generated from the strictly positive jump rates $\mathcal{J}(p_{2},p_{1})>0$.      Hence, there is a gap $\frak{g}$ between $0$ and the real component of the remainder of the spectral values $\Sigma(\mathcal{L}_{0})-\{0\}$.

I now turn to establishing the points (i) and (ii).  For the compactness of $\mcj \textup{I}_{r}$, notice  $\mcj \textup{I}_{r}$ maps elements $y\in L^{1}(\R^{3})$ to integral combinations of the functions $l_{\vec{p}}(\vec{p}_{0})$ through the formula
$$\mcj \textup{I}_{r}(y)= \int_{|\vec{p}|\leq r }d\vec{p} y(\vec{p})\,l_{\vec{p}}.   $$
The compactness of $\mcj \textup{I}_{r}$ follows, since $l_{\vec{p}_{0}}\in L^{1}(\R^{3})$ is a continuous function with respect to the $L^{1}$ norm over the compact set $|\vec{p}_{0}|\leq r $.  

Proving  $\textup{Re}\sum (\mcg_{r})$ is bounded away from zero for large enough $r$ is more involved.  Notice that the semigroup $e^{t\mcg_{r}}:L^{1}(\R^{3})$ maps positive functions to positive functions, since $\mcg_{r}$ agrees with $\mcl_{0}$ on the image of $\textup{I}-\textup{I}_{r}$ and acts as a real-valued multiplication operator on the image of  $\textup{I}_{r}$.  Showing that $y_{t}=e^{t\mcg_{r}}y$ has $\|y_{t}\|_{1}\leq 2e^{-th}\|y\|_{1}$ for $h>0$ and all positive functions $y$ gives the upper bound   $\textup{Re}(\sum \mcg_{r})\leq -h$.  

 It is convenient to consider a slightly deformed norm $\|\cdot\|_{w}$ that is equivalent to the ordinary $\|\cdot \|_{1}$ norm.  Define the function $w(\vec{p})= 2-\frac{1}{|\vec{p}|+1}$ and the corresponding weighted norm
$$\|y\|_{w}= \int_{\R^{3}}d\vec{p}\,w(\vec{p})\,|y(\vec{p})|,\quad \quad y\in L^{1}(\R^{3}),   $$   
so that I have $\|y\|_{1}\leq \|y\|_{w}\leq 2\|y\|_{1}$.  For a positive element $y\in \mathcal{T}_{1}$,  
\begin{align*}
\frac{d}{dt}\|y_{t}\|_{w}&= \int_{\R^{3}} d\vec{p}\,w(\vec{p})\,(\mcg_{r}y_{t})(\vec{p})\\ &= \int_{\R^{3}} d\vec{p}\,(\mcg_{r}^{*}w)(\vec{p})\,y_{t}(\vec{p})\\ &=\int_{|\vec{p}|\geq r} d\vec{p}\,\big(\mathcal{L}_{0}^{*}w\big)(\vec{p})\,y_{t}(\vec{p})-\int_{|\vec{p}|\leq r} d\vec{p}\,w(\vec{p})\mathcal{E}(\vec{p})y_{t}(\vec{p}). 
\end{align*}
Since $\mce(\vec{p})$ is bounded from below by $\mathcal{E}(0)=\underline{\mathcal{E}}$, I only need to demonstrate sufficient decay from the first term on the right-hand side.  My estimates will show  there are $h^{\prime},r>0$ such that
\begin{align}\label{Gipper}
(\mathcal{L}_{0}^{*}w)(\vec{p})=  \int_{\R^{3}} d\vec{p}_{0}\,\big(w(\vec{p}_{0})\mcj(\vec{p}_{0},\vec{p})-w(\vec{p})\mcj(\vec{p}_{0},\vec{p})\big)\leq -h^{\prime} \quad \quad \forall (|\vec{p}|>r).   
\end{align} 
I can then take $h= h^{\prime}\wedge \underline{\mathcal{E}} $ as a lower bound for the rate of exponential decay for $\|y_{t}\|_{w}$ and thus also for $\|y_{t}\|_{1}$.  The physical reason for expecting the exponential decay in $\|y_{t}\|_{w}$ is that when the test particle is traveling with high enough momentum $\vec{p}_{0}$, then the first collision will typically knock it to a momentum $\vec{p}_{1}$ such that $|\vec{p}_{1}|$ is a fraction of the value $|\vec{p}_{0}|$.  How large $\vec{p}_{0}$ should be for this effect will depend on the masses $m$, $M$ and the temperature $\beta$, in order that the test particle is traveling at a much higher speed than the particles in the reservoir.  Finally, this effect is amplified by the fact that collisions occur more frequently in proportion to the speed of the test particle due to the ``hard" character of the interaction.  The combination of these two features is enough to be visible to the gently sloping weight $w(\vec{p})$.

\begin{align}\label{Backward}
&\int_{\R^{3}} d\vec{p}_{0}\,\big(w(\vec{p}_{0})\mcj(\vec{p}_{0},\vec{p})-w(\vec{p})\mcj(\vec{p}_{0},\vec{p})\big)\nonumber \\ =&\eta\,\frac{m}{m_{*}^{2}}\int_{\R^{3}}\frac{d\vec{q}}{|\vec{q}|}\int_{(\vec{q})_{\perp}}d\vec{v}\,r\big(\vec{v}+2^{-1}\frac{m}{m_{*}}\vec{q}+\frac{m}{M}\vec{p}_{\parallel \vec{q}}    \big)\nonumber \\  &\times \big|\mathbf{f}\big(|\vec{p}_{rel}(\vec{q},\vec{v},\vec{p}_{\perp \vec{q}})|,\theta(\vec{q},\vec{v},\vec{p}_{\perp \vec{q}})  \big)\big|^{2}\big(w(\vec{p}+\vec{q})-w(\vec{p}) \big)\nonumber  \\ =&\frac{\eta}{m_{*}}\int_{\Omega} \int_{\R^{3}}d\vec{p}_{rel}\,|\vec{p}_{rel}|\, r\big(\frac{m}{m_{*}}\vec{p}_{rel}+\frac{m}{M}\vec{p}  \big) \,|\mathbf{f}(|\vec{p}_{rel}|,\theta)|^{2}\big(w(\vec{p}+\vec{p}_{rel}-\hat{\theta}|\vec{p}_{rel}|)-w(\vec{p}) \big),
\end{align}
where $\hat{\theta}$ is a vector variable on the unit sphere with angle $\theta$ from $\vec{p}_{\textup{rel}}$, and the second equality above follows by the same change of integration as in the proof of Lem.~\ref{Horses}.

  By changing variables, $\vec{b}=\frac{m}{m_{*}}\vec{p}_{rel}+\frac{m}{M}\vec{p}$ and rewriting the integrand from~(\ref{Backward}), I have
$$ 
-\frac{\eta\, m_{*}^{2} }{m^{3}}\int_{\Omega} \int_{\R^{3}}d\vec{b}\, r(\vec{b}) \,|\mathbf{f}\big(\big|\frac{m_{*}}{m}\vec{b}-\frac{m_{*}}{M}\vec{p}\big|  ,\theta\big)|^{2}\frac{|\frac{m_{*}}{M}\vec{p}-\frac{m_{*}}{m}\vec{b}|( |\vec{p}|-\big|(1+\frac{m_{*}}{M})\vec{p}-\frac{m_{*}}{m}\vec{b}-\hat{\theta}|\frac{m_{*}}{M}\vec{p}-\frac{m_{*}}{m}\vec{b} |  \big|)    }{(|\vec{p}|+1)(\big|(1+\frac{m_{*}}{M})\vec{p}-\frac{m_{*}}{m}\vec{b}-\hat{\theta}|\frac{m}{M}\vec{p}-\frac{m_{*}}{m}\vec{b} |  \big|+1) }.  
$$
 Due to the Gaussian factor $r(\vec{b})$, the integrand decays rapidly for $|\vec{b}|\gg (\frac{m}{\beta})^{\frac{1}{2}}  $.   Hence for large $|\vec{p}|$, I have $|\vec{p}|\gg |\vec{b}|$ and
\begin{multline*}
\frac{|\frac{m_{*}}{M}\vec{p}-\frac{m_{*}}{m}\vec{b}|( |\vec{p}|-\big|(1-\frac{m_{*}}{M})\vec{p}+\frac{m_{*}}{m}\vec{b}-\hat{\theta}|\frac{m_{*}}{M}\vec{p}-\frac{m_{*}}{m}\vec{b} |  \big|)    }{(|\vec{p}|+1)(\big|(1-\frac{m_{*}}{M})\vec{p}+\frac{m_{*}}{m}\vec{b}-\hat{\theta}|\frac{m}{M}\vec{p}-\frac{m_{*}}{m}\vec{b} |  \big|+1) }\\ =\frac{m_{*}}{M}\Big( \big(1-\frac{m_{*}}{M}+\frac{m_{*}}{M}\cos(\theta))^{2}+\frac{m^{2}_{*}}{M^{2}}\sin^{2}(\theta)\big)^{-\frac{1}{2}} -1\Big)    +\mathit{O}( |\vec{p}|^{-1}).
\end{multline*}
I note that 
\begin{align*}
\big(1-\frac{m_{*}}{M}+\frac{m_{*}}{M}\cos(\theta)\big)^{2}+\frac{m^{2}_{*}}{M^{2}}\sin^{2}(\theta)&=(1-\frac{m_{*}}{M})^{2}+\frac{m^{2}_{*}}{M^{2}}+2\frac{m_{*}}{M}(1-\frac{m_{*}}{M})\cos(\theta) \\
&= \frac{m^{2}+M^{2}+2mM\cos(\theta)     }{(m+M)^{2}}\\ &\leq 1, 
\end{align*} 
where equality occurs only for $\theta=0$. 

With the estimate above
\begin{multline}\label{Monty}
\int_{\R^{3}} d\vec{p}_{0}\,\big(w(\vec{p}_{0})\mcj(\vec{p}_{0},\vec{p})-w(\vec{k})\mcj(\vec{p}_{0},\vec{p})\big) =-\eta\frac{m^{3}_{*}}{m^{3}M}\int_{\R^{3}}d\vec{b}\,r(\vec{b})\\ \times \int_{\Omega}  \,|\mathbf{f}\big(|\frac{m}{m_{*}}\vec{b}-\frac{m}{M}\vec{p}|  ,\theta\big)|^{2}\Big( \frac{m+M  }{\big( m^{2}+M^{2}+2m M\cos(\theta) \big)^{\frac{1}{2}}} -1\Big)+\mathit{O}(|\vec{p}|^{-1}).     
\end{multline}
Since $r(\vec{b})$ is probability distribution concentrated around $\vec{b}=0$, I have an average of terms of the form
\beq \label{ScatCross}
\int_{\Omega}  \,|\mathbf{f}(\mathbf{p} ,\theta)|^{2}\Big( \frac{m+M  }{\big( m^{2}+M^{2}+2m M\cos(\theta) \big)^{\frac{1}{2}}}  -1\Big)\quad \quad \text{for}\quad |\mathbf{p}|\gg 1.  
\eeq
However, by Thm.~\ref{Nussen}, $|\mathbf{f}(\mathbf{p} ,\theta)|^{2}$ converges uniformly to $4^{-1}a^{2}$ as $\mathbf{p}\rightarrow \infty$ for $\theta$ bounded away from $0,\pi$.  Hence~(\ref{ScatCross}) will include a contribution from the integration of $\theta$ away from the boundary points of approximately  
\begin{align*}
\frac{a^{2}}{4}\int_{\Omega}\Big(\frac{m+M}{  \big(m^{2}+M^{2}+2mM\cos(\theta)\big)^{\frac{1}{2}} }-1\Big)&= \pi a^{2}\Big((m+M)\frac{m\wedge M}{mM} -1  \Big)\\ &=\pi a^{2}\frac{m\wedge M}{m\vee M}.   
\end{align*}

 Hence, there exists an $r>0$ and an $h^{\prime}>0$ such that~(\ref{Gipper}) holds, which completes the proof.

\vspace{.5cm}

\noindent Part (3):  \\

Since the projection $\mathbf{P}$ commutes with the maps  $\widetilde{\Phi}_{t}^{(0)}$, the operators $ \widetilde{\Phi}_{t}^{(0)}(\textup{I}-\mathbf{P}) $ form a semigroup that may be viewed as operating on the Banach space $\textup{Im}(\textup{I}-\mathbf{P})\subset L^{1}(\R^{3})$.  By~\cite[Sec.IX.4]{Kato}, it is enough to show  there is a $C_{\delta}>0$ such that
\begin{align}\label{SnapOn}
\big\|\big(\mathcal{L}_{0}-z       \big)^{-1} (\textup{I}-\mathbf{P}) \big\|\leq \frac{C_{\delta}}{z+\frak{g}-\delta }  
\end{align}
for all $z>-\frak{g} +\delta $.  By Lem.~\ref{TrivialContract}, for $z>0$ I have the second inequality below 
\begin{align}\label{ShnizzleShnazzle}
\big\|  \big(\mathcal{L}_{0}-z       \big)^{-1}(\textup{I}-\mathbf{P})\big\|\leq   \big\|  \big(\mathcal{L}_{0}-z       \big)^{-1}\big\|+z^{-1}\| \mathbf{P} \|  \leq  2z^{-1}.
\end{align}
On the other hand, the operators $ \big(\mathcal{L}_{0}-z       \big)^{-1}(\textup{I}-\mathbf{P})$ have finite norm for $z\in (-\frak{g},\infty)$, and the norm depends continuously on the parameter $z$.  The continuity gives   
$$\max_{-\frak{g}+\delta \leq z\leq  1   } \big\|  \big(\mathcal{L}_{0}-z       \big)^{-1}(\textup{I}-\mathbf{P})\big\|<\infty,$$
and by combining this with~(\ref{ShnizzleShnazzle}), I can find a $C_{\delta}$ so~(\ref{SnapOn}) holds.

\vspace{.5cm}

\noindent Part (4):  \\

This follows by Part (3), the fact $\|\mathbf{P}\|=1$, and the formula $$\big(\mathcal{L}_{0}-z\big)^{-1}(\textup{I}-\mathbf{P})=-\int_{0}^{\infty}dt e^{-tz}\widetilde{\Phi}_{t}^{(0)}  (\textup{I}-\mathbf{P})$$
for $\textup{Re}(z)>\delta-\frak{g}$.

\vspace{.5cm}

\noindent Part (5):  \\

To bound $\|\cdot\|_{\mathcal{L}_{0}}$ using $\|\cdot\|_{\frak{M}}$, I observe 
\beq \label{Juxt0}
\| \mathcal{L}_{0}y  \|_{1}\leq 2\|\mathcal{E}y\|_{1}\leq 2\big(\overline{\mathcal{E}}\|y\|_{1}+\overline{\mathcal{F}}\|\,|\vec{p}|y\|_{1}\big)\leq 2( \overline{\mathcal{E}}\vee \overline{\mathcal{F}} ) \|y\|_{\frak{M}},  \eeq
where the first inequality uses $\|\mathcal{J}y\|_{1}\leq \|\mathcal{E}y\|_{1}$ (with equality when $y\geq 0$), and the second uses Lem.~\ref{Horses}.  Hence, $\|\mathcal{L}_{0}y\|_{1}$ is smaller than a constant multiple $c=2( \overline{\mathcal{E}}\vee \overline{\mathcal{F}} )$ of $\|y\|_{\frak{M}}$.  To extend to $n>1$, I will use that there are $C_{j}$'s such that for all $y\in L^{1}(\R^{3})$
\beq \label{Juxt1}
\| [\frak{M}^{j},\mathcal{L}_{0}]y  \|_{1} \leq C_{j}\| y\|_{\frak{M}^{j+1}}.   \eeq
 This can be seen by the inequalities
\begin{align*}
\| & [\frak{M}^{j},\mathcal{L}_{0}]y   \|_{1}\\  &\leq  \Big|\int_{\R^{3}}d\vec{p}y(\vec{p})\,\int_{\R^{3}}dq \mathcal{J}(\vec{p}+\vec{q},\vec{p})\big(|\vec{p}+\vec{q}|^{j}-|\vec{p}|^{j}   \big)\Big|  \\ &= \frac{\eta}{m_{*}}\Big|\int_{\R^{3}}d\vec{p}y(\vec{p})\, \int_{\R^{3}}d\vec{p}_{\textup{rel}}\,|\vec{p}_{\textup{rel}}|\int_{\Omega}  \,\big(\big|\vec{p}+  \vec{p}_{\textup{rel}}-\hat{\theta} |\vec{p}_{\textup{rel}}|\big|^{j}-|\vec{p}|^{j}  \big)r\big(\frac{m}{m_{*}}\vec{p}_{\textup{rel}}+\frac{m}{M}\vec{p}   \big) \frac{d\sigma}{d\Omega} (|\vec{p}_{\textup{rel}}|,\theta)\Big|\\ &\leq 
2^{j}\frac{\eta}{m_{*}}\sup_{\mathbf{p}\in \R_{+}}\sigma_{\textup{tot}}(\mathbf{p})\int_{\R^{3}}d\vec{p}|y(\vec{p})|\, \int_{\R^{3}}d\vec{p}_{\textup{rel}}\,|\vec{p}_{\textup{rel}}| ( |\vec{p}|^{j}+2^{j}|\vec{p}_{\textup{rel}}|^{j})r\big(\frac{m}{m_{*}}\vec{p}_{\textup{rel}}+\frac{m}{M}\vec{p}   \big),
\end{align*}
where the equality uses the integral definition~(\ref{JumpRates}) of the jump rates $\mathcal{J}(\vec{p}_{2},\vec{p}_{1})$ and a change of integration from $\vec{q},\vec{v}$ to variables $\hat{\theta}\in \Omega,\,\vec{p}_{\textup{rel}}\in \R^{3}$, as in the proof of Lem.~\ref{Horses}, where $\theta$ is angle between $\hat{\theta}$ and $\vec{p}_{\textup{rel}}$.  The second inequality simply uses $(a+b)^{j}\leq 2^{j}(a^{j}+b^{j})$.  By performing the Gaussian integration, the above is bounded by a linear combination of $\|\frak{M}^{i}y\|_{1}$ for $i\leq j+1$.  By interpolation $\| [\frak{M}^{j},\mathcal{L}_{0}]y   \|_{1}$ is bounded by a constant multiple $C_{j}$ of $\|y\|_{\frak{M}^{j+1}}$.
          
Using~(\ref{Juxt0}) and~(\ref{Juxt1}), $\|\mathcal{L}_{0}^{n}y\|_{1}$ is bounded by
\begin{align}
\|\mathcal{L}_{0}^{n}y\|_{1} & \leq c\|\mathcal{L}_{0}^{n-1}y\|_{\frak{M}} \nonumber \\ &\leq  c \|\mathcal{L}_{0}^{n-1}y\|_{1}+c \|\mathcal{L}_{0}\,\frak{M}\mathcal{L}_{0}^{n-2}y\|_{1}+c \|[\mathcal{L}_{0},\,\frak{M}]\mathcal{L}_{0}^{n-2}y\|_{1} \nonumber 
\\ & \leq c\|\mathcal{L}_{0}^{n-2}y\|_{\frak{M}} +  (c^{2} +cC_{1} )\|\mathcal{L}_{0}^{n-2}y\|_{\frak{M}^{2}}.
\end{align}
With interpolation to bound $\|\cdot\|_{\frak{M}}$ by $\|\cdot\|_{\frak{M}^{2}}$, I can then proceed by induction until the right side is $\|y\|_{\frak{M}^{n}}$.

 Now I must to show $\|\cdot\|_{\frak{M}^{n} }$ is bounded by a constant multiple of $\|\cdot\|_{\mathcal{L}_{0}^{n} }$.  The result will follow by an induction argument  if I can show  
 \begin{align} \label{Beck}
  \|\,|\vec{p}|^{n}y\|_{1}& \leq C_{n}^{\prime}\|\, |\vec{p}|^{n-1} \mcg_{r} y \|_{1}\nonumber \\ &\leq C_{n}^{\prime}(\|\, |\vec{p}|^{n-1}\mcj \textup{I}_{r}\|\|y\|_{1}  +   \|\,|\vec{p}|^{n-1} \mcl_{0} y \|_{1}), 
  \end{align}
    where $\mathcal{G}_{r}$ and $\mcj \textup{I}_{r}$ are defined as in Part (2), and the first inequality holds for some $C_{n}^{\prime}>0$ for large enough $r$.  The second inequality is the triangle inequality using $\mathcal{G}_{r}=\mathcal{L}_{0}- \mcj \textup{I}_{r}$, and the operator norm    
$ \|\, |\vec{p}|^{n-1}\mcj \textup{I}_{r}\|$ is finite by the calculation below.  Since $|\vec{p}|^{n}\mcj \textup{I}_{r}$ has a positive integral kernel, its operator norm as a map on $L^{1}(\R^{3})$ is bounded by the following
\begin{align*}
  \|\, |\vec{p}|^{n}\mcj \textup{I}_{r}\|&= \sup_{|\vec{p}|\leq r} \int_{\R^{3}}d\vec{p_{0}}\,|\vec{p_{0}}|^{n}\mathcal{J}(\vec{p}_{0},\vec{p})\\ &= \frac{\eta}{m_{*}} \sup_{|\vec{p}|\leq r}\int_{\R^{3}}d\vec{p}_{\textup{rel}}\int_{\Omega}  \,\big|\vec{p}+   \vec{p}_{\textup{rel}}-\hat{\theta} |\vec{p}_{\textup{rel}}| \big|^{n}r\big(\frac{m}{m_{*}}\vec{p}_{\textup{rel}}+\frac{m}{M}\vec{p}   \big) \frac{d\sigma}{d\Omega} (|\vec{p}_{\textup{rel}}|,\theta)\\ &\leq  2^{n}\frac{\eta}{m_{*}}\sup_{\mathbf{p}\in \R_{+}}\sigma_{\textup{tot}}(\mathbf{p}) \sup_{|\vec{p}|\leq r}\int_{\R^{3}}d\vec{p}_{\textup{rel}}  \,\big(|\vec{p}|^{n}+ 2^{n}|\vec{p}_{\textup{rel}}|^{n} \big) r\big(\frac{m}{m_{*}}\vec{p}_{\textup{rel}}+\frac{m}{M}\vec{p}   \big)<\infty, 
\end{align*}
where the second equality is the standard change of variables.  The inequality is from $|a+b|^{n}\leq 2^{n}(|a|^{n}+|b|^{n})$ for $a=\vec{p}$ and $ b =\hat{\theta} |\vec{p}_{\textup{rel}}|- \vec{p}_{\textup{rel}} $ along with the fact $|\hat{\theta}|=1$.  The right side is finite since $|\vec{p}|\leq r$ and the integrand is a Gaussian against a polynomial.

 To show~(\ref{Beck}), I begin with an application of Lem.~\ref{Horses} to get the first inequality below
\begin{align}\label{String}
 \| |\vec{p}|^{n} y \|_{1} & \leq \underline{\mathcal{F}}^{-1}\|\,|\vec{p}|^{n-1} \mce y\|_{1}\leq \underline{\mathcal{F}}^{-1}\| \,|\vec{p}|^{n-1}\mce y\|_{w}\nonumber  \\ &\leq (h_{n}\underline{\mathcal{F}})^{-1}\|\,|\vec{p}|^{n-1} \mcg_{r} y \|_{w} \leq 2(h_{n}\underline{\mathcal{F}})^{-1}\|\,|\vec{p}|^{n-1} \mcg_{r} y \|_{1}.   
 \end{align}
 The second and fourth inequalities follow from $\|\cdot\|_{1}\leq \|\cdot \|_{w}\leq 2\|\cdot\|_{1}$.  The critical input in~(\ref{String}) is the third inequality, which follows from the triangle inequality and similar reasoning as for the proof of~(\ref{Gipper}) to get the first and second inequalities below respectively,
\begin{align} \label{Critical}
\| \,|\vec{p}|^{n-1}\mcg_{r} y \|_{w} &\geq \|\,|\vec{p}|^{n-1} \mce y \|_{w}- \| \,|\vec{p}|^{n-1}\mcj(\textup{I}- \textup{I}_{r})y \|_{w}\nonumber \\ &\geq h_{n}\| \,|\vec{p}|^{n-1}\mce y\|_{w},  
\end{align}
for some $h_{n}>0$.   I then set $C_{n}^{\prime}:=2(h_{n}\underline{\mathcal{F}})^{-1}$.  The use of the norm $\|\cdot\|_{w}$ is only required for the $n=1$ case to get the second inequality of~(\ref{Critical}), and the physical motivation for the second inequality is the same as for~(\ref{Gipper}).

 Using~(\ref{Beck}) inductively, then $\|y\|_{\frak{M}^{n}}$ will be bounded by  a multiple of $\|y\|_{ \mathcal{L}_{0}^{n}}$.

\vspace{.5cm}

\noindent Part (6):\\

To show  $\|\cdot  \|_{\mathcal{L}_{\vec{k}}}$ is smaller than a constant multiple of $\|\cdot\|_{\mathcal{L}_{0}}$,  by Part (3) it is sufficient to prove  $\|\mathcal{L}_{\vec{k} }y\|_{1}\leq c\|y\|_{\frak{M}}$ for some $c>0$. By~(\ref{FiberGen0}),  $\mathcal{L}_{\vec{k}}$ is sum of a Hamiltonian term $-\frac{\ii}{\hbar}h_{\vec{k}}$, a gain term $\mathcal{J}_{\vec{k}}$, and  a loss term $\mathcal{E}_{\vec{k}}$.  The Hamiltonian term is multiplication by     
$$h_{\vec{k}}(\vec{p})=  -\frac{1}{M}\vec{k}\cdot\vec{p} +\big(H_{f}(\vec{p}-2^{-1}\hbar\vec{k})- H_{f}(\vec{p}+2^{-1}\hbar\vec{k})       \big).$$     
The first term is linear in $\vec{p}$ and   $H_{f}(\vec{p})$ has a linear bound in $|\vec{p}|$ by Lem.~\ref{Leftover}.  The loss term $\mathcal{E}_{\vec{p}}$ is multiplication by  
 $$2^{-1}\mce(\vec{p}-2^{-1}\hbar\vec{k})+2^{-1}\mce(\vec{p}+2^{-1}\hbar\vec{k}), $$
which is linearly bounded in $|\vec{p}|$ by Lem.~\ref{Horses}.  For the gain term $\mathcal{J}_{\vec{k}}$, the values $\mcj_{\vec{k}}(\vec{p}+\vec{q}  , \vec{p} )$ are bounded by
\begin{align*}
 \big|\mcj_{\vec{k}}(\vec{p}+\vec{q}  , \vec{p} )\big| &\leq  \int_{(\vec{q})_{\perp}}d\vec{v}\,\big|L_{\vec{q},\vec{v}}(\vec{p}-2^{-1}\hbar\vec{k})\,\overline{L}_{\vec{q},\vec{v}}(\vec{p}+2^{-1}\hbar\vec{k})\big| \\ &\leq 2^{-1}\int_{(\vec{q})_{\perp}}d\vec{v}\,\big(\big|L_{\vec{q},\vec{v}}(\vec{p}-2^{-1}\hbar\vec{k})\big|^{2}+\big|L_{\vec{q},\vec{v}}(\vec{p}+2^{-1}\hbar\vec{k})\big|^{2}\big)  \\ &=2^{-1}\big( \mcj(\vec{p}-2^{-1}\hbar\vec{k}+\vec{q}  , \vec{p}-2^{-1}\hbar\vec{k} )+ \mcj(\vec{p}+2^{-1}\hbar\vec{k}+\vec{q}  , \vec{p}+2^{-1}\hbar\vec{k} )\big).
 \end{align*}
The above implies 
\begin{align*}
\|\mathcal{J}_{\vec{k}}y\|_{1}&=\int_{\R^{3}}d\vec{p}\int_{\R^{3}}d\vec{q}\big|\mathcal{J}(\vec{p}+\vec{q},\vec{p})\big|\,\big|y(\vec{p})\big|   \\ &   \leq 2^{-1} \int_{\R^{3}}d\vec{p} \,\big( \mathcal{E}(\vec{p}-2^{-1}\hbar \vec{k})+  \mathcal{E}(\vec{p}+2^{-1}\hbar \vec{k})\big)|y(\vec{p})| \\ &  \leq \int_{\R^{3}}d\vec{p}\,\big(\overline{\mathcal{E}}+2^{-1}\overline{\mathcal{F}}|\vec{p}-2^{-1}\hbar\vec{k}|+2^{-1}\overline{\mathcal{F}}|\vec{p}+2^{-1}\hbar\vec{k}|     \big)|y(\vec{p})|\\  &\leq (\overline{\mathcal{E}}+2^{-1}\overline{\mathcal{F}}|\vec{k}|)\|y\|_{1}+\overline{\mathcal{F}}\|\frak{M}y\|_{1} ,   
 \end{align*}
 and the right side is smaller than a multiple of $\|y\|_{\frak{M}}$.


\end{proof}

\section{Smoothness of the fiber generators $\mathcal{L}_{\vec{k}}$ }\label{SecDiff}

The main goal of this section is to show the operators $\mathcal{L}_{\vec{k}}$ have three derivatives in $\vec{k}$ that are relatively bounded to $\mathcal{L}_{0}$.  Showing this will require bounding derivatives of the scattering amplitude $(\partial_{\mathbf{z}}^{N} \mathbf{f})(\mathbf{p},\theta)$ for $N\leq 3$ for large and small $\mathbf{p}$.  By the discussion in Sect.~\ref{SecDis},   $\mathbf{f}(\mathbf{p},\theta)$ has a peak  with height proportional to $\mathbf{p}$ occurring in the forward direction $\theta\ll 1$  for $\mathbf{p}\gg \frac{\hbar}{a}$.  Hence, it would not be expected that the supremum $\sup_{\theta\in [0,\pi]} \big|(\partial_{\mathbf{z}}^{N}\mathbf{f})(\mathbf{p},\theta)\big|$ is well-behaved for large $\mathbf{p}$.  However, controllable behavior can be found by integrating the square modulus over the angular variables ($d\Omega=\sin(\theta)d\theta d\phi$)
$$\int_{\Omega} \big| ( \partial_{\mathbf{z}}^{N} \mathbf{f})(\mathbf{p},\theta)\big|^{2}=2\pi\int_{0}^{\pi}d\theta \sin(\theta)\big| ( \partial_{\mathbf{z}}^{N} \mathbf{f})(\mathbf{p},\theta)\big|^{2}. $$ 
The analysis in the proof of the lemma below  boils down to standard relations for Bessel functions and  Legendre polynomials.

\begin{lemma}\label{ScatAmpl}\text{  }
 There is a unitless constant $C>0$ such that for $N=1,\,2,\,3$,
\begin{enumerate} 
\item 
 $$ \sup_{\mathbf{p}\geq \frac{\hbar}{a} }  \int_{\Omega} \big| (\partial_{\mathbf{z}}^{N} \mathbf{f})(\mathbf{p},\theta)\big|^{2}\leq a^{2} (\frac{a}{\hbar})^{2N}C, $$
 
\item 
$$ \sup_{0\leq\mathbf{p}\leq \frac{\hbar}{a} }  \int_{\Omega} \big| (\partial_{\mathbf{z}}^{N} \mathbf{f})(\mathbf{p},\theta)\big|^{2}\leq |\frac{\hbar}{a\mathbf{p}}|^{2\delta_{N,3}}  a^{2} (\frac{a}{\hbar})^{2N}C.  $$

\end{enumerate}

\end{lemma}

\begin{proof}
The partial-wave expansion for $\mathbf{f}(\mathbf{p},\theta)$ is given by
\beq \label{PartialWave}
\mathbf{f}(\mathbf{p},\theta)=\frac{a}{2\ii\kappa}\sum_{\ell =0}^{\infty}(2\ell+1)\big( S_{\ell}(\kappa)  -1\big)\frak{L}_{\ell}(\cos(\theta)),  
\eeq
where $a$ is the radius of the sphere, $\kappa=\frac{a}{\hbar}\mathbf{p}$, $\frak{L}_{\ell}$ are the Legendre polynomials, and the values $ S_{\ell}(\kappa)\in \C$ have modulus one and are equal to 
\beq \label{Trans}
\hspace{2cm}S_{\ell}(\kappa)=- \frac{h_{\ell}^{(2)}(\kappa)}{h_{\ell}^{(1)}(\kappa)  }= \frac{y_{\ell}(\kappa)+j_{\ell}(\kappa)\textup{i}}{y_{\ell}(\kappa)-j_{\ell}(\kappa)\textup{i} }, \quad \quad \hspace{.5cm} \kappa \in \R_{+}, 
\eeq
for spherical Bessel functions $j_{\ell}$, $y_{\ell}$ of the first and second kind respectively. Since the values $S_{\ell}(\kappa)-1$ decay super-exponentially in $\ell$ and  $|\frak{L}_{\ell}(x)|\leq 1$ for $-1\leq x\leq 1$, the series~(\ref{PartialWave}) is uniformly convergent.  It follows the series are also convergent in the norm   $L^{2}\big([0,\pi];\,d\theta \sin(\theta)\big)$, since the domain $[0,\pi]$ is compact. In the consideration of the derivatives of $\mathbf{f}(\mathbf{p},\theta)$, the sums involved will also be uniformly convergent for $\theta\in [0,1]$, since $\partial_{\kappa}^{n_{1}}\big(\frac{S_{\ell}(\kappa)-1}{\kappa}\big)$ decays super-exponentially in $\ell$ and $\sup_{-1\leq x\leq 1}\big|\partial_{x}^{n_{2}}\frak{L}_{\ell}(x)\big|$ increases polynomially for $n_{1},n_{2}\in \mathbb{N}$.

 The identities used in the remainder of the proof related to the special functions $j_{\ell}(\kappa)$, $y_{\ell}(\kappa)$ and $\frak{L}_{\ell}$ can all be found in~\cite{Handbook}.  Some of the facts on the Bessel functions are stated in~\cite{Handbook} for the cylindrical Bessel functions $J_{\ell+\frac{1}{2}}(\kappa)=\big(\frac{2\kappa}{\pi}\big)^{\frac{1}{2}}  j_{\ell}(\kappa)$ and  $Y_{\ell+\frac{1}{2}}(\kappa)=\big(\frac{2\kappa}{\pi}\big)^{\frac{1}{2}}  y_{\ell}(\kappa)$.

\vspace{.3cm}

\noindent Part (1):\\

  Since $\partial_{\mathbf{z}}=\cos(2^{-1}\theta)\partial_{\mathbf{p}}-2\mathbf{p}^{-1}\sin(2^{-1}\theta)\partial_{\theta}$ 
 and the factors $\cos(2^{-1}\theta)$, $\sin(2^{-1}\theta)$ and their derivatives are bounded, I can simply bound 
\beq \label{Mixed} \int_{\Omega}\big| \mathbf{p}^{-u}\partial_{\theta}^{v}\partial_{\mathbf{p}}^{w}\mathbf{f}(\mathbf{p},\theta)\big|^{2}\quad \text{for }\quad  v+w\leq N, \quad  u+w=N,
\eeq
where $u$ can only be non-zero if $v$ is also non-zero.  Since I am considering $\frac{a}{\hbar}\mathbf{p}>1$, I can  take $u=v$.     

  The basic pattern of the analysis will be clear from a treatment of the first derivatives  $  \partial_{\mathbf{p}} \mathbf{f}(\mathbf{p},\theta) $ and $ \mathbf{p}^{-1}\,\partial_{\theta} \mathbf{f}(\mathbf{p},\theta)$ along with a few extra comments on extending the bounds to higher derivatives.  I will focus on the angular derivative first:     
$$ \mathbf{p}^{-1}\,\partial_{\theta} \mathbf{f}(\mathbf{p},\theta)=-\frac{a^{2}}{2\ii\hbar\kappa^{2} }\sum_{\ell=0}^{\infty}(2\ell+1)\big(S_{\ell}(\kappa)-1\big)\frak{L}_{\ell}^\prime\big(\cos(\theta)\big)\sin(\theta).    $$
I have that  $\frak{L}_{\ell}^\prime(\cos(\theta))\sin(\theta)=\frak{L}_{\ell}^{m}(\cos(\theta))$ for $m=1$, where $\frak{L}_{\ell}^{m}$ are the associated Legendre functions.  The family $\frak{L}_{\ell}^{m}$, $\ell \geq m$ is orthogonal for fixed $m$ and have norms given by
\begin{align}\label{LegendOrtho}
\int_{-1}^{1}dy\, |\frak{L}_{\ell}^{m}(y)|^{2}=\frac{2}{2\ell +1}\frac{(\ell +m)!   }{(\ell-m)!  }.    
\end{align}
Thus, 
\beq \label{AngleDer}
 \frac{1}{\mathbf{p}^{2} } \int_{\Omega} \big| \partial_{\theta} \mathbf{f}(\mathbf{p},\theta)\big|^{2}= \frac{\pi a^{4}}{\hbar^{2}\kappa^{4} }\sum_{\ell=1}^{\infty}\ell(\ell+1)(2\ell+1)  |S_{\ell}(\kappa)-1|^{2}\, 
 \eeq
For the terms in the sum with $\ell\leq \lfloor 2\kappa \rfloor $, I use the crude upper bound
\begin{align*} \frac{\pi a^{4}}{\hbar^{2}\kappa^{4} }\sum_{\ell=1}^{\lfloor 2\kappa  \rfloor}\ell(\ell+1)(2\ell+1)  |S_{\ell}(\kappa)-1|^{2} & \leq    \frac{4\pi a^{4}}{\hbar^{2}\kappa^{4} }\sum_{\ell=1}^{\lfloor 2\kappa \rfloor}\ell(\ell+1)(2\ell+1)\\ &< \frac{12\pi a^{4}}{\hbar^{2}}\int_{0}^{2}dy\, y^{3} = \frac{48\pi a^{4}}{\hbar^{2}},   \end{align*} 
where the second inequality uses a Riemann approximation. 

For the sum of terms with $\ell> \lfloor 2\kappa \rfloor $, I require a study of equation~(\ref{Trans}).  From~\cite{Handbook}, I have the asymptotic relations 
\begin{eqnarray}
\hspace{1cm} j_{\ell}\big((\ell+2^{-1})\,\textup{sech}(\alpha)\big) &=& \frac{e^{(\ell+\frac{1}{2})(\tanh(\alpha)-\alpha)}}{2 (\ell+\frac{1}{2}) \big(\textup{sech}(\alpha)\tanh(\alpha)\big)^{\frac{1}{2}}}\Big(1+\mathit{O}(\frac{\log (\ell)^{3}}{\ell})\Big) \quad \quad \alpha>0,\,\ell\gg 1 \nonumber \\ 
y_{\ell}\big((\ell+2^{-1})\,\textup{sech}(\alpha)\big)&=& \frac{e^{(\ell+\frac{1}{2})(\alpha- \tanh(\alpha))}}{ (\ell+\frac{1}{2})\big(\textup{sech}(\alpha)\tanh(\alpha)\big)^{\frac{1}{2}}}\Big(1+\mathit{O}(\frac{\log (\ell)^{3}}{\ell})\Big) \label{Asy}.
\end{eqnarray}

For $\ell > \lfloor 2\kappa  \rfloor $ and $\kappa\gg 1$,    
$$  |S_{\ell}(\kappa)-1|=   2\Big|\frac{j_{\ell}(\kappa)}{y_{\ell}(\kappa)-j_{\ell}(\kappa)\ii }\Big|<  2e^{(2\ell+1)(\tanh(\alpha)-\alpha)}\leq  2e^{-(\ell+\frac{1}{2})  } ,            $$
for $(\ell+2^{-1})\textup{sech}(\alpha)=\kappa$.  In the first inequality above, I introduced an extra factor of $2$ to cover the errors from~(\ref{Asy}), and the second  uses  $\tanh(\alpha)-\alpha<-\frac{1}{2}$ for $\cosh(\alpha)\geq 2$. When $\kappa \gg  1$ and $\ell \gg \kappa=(\ell+\frac{1}{2})\textup{sech}(\alpha)$, then~(\ref{Asy}) becomes
\beq \label{Asy2}
 j_{\ell}(\kappa)\approx 2^{-\frac{1}{2}}\kappa^{-\frac{1}{2}}(2\ell+1)^{-\frac{1}{2}}   \big( \frac{e \kappa}{2\ell+1}   \big)^{\ell+\frac{1}{2}}\quad \text{and}\quad  y_{\ell}(\kappa)\approx  2^{\frac{1}{2}}\kappa^{-\frac{1}{2}}(2\ell+1)^{-\frac{1}{2}} \big( \frac{e \kappa}{2\ell+1}   \big)^{-\ell-\frac{1}{2}}.
 \eeq
Hence, for $\ell > \lfloor 2\kappa \rfloor^{2} $,  I have a sharper upper bound than the one above that is given by
$$  |S_{\ell}(\kappa)-1|\leq  2\big(\frac{e\kappa }{2\ell+1 }   \big)^{2\ell+1}.$$

Continuing the analysis for the sum of terms with $\ell > \lfloor 2\kappa  \rfloor $, then with the estimates above
\begin{multline*}
\frac{\pi a^{4}}{\hbar^{2} \kappa^{4} }\sum_{\ell= \lfloor \kappa \rfloor+1}^{\infty}\ell(\ell+1)(2\ell+1)  |S_{\ell}(\kappa)-1|^{2} \\ \leq \frac{2\pi a^{4}}{\hbar^{2} \kappa^{4} }\sum_{\ell > \lfloor 2\kappa  \rfloor+1}^{\lfloor 2\kappa \rfloor^{2}}\ell(\ell+1)(2\ell+1) e^{-(\ell+\frac{1}{2})  }  +\frac{2\pi a^{4}}{\hbar^{2} \kappa^{4} }\sum_{\ell > \lfloor 2\kappa  \rfloor^{2}+1}^{\infty}\ell(\ell+1)(2\ell+1)  \big(\frac{e\kappa }{2\ell+1   }   \big)^{4\ell+2}.
\end{multline*}
However, the first term will decay exponentially for large $\kappa$, and the second term will decay super-exponentially.  

For the higher-order angular derivatives, $(\mathbf{p}^{-1}\partial_{\theta})^{j}\mathbf{f}(\mathbf{p},\theta)$ for $j=2,3$,  I can follow the same routine except I will need to use a few extra identities for Legendre polynomials at the outset.  For the second derivative, I use that $\partial_{\theta}^{2}$ of $\frak{L}_{\ell}\big(\cos(\theta)\big)$, $\ell\geq 2$ can be expressed as 
\begin{align}\label{SecOrder}
     \partial_{\theta}^{2}\frak{L}_{\ell}\big(\cos(\theta)\big)= \sin(\theta) \frak{L}_{\ell}^{1\,\prime }\big(\cos(\theta)\big)   = 2^{-1}(\ell+1)\ell \frak{L}_{\ell}\big(\cos(\theta)\big)-2^{-1} \frak{L}_{\ell}^{2}\big(\cos(\theta)\big),
\end{align}
where $\frak{L}_{\ell}^{m\,\prime }$ is the derivative of the Legendre polynomial $\frak{L}_{\ell}^{m}$.  In the above, I have used the definition of $\frak{L}_{\ell}^{1\,\prime }$ and the identity 
\begin{align}\label{Incurs}
(1-x^{2})^{\frac{1}{2}}\frak{L}_{\ell}^{m\,\prime}(x)=2^{-1}(\ell+1)(\ell-m+1)\frak{L}_{\ell}^{m-1}(x)-2^{-1}\frak{L}_{\ell}^{m+1}(x), \quad \quad   m+1\geq \ell . 
\end{align}

By~(\ref{SecOrder}) and  the inequality $(\sum_{j=1}^{3}x_{j})^{2}\leq 3\sum_{j=1}^{3}x_{j}^{2} $,
\begin{align}\label{Tri}
 \int_{\Omega} \Big| &(\mathbf{p}^{-1} \partial_{\theta})^{2}\mathbf{f}(\mathbf{p},\theta)\Big|^{2}\nonumber \\ \leq &\frac{3\pi a^{6}}{8\hbar^{4}\kappa^{6}}\int_{0}^{\pi}d\theta \sin(\theta)\Big|\sum_{\ell=2}^{\infty}\ell(\ell+1)(2\ell+1)\big(S_{\ell}(\kappa)-1\big) \frak{L}_{\ell}\big(\cos(\theta)\big) \Big|^{2}\nonumber \\ &+\frac{3\pi a^{6}}{8\hbar^{4}\kappa^{6}}\int_{0}^{\pi}d\theta  \sin(\theta) \Big|\sum_{\ell=2}^{\infty}(2\ell+1)\big(S_{\ell}(\kappa)-1\big) \frak{L}_{\ell}^{2}\big(\cos(\theta)\big) \Big|^{2}+\frac{3\pi^{2} a^{6}}{8\hbar^{4}\kappa^{6}}\big|\big(S_{1}(\kappa)-1\big)\big|^{2}
\end{align}
where the last term of the right side is the $\ell=1$ term from the partial-wave expansion.   To simplify the expressions in~(\ref{Tri}), I can then use that  the  $\frak{L}_{\ell}^{m}\big(\cos(\theta)\big)$ are orthogonal for fixed $m$ and have normalization~(\ref{LegendOrtho}).  For the third term on the right-side of~(\ref{Tri}),     
\begin{align*}
\frac{3\pi a^{6}}{8\hbar^{4}\kappa^{6}}\int_{0}^{\pi}d\theta \sin(&\theta) \Big|\sum_{\ell=2}^{\infty}(2\ell+1)\big(S_{\ell}(\kappa)-1\big) \frak{L}_{\ell}^{2}\big(\cos(\theta)\big) \Big|^{2}\\ &=\frac{3\pi a^{6}}{4\hbar^{4}\kappa^{6}}\sum_{\ell=2}^{\infty}(2\ell+1)(\ell+2)(\ell+1)\ell(\ell-1)  \big|S_{\ell}(\kappa)-1\big|^{2} \\ &\leq \frac{3\pi a^{6}}{\hbar^{4}\kappa^{6}}\sum_{\ell=2}^{\lfloor 2\kappa  \rfloor}(\ell+2)^{5} +\frac{3\pi a^{6}}{2\hbar^{4}\kappa^{6}}\sum_{\ell=\lfloor 2\kappa \rfloor+1}^{\infty}(\ell+2)^{5}  \big|S_{\ell}(\kappa)-1\big|^{2} . 
\end{align*}
I can apply the same analysis as for the first order derivative by giving a separate treatment of the terms with $ \ell\leq \lfloor 2\kappa  \rfloor$, $ \lfloor 2\kappa  \rfloor<\ell \leq  \lfloor 2\kappa  \rfloor^{2} $, and $  \lfloor 2\kappa  \rfloor^{2} <\ell $.  It is crucial the degree of the polynomial in $\ell$ is strictly smaller than the power of $\kappa^{-1}$ appearing in the prefactor. 

 The first term in~(\ref{Tri}) is handled similarly.  The third derivative requires two applications of~(\ref{Incurs}) to write $  (\mathbf{p}^{-1} \partial_{\theta})^{2}\frak{L}_{\ell}\big(\cos(\theta)\big)$ in terms of $\frak{L}_{\ell}^{u}\big(\cos(\theta)   \big)$ for $0\leq u\leq 3$ and is otherwise similar.

Now I bound the first-order radial derivative $  \partial_{\mathbf{p}} \mathbf{f}(\mathbf{p},\theta) $.  Differentiating term by term
\beq
\partial_{\mathbf{p}} \mathbf{f}(\mathbf{p},\theta)=\frac{a^{2}}{2\ii\hbar }\sum_{\ell=0}^{\infty}(2\ell+1)\partial_{\kappa}\big(\frac{S_{\ell}(\kappa)-1}{\kappa}\big)  \frak{L}_{\ell}(\cos(\theta)).  
\eeq
By the orthogonality of the functions $\frak{L}_{\ell}(\cos(\theta))$ and the inequality $(x+y)^{2}\leq 2(x^{2}+y^{2})$, 
\begin{align}\label{Gaza}
 \int_{\Omega} \big| \partial_{\mathbf{p}} \mathbf{f}(\mathbf{p},\theta)\big|^{2}&=  \frac{\pi a^{4}}{\hbar^{2}\kappa^{2}}\sum_{\ell=0}^{\infty} (2\ell+1) \big|-\frac{1}{\kappa}\big(S_{\ell}\big(\kappa\big)-1\big)+\partial_{\kappa}S_{\ell}\big(\kappa\big)    \big|^{2}\nonumber  \\ &\leq \frac{2\pi a^{4}}{\hbar^{2}\kappa^{4}}\sum_{\ell=0}^{\infty} (2\ell+1) \big|S_{\ell}\big(\kappa\big)-1\big|^{2}+\frac{2\pi a^{4}}{\hbar^{2}\kappa^{2}}\sum_{\ell=0}^{\infty} (2\ell+1)\big|\partial_{\kappa} S_{\ell}(\kappa) \big|^{2}. 
 \end{align}
   The first term on the right side is $\frac{2\hbar^{2}}{a^{2}\mathbf{p}^{4}}\sigma_{\textup{tot}}(\mathbf{p})$, which is bounded over the domain $\mathbf{p}\geq 1$,  since $    \sigma_{\textup{tot}}(\mathbf{p})$ is bounded.  For the second term on the right of~(\ref{Gaza}), I can obtain an expression for $\partial_{\kappa}S_{\ell}(\kappa)$ as 
\begin{align} \label{Hankel}
\partial_{\kappa}S_{\ell}(\kappa)& =\frac{2\ii}{\big(y_{\ell}(\kappa)- j_{\ell}(\kappa)\ii    \big)^{2}} \big(j_{\ell}^\prime(\kappa)y_{\ell}(\kappa)-j_{\ell}(\kappa)y_{\ell}^{\prime}(\kappa)  \big)\nonumber \\ &= \frac{-2\ii\kappa^{-2}}{\big(h_{\ell}^{(1)}(\kappa)  \big)^{2}}, 
\end{align}
where $h_{\ell}^{(1)}(\kappa)=  j_{\ell}(\kappa)+y_{\ell}(\kappa)\textup{i} $  and I have used the Wronskian identity $W\big[j_{\ell}(\kappa),\,y_{\ell}(\kappa)\big]=\kappa^{-2}$.   Thus, 
\begin{align}
\frac{2\pi a^{4}}{\hbar^{2} \kappa^{2}}\sum_{\ell=0}^{\infty}(2\ell+1) \big|\partial_{\kappa}S_{\ell}\big(\kappa  \big) \big|^{2}=  \frac{8\pi a^{4}}{\hbar^{2}\kappa^{2} }\sum_{\ell=0}^{\infty}(2\ell+1)\frac{\kappa^{-4} }{m_{\ell }^{4}(\kappa) }, 
\end{align}
where $m_{v}(\kappa)=\big(y_{v}^{2}(\kappa )+j_{v}^{2}(\kappa )\big)^{\frac{1}{2}}$.  Again I give a separate analysis for the terms $\ell> \lfloor 2\kappa \rfloor $ and $\ell\leq \lfloor 2\kappa  \rfloor $.   For  the sum  $\ell\leq \lfloor 2\kappa  \rfloor $, we can use  $m_{\ell}(\kappa)\geq \kappa^{-1} $ for all $\ell \in\mathbb{N}$ and $\kappa\in \R_{+}$  
$$ 
 \frac{8\pi a^{4}}{\hbar^{2}\kappa^{2}}\sum_{\ell=0}^{ \lfloor 2\kappa  \rfloor}(2\ell+1)\frac{\kappa^{-4}}{ m_{\ell}^{4}(\kappa)}    \leq  \frac{8\pi a^{4}}{\hbar^{2}\kappa^{2}}\sum_{\ell=0}^{ \lfloor 2\kappa \rfloor}(2\ell+1)\leq  \frac{32\pi a^{4}}{\hbar^{2}}.  $$
The sum of terms with $\ell> \lfloor 2\kappa  \rfloor $ can be treated with the estimates~(\ref{Asy}) and~(\ref{Asy2}), which I used for the angular derivative $\partial_{\theta} \mathbf{f}(\mathbf{p},\theta)$:  
\begin{multline*}
\frac{8\pi a^{4}}{\hbar^{2}\kappa^{2}}\sum_{\ell> \lfloor 2\kappa \rfloor }^{\infty}(2\ell+1)\frac{\kappa^{-4}}{m_{\ell }^{4}(\kappa)  } \leq  \frac{8\pi a^{4}}{\hbar^{2}\kappa^{6}}\sum_{\ell> \lfloor 2 \kappa \rfloor }^{\infty}\frac{2\ell+1}{y_{\ell }^{4}(\kappa) }\\ \leq \frac{8 \pi a^{4} }{\hbar^{2}\kappa^{2}}\sum_{\ell> \lfloor 2\kappa \rfloor }^{ \lfloor 2\kappa \rfloor^{2} }(2\ell+1)e^{-\frac{1}{2}(\ell+\frac{1}{2})} + \frac{4\pi a^{4} e^{2}}{\hbar^{2}\kappa^{2}}\sum_{\ell> \lfloor 2\kappa \rfloor^{2} }^{\infty} (2\ell+1)\big(\frac{ e \kappa   }{2\ell+1 }    \big)^{4\ell},  
\end{multline*}
where I have doubled the coefficients to cover the error for the asymptotics. The right side is exponentially small for large $\kappa$.

For the higher derivatives $\partial_{p}^{w}\mathbf{f}(\mathbf{p},\theta)$ with $w=2,3$, the coefficients in the partial-wave expansion include $\partial_{\kappa}^{w}\big(\frac{S(\kappa)-1}{\kappa}\big)   $, and it is necessary to find convenient expressions for the derivatives $\partial_{\kappa}^{w}S(\kappa)$.  Using the relation
 \beq \label{HankelDir} 
  \partial_{\kappa}h_{\ell}^{(1)}(\kappa)= \frac{\ell}{\kappa}h_{\ell}^{(1)}(\kappa)-h_{\ell+1}^{(1)}(\kappa) 
 \eeq 
 and~(\ref{Hankel}), then  
\beq \label{SecEss}
 \partial_{\kappa}^{2}S_{\ell}(\kappa)=\partial_{\kappa}^{1}\Big( \frac{-2 \ii }{\kappa^{2}\big(h_{\ell}^{(1)}(\kappa)\big)^{2}} \Big)= \frac{4(\ell+1)\ii }{\kappa^{3}\big(h_{\ell}^{(1)}(\kappa)\big)^{2}}-\frac{4h_{\ell+1}^{(1)}(\kappa)\ii }{\kappa^{2}\big(h_{\ell}^{(1)}(\kappa)\big)^{3}}.  
 \eeq
The absolute value is less than
$$ \big|\partial_{\kappa}^{2}S(\kappa)\big|\leq  \frac{4(\ell+1) }{\kappa^{3} m_{\ell}^{2}(\kappa)}+\frac{4m_{\ell+1}(\kappa) }{\kappa^{2}m_{\ell}^{3}(\kappa)}.   $$
Hence, I have the upper-bound for the integral of $\big| (\partial_{\mathbf{p}}^{2} \mathbf{f})(\mathbf{p},\theta)\big|^{2}$ given by
\begin{multline}\label{NoLand}
\int_{\Omega} \big| \partial_{\mathbf{p}}^{2} \mathbf{f}(\mathbf{p},\theta)\big|^{2} \\ \leq  \frac{12\pi a^{4}}{\hbar^{2} }\sum_{\ell=0}^{\infty}(2\ell+1)\Big(\kappa^{-6} \big|S_{\ell}\big(\kappa\big)-1 \big|^{2}+ \kappa^{-8}(\ell+3)^{2} m_{\ell}^{-4}\big(\kappa\big)+\kappa^{-6}\frac{m_{\ell+1}^{2}(\kappa)}{m_{\ell}^{6}(\kappa)}   \Big), 
\end{multline}
where I used  $\big(x_{1}+x_{2}+x_{3})^{2}\leq 3\big(x_{1}^{2}+x_{2}^{2}+x_{3}^{2}\big)$.  As usual, the terms in the sum can be split  into $\ell\leq \lfloor 2\kappa\rfloor$, $ \lfloor 2\kappa\rfloor < \ell \leq  \lfloor 2\kappa\rfloor^{2} $, and $\ell>\lfloor 2\kappa\rfloor^{2}$.  For the $\ell\leq \lfloor 2\kappa\rfloor$ terms, I use  
\beq \label{EmRatio}
 \sup_{\kappa\geq 1,\,  0\leq \ell\leq \lfloor 2 \kappa\rfloor}\frac{m_{\ell+1}(\kappa)}{m_{\ell}(\kappa)}<\infty.    
 \eeq
  This can be seen by using the polynomial expression
$$ m_{\ell}(\kappa)=\frac{1}{\kappa^{2} }\sum_{n=0}^{\ell}\frac{(2\ell-n)!(2\ell-2n)!      }{n!(\ell-n)!^{2}   }(2\kappa)^{2(n-\ell)}  .$$
Hence, I can again use $m_{\ell}(\kappa)\geq \kappa^{-1}$ for the first  $\lfloor 2\kappa\rfloor$ terms in~(\ref{NoLand}).  The terms with $\ell> \lfloor 2\kappa\rfloor$ can be treated using the approximations~(\ref{Asy}) and (\ref{Asy2}).    

The third-order derivative $\partial_{p}^{3}$ can be bounded by the same techniques as for second-order.  I use~(\ref{HankelDir}) to reduce $\partial_{\kappa}^{3}S_{\ell}(\kappa)$ to an expression in terms of Hankel functions, and then  bound $|\partial_{\kappa}^{3}S_{\ell}(\kappa)|$ using the $m_{\ell}(\kappa)$'s.  An inequality of the form~(\ref{EmRatio}) will be required for the ratio $\frac{m_{\ell+2}(\kappa)}{m_{\ell}(\kappa)}$.

  Finally, the mixed derivatives $v\neq 0$, $w\neq 0$ in~(\ref{Mixed}) require no new observations from those above.

\vspace{.5cm}

\noindent Part (2):\\  

I now consider the case $\kappa=\frac{a}{\hbar}\mathbf{p}< 1$, although  my concern is for the regime $\kappa\ll 1$, since finite values of $\kappa$ are covered by the analysis in Part (1).  The factors $\mathbf{p}^{-1}$ appearing in~(\ref{Mixed}) grow as $ \mathbf{p}\rightarrow 0$, but this turns out only to prevent a constant upper bound for~(\ref{Mixed}) in the cases when $u=3$.   

The expression~(\ref{Mixed}) for $v=1$ and $w=0$ is equal to 
\begin{align} \label{Trisket}
 \int_{\Omega} \big| \mathbf{p}^{-N}\partial_{\mathbf{\theta}} \mathbf{f}(\mathbf{p},\theta)\big|^{2} &= \frac{\pi a^{2N+2}}{\hbar^{2N}\kappa^{2N} }\sum_{\ell=1}^{\infty}\ell(\ell+1)(2\ell+1)  \Big|\frac{S_{\ell}(\kappa)-1}{\kappa}\Big|^{2}\nonumber \\ &=  \frac{8\pi a^{2N+2}}{3\hbar^{2N}}\kappa^{2(2-N)}\big(1+\mathit{O}(\kappa^{2}) \big),  
 \end{align}
where the first equality is~(\ref{AngleDer}), and the approximation is for  $\kappa\ll 1$ and follows from the power expansions
\begin{eqnarray} 
j_{\ell}(\kappa)&= & \frac{\kappa^{\ell} }{(2\ell+1)!!}  \Big(1-   \frac{ 2^{-1}\kappa^{2} }{  1!(2\ell+3)}+\frac{ (2^{-1}\kappa^{2})^{2} }{  2!(2\ell+3)(2\ell+5)  }- \cdots    \Big), \quad \ell \in \mathbb{N} \nonumber \\ \label{SmallKappa} y_{\ell}(\kappa)& =& \frac{ (2\ell-1)!!}{ \kappa^{\ell+1}}  \Big(1-   \frac{ 2^{-1}\kappa^{2} }{  1!(1-2\ell)}+\frac{ (2^{-1}\kappa^{2})^{2} }{  2!(1-2\ell)(3-2\ell)   }- \cdots    \Big). 
\end{eqnarray}
With simpler arguments than for Part~(1), then~(\ref{Trisket}) has a constant upper bound for all $0\leq \mathbf{p}\leq \frac{\hbar}{a}$ when $N=1,2$ and is bounded by a constant multiple of $\mathbf{p}^{-1}$ when $N=3$.  For~(\ref{Mixed}) when $v=2,3$ and $w=0$,  I again use the identities~(\ref{LegendOrtho}) and~(\ref{Incurs})  as in Part (1) along with the approximation~(\ref{SmallKappa}).  When $w\neq 0$, then I write the derivatives  $\partial^{w}_{\kappa}\big(\frac{S_{\ell}(\kappa)-1}{\kappa}\big)$ in terms of Hankel functions as in Part (1) and use~(\ref{SmallKappa}).

\end{proof}

Part (2) of the lemma below is a technical point, which I use in the proof of the convergence of the of the diffusion constant in Thm.~\ref{MainThm}.

\begin{lemma}\label{TwoDerivatives}
Let $\vec{b},\vec{k}\in \R^{3}$.  

\begin{enumerate}
\item
 There is an asymptotic expansion 
$$\mathcal{L}_{\vec{k}+\vec{b}}=\mathcal{L}_{\vec{k}}+\sum_{j=1}^{3}b_{j}\mathcal{L}_{\vec{k}}^{(1)}(j)+\frac{1}{2}\sum_{i,j=1}^{3} b_{i}b_{j} \mathcal{L}_{\vec{k}}^{(2)}(i,j)+|\vec{b}|^{3}\frak{R}_{\vec{k}}(\vec{b}),     $$
for some operators $\mathcal{L}_{\vec{k}}^{(1)}(j)$, $\mathcal{L}_{\vec{k}}^{(2)}(i,j)=\mathcal{L}_{\vec{k}}^{(2)}(j,i)$,  and $\frak{R}_{\vec{k}}(\vec{b})$ that are relatively bounded to $\mathcal{L}_{0}$.  Moreover, the constants for the relative bounds are uniform for $\vec{k},\vec{b}$ in any bounded region.

\item  For all $\vec{b},\vec{k}$ in a bounded region and any $G=\mathcal{L}_{\vec{k}}^{(1)}(j)$, $\mathcal{L}_{\vec{k}}^{(2)}(i,j)$,  or $\frak{R}_{\vec{k}}(\vec{b})$, there is a $C_{n}>0$ such that 
$$\|G  y   \|_{\mathcal{L}_{0}^{n} }\leq C_{n}\| y\|_{\mathcal{L}_{0}^{n+1} }.      $$

\end{enumerate}

\end{lemma}

\begin{proof} \text{   }\\
\noindent Part (1):\\

By Part (5) of Lem.~\ref{RelBnds}, the graph norm of $\mathcal{L}_{0}$ is equivalent to the weighted norm $\|\cdot\|_{\frak{M}}$.  It is hence sufficient to show  the operator valued derivatives of $\mathcal{L}_{\vec{k}}$ up to third order are relatively bounded to multiplication by $|\vec{p}|$ with uniform constants in a bounded region of $\vec{k}$.  The generator $\mathcal{L}_{\vec{k}}$ splits into three terms
$$ \mathcal{L}_{\vec{k}}= -\frac{\ii}{\hbar}h_{\vec{k}}+        \mathcal{J}_{\vec{k}}-\mathcal{E}_{\vec{k}},        $$
which  operate as in~(\ref{FiberGen1}).  These terms will be treated
in Parts (i), (ii), (iii) below for 
$$ (\textup{i}).\,\, h_{\vec{k}}, \quad \quad \quad  (\textup{ii}).\,\,  \mathcal{J}_{\vec{k}}, \quad \quad  \quad (\textup{iii}).\,\,  \mathcal{E}_{\vec{k}}.        $$ 
These operators have domains including $\mathcal{T}_{1}$ by the proof of Part (6) of Prop.~\ref{RelBnds}.  An implication of the analysis here is that the operator derivatives can also be regarded as maps $(\nabla^{\otimes^{N}}_{\vec{k}}\mathcal{L}_{\vec{k}}):\mathcal{T}_{1}\rightarrow L^{1}(\R^{3},\C^{3^{N}}) $.

\vspace{.5cm}

\noindent (i).\hspace{.25cm} 
The operator $\mathcal{E}_{\vec{k}}$ acts as multiplication  with function  
$$\mathcal{E}_{\vec{k}}(\vec{p})=\frac{1}{2}\big(\mathcal{E}(\vec{p}-2^{-1}\hbar \vec{k})+\mathcal{E}(\vec{p}+2^{-1}\hbar \vec{k})\big).$$
The map $\nabla_{\vec{k}}^{\otimes^N}\mathcal{E}_{\vec{k}}:\mathcal{T}_{1}\rightarrow L^{1}(\R^{3},\, \C^{3^{N}})$ acts as multiplication by the vector-valued function
\beq \label{MultFun}
\big(\nabla_{\vec{k}}^{\otimes^N}\mathcal{E}_{\vec{k}}\big)(\vec{p})= \frac{1}{2}\big(\frac{\hbar}{2}\big)^{N}\big( (-1)^{N}(\nabla^{\otimes^N}\mathcal{E})(\vec{p}-2^{-1}\hbar \vec{k})+(\nabla^{\otimes^N}\mathcal{E})(\vec{p}+2^{-1}\hbar \vec{k})\big)  \eeq
Hence, it is sufficient to bound  $|\big(\nabla^{\otimes^N}\mce\big)(\vec{p})|$.  
The vectors $\big(\nabla^{\otimes^N}\mce\big)(\vec{p})$ can be expressed as
\begin{align*}
\big(\nabla^{\otimes^N}\mce\big)(\vec{p})=&  \eta \frac{m}{m_{*}^{2}}\nabla_{\vec{p}}^{\otimes^N}\int_{\R^{3}} d\vec{q}\frac{1}{|\vec{q}|}\int_{(\vec{q})_{\perp} }d\vec{v}\, r\big(\vec{v}+2^{-1}\frac{m}{m_{*}}\vec{q}+\frac{m}{M}\vec{p}_{\parallel \vec{q}} \big)  \,\frac{d\sigma}{d\Omega}\big(|\vec{p}_{\textup{rel} }(\vec{q},\vec{v},\vec{p}_{\perp \vec{q}})|,\,\theta(\vec{q},\vec{v},\vec{p}_{\perp \vec{q}})  \big)\\ =& \frac{\eta}{m_{*}}\nabla_{\vec{p}}^{\otimes^N}\int_{\R^{3}} d\vec{p}_{\textup{rel} }\big|\vec{p}_{\textup{rel}}  \big|\,r\big(\frac{m}{m_{*}}\vec{p}_{\textup{rel}}+\frac{m}{M}\vec{p}  \big)\sigma_{\textup{tot}}(|\vec{p}_{\textup{rel}}|)   \\ =&\frac{\eta}{m_{*}}\Big(\sqrt{\frac{\beta m}{M^{2}}}\Big)^N\int_{\R^{3}} d\vec{p}_{\textup{rel} }\big|\vec{p}_{\textup{rel}}  \big|\,\Big(\frac{  \frac{m}{m_{*}}\vec{p}_{\textup{rel}}+\frac{m}{M}\vec{p}   }{\big|  \frac{m}{m_{*}}\vec{p}_{\textup{rel}}+\frac{m}{M}\vec{p}   \big|   } \Big)^{\otimes^{N}}\\ &\times \frak{H}_{N}\Big(\sqrt{\frac{\beta}{m}}\Big| \frac{m}{m_{*}}\vec{p}_{\textup{rel}}+\frac{m}{M}\vec{p} \Big|   \Big)\,r\big(\frac{m}{m_{*}}\vec{p}_{\textup{rel}}+\frac{m}{M}\vec{p}  \big)\sigma_{\textup{tot}}(|\vec{p}_{\textup{rel}}|) ,
\end{align*}
where $\frak{H}_{N}$ is the $N$th  Hermite polynomial.  The second equality is the same change of integration variables as in the proof of Lem.~\ref{Horses}.

Using the above integral expression, $|\nabla^{\otimes^N}\mce(\vec{p})|$ is bounded by
\begin{multline*}
  \frac{\eta}{m_{*}}\big(\sup_{\mathbf{p}\geq 0}\sigma_{\textup{tot}}(\mathbf{p})\big) \Big(\sqrt{\frac{\beta m}{M^{2}}    }  \Big)^{N} \int_{\R^{3}} d\vec{p}_{\textup{rel}}  \big|\vec{p}_{\textup{rel}}  \big|\,r\big(\frac{m}{m_{*}}\vec{p}_{\textup{rel}}+\frac{m}{M}\vec{p}   \big)\, \Big| \frak{H}_{N}\Big(\sqrt{\frac{\beta}{m}}\Big| \frac{m}{m_{*}}\vec{p}_{\textup{rel}}+\frac{m}{M}\vec{p}  \Big|  \Big)\Big|
\\ \leq \eta\, \frac{m_{*}^{3}}{m^{4}} \, (\frac{m}{\beta})^{\frac{1}{2}}  \big(\sup_{\mathbf{p}\geq 0}\sigma_{\textup{tot}}(\mathbf{p})\big) \Big(\sqrt{\frac{\beta m}{M^{2}}    }  \Big)^{N}\int_{\R^{3}} d\vec{q}\, \big|\vec{q} -(\frac{\beta m}{ M^{2} })^{\frac{1}{2}}\vec{p}    \big|\,  \, \big| \frak{H}_{N}(|\vec{q}|)\big|\,\frac{ e^{-\frac{1}{2}\vec{q}^{2}} }{(2\pi)^{\frac{3}{2}} }, 
\end{multline*}
where I have made a change of integration  $\vec{p}_{\textup{rel}}\rightarrow \vec{q}=\sqrt{\frac{\beta}{m}}\big(\frac{m}{m_{*}}\vec{p}_{\textup{rel}}+\frac{m}{M}\vec{p}    \big)$ for the second inequality. 

 By the triangle inequality, the above is bounded by 
$$  \eta\, \frac{m_{*}^{3}}{m^{4}} \, (\frac{m}{\beta})^{\frac{1}{2}} \big(\sup_{\mathbf{p}\geq 0}\sigma_{\textup{tot}}(\mathbf{p})\big) \Big(\sqrt{\frac{\beta m}{M^{2}}    }  \Big)^{N}\Big( \big\| \,  |\vec{q}|  \frak{H}_{N}(|\vec{q}|)\, e^{-\frac{1}{2}\vec{q}^{2}} \big\|_{1}+|\vec{p}| \sqrt{\frac{\beta m}{M^{2}} }  \big\|  \frak{H}_{N}(|\vec{q}|)\, e^{-\frac{1}{2}\vec{q}^{2}}  \big\|_{1}   \Big),  $$
which gives a linear bound in $|\vec{p}|$ for~(\ref{MultFun}).  It follows that $ \nabla_{\vec{k}}^{\otimes^n}\mathcal{E}_{\vec{k}}$ is relatively bounded to the multiplication operator $|\vec{p}|$.  

\vspace{.5cm}

\noindent (ii).\hspace{.2cm} By the form~(\ref{OffRates}) for $\mathcal{J}_{\vec{k}}:\mathcal{T}_{1}\rightarrow L^{1}(\R^{3})$,
\begin{align}\label{Personne0}
\nabla_{\vec{k}}^{\otimes^N}\mathcal{J}_{\vec{k}}= 4^{-1}\hbar^{2}\int_{\R^{3}} d\vec{q}\,\tau_{\vec{q}}\,w_{\vec{q}}^{(N)},\quad  w_{\vec{q}}^{(N)}(\vec{p})= (\nabla_{\vec{k}}^{\otimes^N}\mathcal{J}_{\vec{k}})(\vec{p}+\vec{q},\vec{p}) , 
\end{align}
where  $\tau_{\vec{q}}$ is the shift operator $(\tau_{\vec{q}}y)(\vec{p})=y(\vec{p}-\vec{q})$ for $y\in L^{1}(\R^{3})$, and $w_{\vec{q}}^{(N)}$ acts as multiplication as above.   
The norm of $\nabla_{\vec{k}}^{\otimes^N}\mathcal{J}_{\vec{k}}$ operating on an integrable function $f$ is smaller than
$$\|\nabla_{\vec{k}}^{\otimes^N}\mathcal{J}_{\vec{k}}y    \|_{1}\leq \int_{\R^{3}}d\vec{p}\,|y(p)|\,\Big(  \int_{\R^{3}}d\vec{q}\,\big| (\nabla_{\vec{k}}^{\otimes^N}\mathcal{J}_{\vec{k}})(\vec{p}+\vec{q},\vec{p})\big|    \Big).     $$
My aim is to show  $\int_{\R^{3}}d\vec{q}\,\big| (\nabla_{\vec{k}}^{\otimes^N}\mathcal{J}_{\vec{k}})(\vec{p}+\vec{q},\vec{p})\big|$ is linearly bounded by $|\vec{p}|$.
     
\begin{align}\label{NthDer}
 \int_{\R^{3}}d\vec{q}\,\big| (\nabla_{\vec{k}}^{\otimes^N}\mathcal{J}_{\vec{k}})(\vec{p}+\vec{q},\vec{p})\big| \leq &\eta\frac{m }{m_{*}^{2}} \int_{\R^{3}} d\vec{q}\frac{1}{|\vec{q}|} \int_{\R^{2}}d\vec{v}\,r\big(\vec{v}+2^{-1}\frac{m}{m_{*}}\vec{q}+\frac{m}{M}\vec{p}_{\parallel \vec{q}}    \big)\,\nonumber \\ &\times \Big|\nabla_{\vec{k}}^{\otimes^N}\Big[ e^{-\frac{\beta m \hbar^{2}}{8M^{2}}\vec{k}_{\parallel \vec{q}}^{2}   }                  \mathbf{f}\big(\mathbf{Q}_{-},\Theta_{-}\big)\, \overline{\mathbf{f}}\big(\mathbf{Q}_{+},\Theta_{+}\big)    \Big] \Big|.
   \end{align}
I will take the gradient $\nabla_{\vec{k}}$ to have components in a basis with directions $\vec{q}$ and two orthogonal vectors.  By the product rule,
\begin{align*}
\nabla_{\vec{k}}^{\otimes^N}\Big[& e^{-\frac{\beta m \hbar^{2} }{8M^{2}}\vec{k}_{\parallel \vec{q}}^{2}   }                 \mathbf{f}\big(\mathbf{Q}_{-},\Theta_{-}\big)\, \overline{\mathbf{f}}\big(\mathbf{Q}_{+},\Theta_{+}\big)    \Big] \\ =& e^{-\frac{\beta m \hbar^{2} }{8M^{2} }\vec{k}_{\parallel \vec{q}}^{2}   }\sum_{N_{1}+N_{2}+N_{3}=N}\frac{N!}{N_{1}!N_{2}!N_{3}!}\Big(\sqrt{\frac{\beta m \hbar^{2}}{4M^{2} }}\Big)^{N_{1}}\frak{H}_{N_{1}}\Big(\sqrt{\frac{\beta m \hbar^{2}}{4M^{2} }}k_{\parallel \vec{q}} \Big) \big(\frac{\vec{q}}{|\vec{q}|}\big)^{\otimes^{N_{1}} }\\  & \times  (-1)^{N_{3}} (\nabla_{\vec{k}}^{\otimes^{N_{2}} }\mathbf{f})\big(\mathbf{Q}_{-},\Theta_{-}\big)\, (\nabla_{\vec{k}}^{\otimes^{N_{3}}} \overline{\mathbf{f}})\big(\mathbf{Q}_{+},\Theta_{+}\big).
\end{align*}

Using the identity~(\ref{ARelate}) and the inequality $2|xy|\leq x^{2}+y^{2}$, 
\begin{align*}
r\big(\vec{v}+\frac{1}{2}\frac{m}{m_{*}}&\vec{q}+\frac{m}{M}\vec{p}_{\parallel \vec{q}}    \big)\big|\nabla_{\vec{k}}^{N}\big[ e^{-\frac{\beta m \hbar^{2}}{8M^{2}}\vec{k}_{\parallel \vec{q}}^{2}   }                  \mathbf{f}\big(\mathbf{Q}_{-},\Theta_{-}\big)\, \overline{\mathbf{f}}\big(\mathbf{Q}_{+},\Theta_{+}\big)    \big]\big|\\  \leq & 2^{-1}e^{-\frac{\beta m \hbar^{2}}{8M^{2}}\vec{k}_{\parallel \vec{q}}^{2}   }                 \sum_{N_{1}+N_{2}+N_{3}=N}\frac{N!}{N_{1}!N_{2}!N_{3}!}\Big(\sqrt{\frac{\beta m \hbar^{2}}{4M^{2}}}\Big)^{N_{1}} \Big|\frak{H}_{N_{1}}\Big(\sqrt{\frac{\beta m \hbar^{2}}{4M^{2}}}\vec{k}_{\parallel \vec{q}} \Big) \Big|\\ &\times \Big( a^{N_{3}-N_{2}}   r\big(\vec{z}_{-} \big)\big|\nabla_{\vec{k}}^{\otimes^{N_{2}}}\mathbf{f}\big(\mathbf{Q}_{-},\Theta_{f}\big) \big|^{2}  + a^{N_{2}-N_{3}} r\big(\vec{z}_{+} \big)\big|\nabla_{\vec{k}}^{\otimes^{N_{3}}}\mathbf{f}\big(\mathbf{Q}_{+},\Theta_{+}\big) \big|^{2}\Big).
\end{align*}
In applying the inequality $2|xy|\leq x^{2}+y^{2}$, I have multiplied and divided by a factor of $a^{\frac{1}{2}(N_{2}-N_{3})}$ to keep the units consistent.   Returning to the integral~(\ref{NthDer}), a single term from the sum over $N_{1}+N_{2}+N_{3}=N$ can be bounded by
\begin{multline}\label{FixedNs}
\frac{\eta}{2}\frac{N!}{N_{1}!N_{2}!N_{3}!}\Big(\sqrt{\frac{\beta m \hbar^{2} }{4M^{2}}}\Big)^{N_{1}} \Big( \sup_{\gamma \in \R_{+}} e^{-\frac{1}{2}\gamma^{2}   }    \frak{H}_{N_{1}}(\gamma)\Big) \frac{m }{m_{*}^{2}} \int_{\R^{3}} d\vec{q}\frac{1}{|\vec{q}|} \int_{\R^{2}}d\vec{v}\\ \times\Big( a^{N_{3}-N_{2}} r\big(\vec{z}_{-} )\big|\nabla_{\vec{k}}^{\otimes^{N_{2}}}\mathbf{f}\big(\mathbf{Q}_{-},\Theta_{-}\big) \big|^{2} +a^{N_{2}-N_{3}} r\big(\vec{z}_{+}\big)\big|\nabla_{\vec{k}}^{\otimes^{N_{3}}}\mathbf{f}\big(\mathbf{Q}_{+},\Theta_{+}\big) \big|^{2} \Big),
\end{multline}
where I have pulled out the supremum for the part of the integrand depending on $\vec{k}_{\parallel q}$.

I will focus on the term $r\big(\vec{z}_{-} )\big|\nabla_{\vec{k}}^{\otimes^{N_{2}}}\mathbf{f}\big(\mathbf{Q}_{-},\Theta_{-}\big) \big|^{2}$.  
It is convenient to switch the integration from  $\int_{\R^{3}} d\vec{q}\frac{1}{|\vec{q}|} \int_{\R^{2}}d\vec{v}$ to $\frac{m_{*} }{m}\int_{\R^{3}}d\vec{p}_{\textup{rel}}\int_{\Omega} $ where 
$$\vec{p}_{\textup{rel}}= \frac{m_{*}}{m}\vec{z}_{-} -\frac{m_{*}}{M}( \vec{p} -2^{-1}\hbar\vec{k} ),\quad \text{so that} \quad |\vec{p}_{\textup{rel}}|=\mathbf{Q}_{-},       $$
and the integration over $\Omega$ reduces to $\int_{\Omega}=2\pi\int_{[0,\pi]}d\Theta_{-}\sin(\Theta_{-})$, since the integrand~(\ref{FixedNs}) does not depend on the azimuthal angle $\phi_{-}$.  The gradient $\nabla_{\vec{k}}\mathbf{f}$ can be written
\begin{align*}
\nabla_{\vec{k}}\mathbf{f}\big(&|\vec{p}_{\textup{rel}}|, \Theta_{-}\big)\\ &=  \frac{\hbar m_{*}}{2M}  \frac{\vec{w}    }{|\vec{w} |    } \Big( \cos(2^{-1}\Theta_{-}) (\partial_{ \mathbf{p}}\mathbf{f})\big(|\vec{p}_{\textup{rel}}|, \Theta_{-}\big) -2 |\vec{p}_{\textup{rel}}|^{-1} \sin(2^{-1}\Theta_{-})(\partial_{ \theta}\mathbf{f})\big(|\vec{p}_{\textup{rel}}|, \Theta_{-}\big)   \Big) 
\end{align*}
for $\vec{w}=\frac{m_{*}}{m}\vec{v}-\frac{m_{*}}{M}(\vec{p}_{\perp q}-2^{-1}\hbar \vec{k}_{\perp \vec{q}}    )    $ and the specific dependence of $\vec{w}$ on $\vec{p}_{\textup{rel}}$, $\Theta_{-}$, $\phi_{-}$ is not important here, since only the norms of the gradients appear in~(\ref{FixedNs}).  More generally,
$$ \nabla_{\vec{k}}^{\otimes^n}\mathbf{f}\big(|\vec{p}_{\textup{rel}}|, \Theta_{-}\big)=   \big(\frac{\hbar m_{*}}{2M}\big)^{n}\Big(\frac{\vec{w}   }{|\vec{w}|    }\Big)^{\otimes^{n}} (\partial^{n}_{\mathbf{z}}     \mathbf{f})(|\vec{p}_{\textup{rel}}|,\Theta_{-}),  $$
where $\partial_{\mathbf{z}}$ is defined as in Lem.~\ref{ScatAmpl}.
 Notice   $\partial_{\mathbf{z}}$ is the derivative with respect to the variable $\mathbf{z}=\mathbf{p}\cos(2^{-1}\theta)$ while holding the variable $\mathbf{y}=\mathbf{p}\sin(2^{-1}\theta) $ fixed.      
By changing integration variables as suggested above  and applying Lem.~\ref{ScatAmpl}, I have the first equality and inequality below 
\begin{align}\label{Slovakia}
\frac{m }{m_{*}^{2}} \int_{\R^{3}} d\vec{q}&\frac{1}{|\vec{q}|} \int_{\R^{2}}d\vec{v}\, r\big(\vec{z}_{-} )\big|\nabla_{\vec{k}}^{\otimes^{N_{2}}}\mathbf{f}\big(\mathbf{Q}_{-},\Theta_{-}\big) \big|^{2}\nonumber  \\ = &\frac{1}{m_{*}} \big(\frac{\hbar\, m_{*}}{2M}\big)^{2N_{2}}\int_{\R^{3}}d\vec{p}_{\textup{rel}}\,|\vec{p}_{\textup{rel}}|\, r\big( \frac{m}{m_{*}}\vec{p}_{rel}+\frac{m}{M}(\vec{p}- 2^{-1}\hbar \vec{k})  \big)\int_{\Omega}|(\partial^{N_{2}}_{\mathbf{z}}\mathbf{f})(|\vec{p}_{\textup{rel}}|,\theta)|^{2}\nonumber \\  \leq &\frac{Ca^{2}}{m_{*}} \big(\frac{a\, m_{*}}{2 M}\big)^{2N_{2}}\int_{\R^{3}}d\vec{p}_{\textup{rel}}\,|\vec{p}_{\textup{rel}}|\, r\big( \frac{m}{m_{*}}\vec{p}_{rel}+\frac{m}{M}(\vec{p}- 2^{-1}\hbar \vec{k})  \big)\Big(1+\big(\frac{\hbar}{a|\vec{p}_{\textup{rel}}| }\big)^{2} \delta_{N,3}   \Big)\nonumber  \\ =& \frac{Ca^{2}m_{*}^{3}}{m^{4}} \big(\frac{a\, m_{*}}{2 M}\big)^{2N_{2}}\int_{\R^{3}}d\vec{w}\,\big|\vec{w}-\frac{m}{M}(\vec{p}-2^{-1}\hbar\vec{k})\big|\, r\big( \vec{w})\nonumber   \\ &+\delta_{N_{2},3} \frac{C\hbar^{2}}{m_{*}} \big(\frac{a^{2} m_{*}}{2 M}\big)^{2N_{2}}\int_{\Omega}\int_{\R_{+} }d\mathbf{p}\,\mathbf{p}\, r\big( \frac{m}{m_{*}}\mathbf{p}\,\widehat{\theta}+\frac{m}{M}(\vec{p}- 2^{-1}\hbar \vec{k})  \big).
\end{align}
where $\widehat{\theta}\in \Omega$ is a unit vector.   The left term on the right side of~(\ref{Slovakia}) can be bounded by the triangle inequality and computations with the Gaussian $r(\vec{w})$ as in the proof of Lem.~\ref{Horses}
$$\frac{Ca^{2}m_{*}^{3}}{m^{4}} \big(\frac{a\, m_{*}}{2 M}\big)^{2N_{2}}\big( 2(\frac{2m}{\pi\beta})^{\frac{1}{2}} +\frac{m}{M}|\vec{p}|+\frac{\hbar m}{2M}|\vec{k}|    \big).   $$
The integrand for $\int_{\Omega}$ in the right-most term of~(\ref{Slovakia}) is maximized when the vector $\mathbf{p}\,\hat{\theta}$ points in the opposite direction as $\vec{p}- 2^{-1}\hbar \vec{k}$, so  
\begin{align*}
\int_{\Omega}\int_{\R_{+} }d\mathbf{p}\,\mathbf{p}\, r\big( \frac{m}{m_{*}}\mathbf{p}\,\hat{\theta}+\frac{m}{M}(\vec{p}- 2^{-1}\hbar \vec{k})  \big) &\leq 4\pi\int_{\R_{+} }d\mathbf{p}\,\mathbf{p}\, r\big( \frac{m}{m_{*}}\mathbf{p}-\frac{m}{M}|\vec{p}- 2^{-1}\hbar \vec{k}|   \big)\\ &\leq  4\pi\int_{\R}d\mathbf{p}\,|\mathbf{p}|\, r\big( \frac{m}{m_{*}}\mathbf{p}-\frac{m}{M}|\vec{p}- 2^{-1}\hbar \vec{k}|   \big)\\ &\leq \frac{2\beta m_{*}^{2} }{m^{3}}\big(\big(\frac{2m}{\pi\beta})^{\frac{1}{2}} +\frac{m}{M}|\vec{p}|+\frac{\hbar m}{2M}|\vec{k}|\big),
\end{align*}
where the second inequality is an extension of the integration domain, and the third inequality follows by the triangle inequality after computing the Gaussian integral.

For constants $C_{N_{1},N_{2},N_{3}}>0$  defined as
\begin{multline*}
C_{N_{1},N_{2},N_{3}}:=\frac{ \eta}{2} \frac{N!}{N_{1}!N_{2}!N_{3}!}\Big(\sqrt{\frac{\beta m \hbar^{2}}{4M^{2}}}\Big)^{N_{1}} \Big( \sup_{\gamma \in \R_{+}} e^{-\frac{1}{2}\gamma^{2}   }    \frak{H}_{N_{1}}(\gamma)\Big) \\ \times\frac{Ca^{2}}{m_{*}}\Big( (\frac{m_{*}}{m})^{4}+\frac{\beta \hbar^{2}m_{*}^{2}}{a^{2}m^{3}}    \Big)  \big( 2(\frac{2m}{\pi\beta})^{\frac{1}{2}} +\frac{m}{M}|\vec{p}|+\frac{\hbar m}{2M}|\vec{k}|    \big)   \Big(a^{N_{3}-N_{2}}\big(\frac{a m_{*}}{2 M}\big)^{2N_{2}}+ a^{N_{2}-N_{3}}\big(\frac{a m_{*}}{2 M}\big)^{2N_{3}}  \Big),
\end{multline*}
then~(\ref{NthDer}) is bounded by
\beq\label{LinBnd}
\sum_{N_{1}+N_{2}+N_{3}=N} C_{N_{1},N_{2},N_{3}}<\infty.  
\eeq

The bound~(\ref{LinBnd}) is linear in $|\vec{p}|$.  It follows  the partial derivatives of $[\Psi]_{\vec{k}}$ up to third-order are uniformly relatively bounded to $|\vec{p}|$ in a bounded region of $\vec{k}\in \R^{3}$.  

\vspace{.5cm}

\noindent (iii).\hspace{.2cm} 
The map $h_{\vec{k}}:\mathcal{T}_{1}\rightarrow L^{1}(\R^{3}) $ acts as multiplication by
$$ h_{\vec{k}}(\vec{p})= H(\vec{p}-2^{-1}\hbar\vec{k})-H(\vec{p}+2^{-1}\hbar\vec{k})      .    $$
Thus, $\nabla_{\vec{k}}^{\otimes^N} h_{\vec{k}}$ is a multiplication operator with function 
\begin{multline*}
 \big(\frac{\hbar}{2}\big)^{N}\big( (-1)^{N}(\nabla^{\otimes^N}H)(\vec{p}-2^{-1}\hbar\vec{k})-(\nabla^{\otimes^N}H)(\vec{p}+2^{-1}\hbar\vec{k})      \big)\\ = -\ii \delta_{N,1}\frac{\hbar}{M}\vec{p}+  \ii\big(\frac{\hbar}{2}\big)^{N}\big( (-1)^{N}(\nabla^{\otimes^N}H_{f})(\vec{p}-2^{-1}\hbar\vec{k})-(\nabla^{\otimes^N}H_{f})(\vec{p}+2^{-1}\hbar\vec{k})      \big),         
 \end{multline*}
since $H(\vec{p})=\frac{1}{2M}\vec{p}^{2}+H_{f}(\vec{p}) $ where $H_{f}(\vec{p})$ is defined in~(\ref{AltDisp}).  By a change of variables, I can write   
$$ H_{f}(\vec{p})= -\frac{2\pi \eta  \hbar^{2}m^{3}}{ m_{*}^{4} } \int_{\R^{3}}d\vec{p}_{\textup{rel}}\, r\big(\frac{m}{m_{*}}\vec{p}_{\textup{rel}}+\frac{m}{M}\vec{p}  \big)\,\textup{Re}\big[\mathbf{f}\big(|\vec{p}_{\textup{rel}}|,\,0 )\big].  $$
The derivative has the form 
\begin{multline*}
  (\nabla^{\otimes^N}H_{f})(\vec{p})= -\frac{2\pi \eta \hbar^{2} m^{3} }{ m_{*}^{4}}\Big(\sqrt{\frac{\beta m}{M^{2}} }\Big)^{N}  \int_{\R^{3}}d\vec{p}_{\textup{rel}} \,\Big(\frac{\frac{m }{m_{*}}\vec{p}_{\textup{rel}}+\frac{m }{M}\vec{p}}{\big|\frac{m }{m_{*}}\vec{p}_{\textup{rel}}+\frac{m }{M}\vec{p}    \big|   }\Big)^{\otimes^{N}} H_{N}\Big(\sqrt{\frac{\beta }{m}} \Big|\frac{m }{m_{*}}\vec{p}_{\textup{rel}}+\frac{m }{M}\vec{p} \Big|    \Big) \\ \times r\big(\frac{m}{m_{*}}\vec{p}_{\textup{rel}}+\frac{m}{M}\vec{p}  \big)\,\textup{Re}\big[\mathbf{f}\big(|\vec{p}_{\textup{rel}}|,\,0 ) \big].  
  \end{multline*}

By the same change of variables as for (i), the term $\big|(\nabla^{\otimes^{N}}H_{f})(\vec{p})\big|$ is bounded by
\begin{align}\label{Samples}
\frac{2\pi \eta\hbar^{2} m^{3}}{ m_{*}^{4}}&\Big(\sqrt{\frac{\beta m}{M^{2}} }\Big)^{N}  \int_{\R^{3}}d\vec{p}_{\textup{rel}} \,\Big| \frak{H}_{N}\Big(\sqrt{\frac{\beta }{m}} \Big|\frac{m }{m_{*}}\vec{p}_{\textup{rel}}+\frac{m }{M}\vec{p} \Big|    \Big)\Big| \,r\big(\frac{m}{m_{*}}\vec{p}_{\textup{rel}}+\frac{m}{M}\vec{p}  \big)\,\big|\textup{Re}\big[\mathbf{f}\big(|\vec{p}_{rel}|,\,0 ) \big] \big|\nonumber  \\ &\leq \frac{2\pi\hbar^{2}\eta}{ m_{*} }  \Big(\sqrt{\frac{\beta m}{M^{2}}    }  \Big)^{N}\int_{\R^{3}} d\vec{q}\, \big|\textup{Re}\big[\mathbf{f}\big(\big|(\frac{m}{\beta})^{\frac{1}{2}}\vec{q}-\frac{m}{M}\vec{p}     \big|,\,0 \big) \big]\big|\,\big| \frak{H}_{N}(|\vec{q}|)\big|\,\frac{ e^{-\frac{1}{2}\vec{q}^{2}} }{(2\pi)^{\frac{3}{2}} }\nonumber  \\ & \leq \frac{2\pi c \hbar^{2}\eta}{ m_{*} }  \Big(\sqrt{\frac{\beta m}{M^{2}}    }  \Big)^{N}\int_{\R^{3}} d\vec{q}\, \big(\frac{\hbar}{a}+ (\frac{m}{\beta})^{\frac{1}{2}}|\vec{q}|+\frac{m}{M}|\vec{p}| \big)\,\big| \frak{H}_{N}(|\vec{q}|)\big|\,\frac{ e^{-\frac{1}{2}\vec{q}^{2}} }{(2\pi)^{\frac{3}{2}} },
\end{align}
where the last inequality uses  $|\mathbf{f}\big(\mathbf{p},\,0 \big)|\leq c\big(a^{-1}\hbar+\mathbf{p}\big)$ for some $c>0$, which is shown in the proof of Lem.~\ref{Leftover}. However, the functions $ \big| \frak{H}_{N}(|\vec{q}|)\big|\,\frac{ e^{-\frac{1}{2}\vec{q}^{2}} }{(2\pi)^{\frac{3}{2}} }$ and $|\vec{q}| \big| \frak{H}_{N}(|\vec{q}|)\big|\,\frac{ e^{-\frac{1}{2}\vec{q}^{2}} }{(2\pi)^{\frac{3}{2}} }$ are in $L^{1}(\R^{3})$ as used in Part (i), so~(\ref{Samples}) gives a  linear upper bound for $\big|(\nabla^{\otimes^{N}}H_{f})(\vec{p})\big|$ in terms of $|\vec{p}|$.  The operator $\nabla_{\vec{k}}^{\otimes^N}h_{\vec{k}}$ is therefore  relatively bounded to multiplication by $|\vec{p}|$.  \vspace{.6cm}

\noindent Part (2):\\

By Part (4) of Prop.~\ref{RelBnds}, the norms $\|\cdot   \|_{\mathcal{L}_{0}^{n} }$ and $\| \cdot  \|_{\frak{M}^{n} }$ are equivalent, and it is sufficient to prove $\|G  y   \|_{\frak{M}^{n} }\leq C\| y\|_{\frak{M}^{n+1} } $ for $G =\nabla_{\vec{k}}^{\otimes^{N}} \mathcal{L}_{\vec{k}}$, $N=1,\,2,\,3$ and $\vec{k}$ in a bounded region. This will be convenient,  since $\mathcal{M}^{n}$ commutes with the parts of  $\nabla_{\vec{k}}^{\otimes^{N}} \mathcal{L}_{\vec{k}}$ corresponding to the loss and Hamiltonian parts of $\mathcal{L}_{\vec{k}}$ (as in (i) and (iii) of Part (1)).  By the triangle inequality,     
\begin{align*}
\|G  y   \|_{\frak{M}^{n} } & \leq  \|  Gy   \|_{1 }+\| G  \frak{M}^{n}y     \|_{1}+ \|\,[\frak{M}^{n}, G ]y   \|_{1}\\ &\leq c\|y\|_{ \frak{M}^{n+1} }+ \|\,[\frak{M}^{n}, G ]y   \|_{1}  ,
\end{align*}
where $c>0$ exists by the results of Part (1) and the equivalence of $\|\cdot\|_{\mathcal{L}_{0}}$ and $\|\cdot\|_{\frak{M}}$.

Since $\frak{M}^{n}$ commutes with multiplication operators, $[\frak{M}^{n}, \nabla_{\vec{k}}^{\otimes^{N}} \mathcal{L}_{\vec{k}}  ]=[\frak{M}^{n}, \nabla_{\vec{k}}^{\otimes^{N}} \mathcal{J}_{\vec{k}}  ]$.  By~(\ref{Personne0}), the commutation $[\frak{M}^{n}, \nabla_{\vec{k}}^{\otimes^{N}} \mathcal{J}_{\vec{k}}  ]$ can be written as   
\begin{align*}
[\frak{M}^{n}, \nabla_{\vec{k}}^{\otimes^{N}} \mathcal{J}_{\vec{k}}  ] = 4^{-1}\hbar^{2}\int_{\R^{3}} d\vec{q}\,\tau_{\vec{q}}\,z_{\vec{q}}^{(N)}, 
\end{align*}
where $z_{\vec{q}}^{(N)}$ is multiplication by  
$$z_{\vec{q}}^{(N)}(\vec{p}) =(|\vec{p}+\vec{q}|^{n}-|\vec{p}|^{n}) (\nabla_{\vec{k}}^{\otimes^N}\mathcal{J}_{\vec{k}})(\vec{p}+\vec{q},\vec{p}).$$
Again I have
$$\|[\frak{M}^{n}, \nabla_{\vec{k}}^{\otimes^{N}} \mathcal{J}_{\vec{k}}  ]  y  \|_{1} \leq  \int_{\R^{3}}d\vec{p}\,|y(\vec{p})|\int_{\R^{3}}d\vec{q}\,\big| z_{\vec{q}}^{(N)}(\vec{p})   \big|.     $$
The value $\big| z_{\vec{q}}^{(N)}(\vec{p})   \big|$ is smaller than a product of 
$\big||\vec{p}+\vec{q}|^{n}-|\vec{p}|^{n}\big|$ and $\big|\nabla_{\vec{k}}^{\otimes^{N}} \mathcal{J}_{\vec{k}}(\vec{p}+\vec{q},\vec{p})   \big|$, and by  (ii) of Part (1), 
\begin{multline*}
\big|\nabla_{\vec{k}}^{\otimes^{N}} \mathcal{J}_{\vec{k}}(\vec{p}+\vec{q},\vec{p})   \big| \leq \sum_{N_{1}+N_{2}+N_{3}=N}\frac{\eta}{2}\frac{N!}{N_{1}!N_{2}!N_{3}!}\Big(\sqrt{\frac{\beta m \hbar^{2} }{4M^{2}}}\Big)^{N_{1}} \Big( \sup_{\gamma \in \R_{+}} e^{-\frac{1}{2}\gamma^{2}   }    \frak{H}_{N_{1}}(\gamma)\Big)\\ \times \frac{m }{m_{*}^{2}}\frac{1}{|\vec{q}|} \int_{\R^{2}}d\vec{v} \Big( a^{N_{3}-N_{2}} r\big(\vec{z}_{-} )\big|\nabla_{\vec{k}}^{\otimes^{N_{2}}}\mathbf{f}\big(\mathbf{Q}_{-},\Theta_{-}\big) \big|^{2} +a^{N_{2}-N_{3}} r\big(\vec{z}_{+}\big)\big|\nabla_{\vec{k}}^{\otimes^{N_{3}}}\mathbf{f}\big(\mathbf{Q}_{+},\Theta_{+}\big) \big|^{2} \Big).  
\end{multline*}
For factor $\big||\vec{p}+\vec{q}|^{n}-|\vec{p}|^{n}\big|$ has the brute upper bound $2^{n}(|\vec{p}|^{n}+|\vec{q}|^{n})$.  Applying the standard change of coordinates from $\vec{q},\vec{v}$ to $\vec{p}_{\textup{rel}},\hat{\theta}$ and following the reasoning in (ii) of Part (1), then $\int_{\R^{3}}d\vec{q}\,\big| z_{\vec{q}}^{(N)}(\vec{p})   \big|$ can be bounded a degree $n+1$ polynomial in $|\vec{p}|$.  By interpolation  $\|[\frak{M}^{n}, \nabla_{\vec{k}}^{\otimes^{N}} \mathcal{J}_{\vec{k}}  ]y\|_{1}$ is bounded by a multiple of $\|y\|_{\mathcal{M}^{n+1}}$.

\end{proof}

\section{The diffusion constant}\label{SecDiffConstant}

The following proposition relates the diffusion constant $D=D_{\textup{kin}}+D_{\textup{jps}}$ to the second derivative for an eigenvalue for $\mathcal{L}_{\vec{k}}$.  Part (3) is another another technical point required for the proof of the convergence stated for  the diffusion constant in Thm.~\ref{MainThm}.

\begin{proposition}\label{Perturbation}

Let $\frak{g}>0$ be defined as in Prop.~\ref{RelBnds}.  There is an $r>0$ such that for all $|\vec{k}|\leq r$,  the following statements hold: 

\begin{enumerate}

\item The generator $\mathcal{L}_{\vec{k}}$ has an isolated eigenvalue $\epsilon(\vec{k})\in \R_{-}  +\ii \,\mathbb{R} $ with  $$\hspace{3cm}  \epsilon(\vec{k})= - \frac{1}{2}D|\vec{k}|^{2} +   \mathit{O}(|\vec{k}|^{3}).$$   

\item The spectral projection $ \mathbf{P}_{\vec{k} }$ corresponding to $\epsilon(\vec{k})$ is non-degenerate and has a Taylor expansion
$$ \mathbf{P}_{\vec{k}}= \mathbf{P}_{0}+\sum_{i}k_{i}(\partial_{i}  \mathbf{P}_{\vec{v}}|_{\vec{v}=0 })+ \sum_{i}k_{i}k_{j}(\partial_{i}\partial_{j}  \mathbf{P}_{\vec{v}}|_{\vec{v}=0 })+ |\vec{k}|^{3}\frak{R}^{\prime}(\vec{k}),   $$ 
where $\partial_{i}  \mathbf{P}_{\vec{k}}|_{\vec{k}=0 }$, $\partial_{i}\partial_{j}  \mathbf{P}_{\vec{k}}|_{\vec{k}=0 }$ and $\frak{R}^{\prime}(\vec{k})$ are uniformly bounded 
with respect to the operator norm  for linear maps on $L^{1}(\R^{3})$.

\item  For $n\geq 1$, there is a $C_{n}>0$ such that for any  $G=\partial_{i}  \mathbf{P}_{\vec{v}}|_{\vec{v}=0 }$, $\partial_{i}\partial_{j} \mathbf{P}_{\vec{v}}|_{\vec{v}=0 }$,  or $\frak{R}^{\prime}(\vec{k})$
$$\|G  f   \|_{\mathcal{L}_{0}^{n} }\leq C_{n}\| f\|_{\mathcal{L}_{0}^{n-1} }.      $$

\item The spectrum of $\mathcal{L}_{\vec{k}}$ satisfies $\Sigma(\mathcal{L}_{\vec{k}})-\{ \epsilon(\vec{k})  \}\subset (-\infty, -2^{-1}\frak{g}]+\ii\,\mathbb{R}   $. 

\item  The semigroup $e^{t\mathcal{L}_{\vec{k}}' }$ with generator $\mathcal{L}_{\vec{k}}'= \mathcal{L}_{\vec{k}}+ \frac{\frak{g}}{2}\big( \textup{I}-\mathbf{P}_{\vec{k}}\big)$ (having domain $\textup{D}(\mathcal{L}_{\vec{k}}') =\textup{D}(\mathcal{L}_{\vec{k}})$) satisfies
$$  \big\| e^{t\mathcal{L}_{\vec{k}}' } \big\| \leq C        $$
for some $C>0$ and all $t\in \R_{+}$.  Also,  $\big\| e^{t\mathcal{L}_{\vec{k}} } \big( \textup{I}-\mathbf{P}_{\vec{k}}\big)\big\| \leq C e^{-t\frac{\frak{g}}{2} }  $.

\end{enumerate}

\end{proposition}

\begin{proof}\text{ }\\
\noindent Part (1):\\

By the second-order expansion in~(\ref{TwoDerivatives}) with operator derives relatively bounded to $\mathcal{L}_{0}$, perturbation theory~\cite{Kato} guarantees  there is an isolated eigenvalue $\epsilon (\vec{k})$ near zero with expansion
$$\epsilon(\vec{k})= \vec{k}\cdot\mathbf{a} +2^{-1}\vec{k}^{\otimes^{2}}\cdot(\mathbf{b}-\mathbf{c})+\mathit{O}(|\vec{k}|^{-3}),       $$
where $\mathbf{a}\in \C^{3}$, $\mathbf{b},\mathbf{c}\in \C^{3}\otimes \C^{3}$ are of the form
$$ \mathbf{a}_{j}=\langle 1|  \mathcal{L}_{0}^{(1)}(j)  |\nu_{\infty}\rangle,\quad    \mathbf{b}_{i,j}=\langle 1|  \mathcal{L}_{0}^{(2)}(i,j)  |\nu_{\infty}\rangle,\quad  \mathbf{c}_{i,j}=\langle 1|  \mathcal{L}_{0}^{(1)}(i) S \mathcal{L}_{0}^{(1)}(j)   |\nu_{\infty}\rangle,  $$
where $S$ is the reduced resolvent of $\mathcal{L}_{0}$.   Note that  $\langle 1|  A |\nu_{\infty}\rangle =\int_{\R^{3}}d\vec{p}\,  (A\nu_{\infty})(\vec{p})  $ for an operator $A$ on $L^{1}(\R^{3})$.  I claim  $\mathbf{a}=0$, $\mathbf{b}_{i,j}=-\delta_{i,j}D_{\textup{jps}}$, and  $\mathbf{c}_{i,j}=\delta_{i,j}D_{\textup{kin}}$. 

I begin by finding expressions for the first two derivatives $\mathcal{L}_{\vec{k}}$ at $\vec{k}=0$.  The derivatives are given by  
$$\hspace{2cm}(\nabla^{\otimes^N} \mathcal{L}_{\vec{r}})\big|_{\vec{r}=0}= -i\hbar^{-1}(\nabla^{\otimes^N} h_{\vec{r}})\big|_{\vec{r}=0}+(\nabla^{\otimes^N}\mathcal{J}_{\vec{r}})\big|_{\vec{r}=0}-(\nabla^{\otimes^N} \mathcal{E}_{\vec{r}})\big|_{\vec{r}=0},\quad \quad N=1,\,2,             $$       
where the terms on the right are discussed presently.  As a result of Lem.~\ref{TwoDerivatives}, the operators above map $\mathcal{T}_{1}$ into  $L^{1}(\R^{3}, \C^{3^{N}} )$.

The operator $-(\nabla h_{\vec{r}})|_{\vec{r}=0}$ acts as a vector of multiplication operators with multiplication function $  \hbar\,(\nabla H)(\vec{p}) $.  The terms  $(\nabla^{\otimes^2}h_{\vec{r}})|_{\vec{r}=0}$ and  $(\nabla\mathcal{E}_{\vec{r}})|_{\vec{r}=0}$ are both zero due to symmetry around $\vec{r}=0$.  To write $(\nabla^{\otimes^2} \mathcal{E}_{\vec{r}})|_{\vec{r}=0}$, $(\nabla\mathcal{J}_{\vec{r}})|_{\vec{r}=0}$, and $(\nabla^{\otimes^{2}}\mathcal{J}_{\vec{r}})|_{\vec{r}=0}$, I will make use of the functions $\textup{T}_{\vec{q}}(\vec{p}_{1},\vec{p}_{2})$ in the form~(\ref{Ems}).      
The operator $(\nabla^{\otimes^2} \mathcal{E}_{\vec{r}})\big|_{\vec{r}=0} $ acts as a vector of multiplication operators with multiplication functions
\begin{align} \label{Zijn}
(\nabla^{\otimes^2}\mathcal{E}_{\vec{r}})(\vec{p})\big|_{\vec{r}=0}   = 4^{-1}\hbar^{2} \int_{\R^{3}}d\vec{q}\,((\nabla_{1}+\nabla_{2})^{\otimes^{2}}\textup{T}_{\vec{q}})(\vec{p},\vec{p}),  
\end{align}
where $\nabla_{1}$ and $\nabla_{2}$ are the gradients with respect to the first and second arguments of $\textup{T}_{\vec{q}}(\vec{p}_{1},\vec{p}_{2})$. The operators $\nabla\mathcal{J}_{\vec{r}}|_{\vec{r}=0}$ and $ \nabla^{\otimes^2}\mathcal{J}_{\vec{r}}|_{\vec{r}=0}$ have the form
\begin{eqnarray}\label{Personne2}
(\nabla\mathcal{J}_{\vec{r}})\big|_{\vec{r}=0}&=& 2^{-1}\hbar\int_{\R^{3}} d\vec{q}\,\tau_{\vec{q}}\,w_{\vec{q}}^{(1)},\\
(\nabla^{\otimes^2}\mathcal{J}_{\vec{r}})\big|_{\vec{r}=0}&=& 4^{-1}\hbar^{2}\int_{\R^{3}} d\vec{q}\,\tau_{\vec{q}}\,w_{\vec{q}}^{(2)}\label{Personne},  
\end{eqnarray}
where  $\tau_{\vec{q}}$ acts as a shift $(\tau_{\vec{q}}y)(\vec{p})=y(\vec{p}-\vec{q})$, $y\in L^{1}(\R^{3})$ and $w_{\vec{q}}^{(1)}$ and
$w_{\vec{q}}^{(2)}$ are vectors of multiplication operators with  $$w_{\vec{q}}^{(1)}(\vec{p}):= -((\nabla_{1}-\nabla_{2})\textup{T}_{\vec{q}})(\vec{p},\vec{p})\quad \quad \text{and}\quad  \quad w_{\vec{q}}^{(2)}(\vec{p}):= ((\nabla_{1}-\nabla_{2})^{\otimes^2}\textup{T}_{\vec{q}})(\vec{p},\vec{p}).$$

Now I am almost ready to evaluate $\mathbf{a}$, $\mathbf{b}$, and $\mathbf{c}$.  As a preliminary observation,  there are rotational symmetries 
\beq \label{Rot}\nu_{\infty}(\textup{R}\vec{p})=\nu_{\infty}(\vec{p}),\quad H(\textup{R}\vec{p})=H(\vec{p}),\quad \textup{T}_{\textup{R}\vec{q}}(\textup{R}\vec{p}_{1},\textup{R}\vec{p}_{2})= \textup{T}_{\vec{q}}(\vec{p}_{1},\vec{p}_{2}),\quad \quad \textup{R}\in SO_{3}. \eeq
It follows that $\mathbf{a}$ must be zero.  Moreover, the tensors $\mathbf{b}$ and $\mathbf{c}$ must be invariant under the operation of $\textup{R}\otimes \textup{R}$ for all rotations $\textup{R}$, and therefore   $\mathbf{b}$, $\mathbf{c}$ are constant multiples of the identity tensor.

 With~(\ref{Zijn}) and~(\ref{Personne}), the value for $\mathbf{b}$ can be written  
\begin{align*}
\mathbf{b}&=\big\langle 1 \big|
\big(\nabla^{\otimes^2}\mathcal{J}_{\vec{r}}\big|_{\vec{r}=0}-\nabla^{\otimes^2} \mathcal{E}_{\vec{r}}\big|_{\vec{r}=0}\big)\nu_{\infty}\big\rangle\\ &= -\hbar^{2}\int_{\R^{3}}d\vec{p}\int_{\R^{3}} d\vec{q}\,\nu_{\infty}(\vec{p}) (\nabla_{1}\otimes_{s}\nabla_{2}\textup{T}_{\vec{q}})(\vec{p},\vec{p}),
\end{align*}
where $\otimes_{s}$ is the symmetrized tensor product. By the rotational symmetry mentioned above, I can write
$$ \big[\big\langle 1 \big|
\big(\nabla^{\otimes^2}\mathcal{J}_{\vec{r}}\big|_{\vec{r}=0}-\nabla^{\otimes^2} \mathcal{E}_{r}\big|_{\vec{r}=0}\big)\nu_{\infty}\big\rangle \big]_{i,j}= -\delta_{i,j} 3^{-1}\hbar^{2}\int_{\R^{3}}d\vec{p}\int_{\R^{3}} d\vec{q}\,\nu_{\infty}(\vec{p})(\widetilde{\Delta}\textup{T}_{\vec{q}})(\vec{p},\vec{p}),  $$
where $\widetilde{\Delta}=\sum_{j=1,2,3} \partial_{1,j}\partial_{2.j}$, and $\partial_{i,j}$ is the derivative for $i$th component in the $e_{j}$ direction for some orthonormal basis $(e_{1},e_{2},e_{3})$.  The definition of $\widetilde{\Delta}$ is invariant of the basis used.  The expression $(\widetilde{\Delta}\textup{T}_{\vec{q}})(\vec{p},\vec{p})$ is equal to 
\begin{align}\label{Aqui}
(\widetilde{\Delta}\textup{T}_{\vec{q}})(\vec{p},\vec{p})=&\int_{(\vec{q})_{\perp} }d\vec{v}\sum_{j=1}^{3}\big|\partial_{j}L_{\vec{q},\vec{v}}(\vec{p})\big|^{2}\nonumber \\  = &\eta \frac{m}{m_{*}^{2}}\frac{1}{|\vec{q}|}\int_{(\vec{q})_{\perp} }d\vec{v}\, r\big(\vec{v}+2^{-1}\frac{m}{m_{*}}\vec{q}+\frac{m}{M}\vec{p}_{\parallel \vec{q}  }    \big)\nonumber \\ & \times \Big(  \frac{\beta^{2}}{4M^{2}}\big| \vec{v}+2^{-1}\frac{m}{m_{*}}\vec{q}+\frac{m}{M}\vec{p}_{\parallel \vec{q}} \big|^{2}\,\big| \mathbf{f}\big(|\vec{p}_{\textup{rel}}|, \theta   \big)\big|^{2}+ \frac{m_{*}^{2}}{M^{2}}\big|\partial_{\mathbf{z}}\mathbf{f}\big(|\vec{p}_{\textup{rel}}|, \theta   \big)\big|^{2}     \Big) ,
 \end{align}
where $\vec{p}_{\textup{rel}} = \frac{m_{*}}{m}\vec{v}-\frac{m_{*}}{M}\vec{P}_{\perp \vec{q}}+2^{-1}\vec{q}$  and $2\tan(2^{-1}\Theta)=\big|\vec{q}  \big|\,\big|\frac{m_{*}}{m}\vec{v}-\frac{m_{*}}{M}\vec{p}_{\perp q} \big|^{-1}$.  In the above, the basis vectors $(e_{1},e_{2},e_{3})$ for $\tilde{\Delta}$ have been chosen to depend on $\vec{q},\vec{v},\vec{p}$ as $e_{1}=\frac{\vec{q}}{|\vec{q}|}$, $e_{2}=\frac{\frac{m_{*}}{m}\vec{v}-\frac{m_{*}}{M}\vec{p}_{\perp q}}{|\frac{m_{*}}{m}\vec{v}-\frac{m_{*}}{M}\vec{p}_{\perp q}|}$, and with $e_{3}$ as one of the two normalized vectors orthogonal to $e_{1}$ and $e_{2}$.  The derivative $\partial_{3}L_{\vec{q},\vec{v}}(\vec{p})$ with respect to the third direction $e_{3}$ is equal to zero.

By using equality~(\ref{Aqui}) and by shifting to  center-of-mass coordinates for the integration variables 
\begin{align*}
\int_{\R^{3}}d\vec{p}&\int_{\R^{3}} d\vec{q}\,\nu_{\infty}(\vec{p})(\widetilde{\Delta}\textup{T}_{\vec{q}})(\vec{p},\vec{p}) \\  = & \frac{\eta}{m_{*}}\int_{\R^{3}}d\vec{p}_{\textup{cm}}\frac{ e^{-\frac{\beta}{2(m+M)} \vec{p}_{\textup{cm}}^{2}}   }{\big(2\pi(m+M) \beta^{-1}   \big)^{\frac{3}{2}}   }\int_{\R^{3}}d\vec{p}_{\textup{rel}} \big|\vec{p}_{\textup{rel}} \big|  \frac{ e^{-\frac{\beta}{2m_{*}} \vec{p}_{\textup{rel}}^{2}}   }{\big(2\pi m_{*} \beta^{-1}   \big)^{\frac{3}{2}}   } \\ &\times \int_{\Omega} \Big(  \frac{\beta^{2}}{4M^{2}}\big| \vec{p}_{\textup{rel}}+\frac{m_{*}}{M}\vec{p}_{\textup{cm}} \big|^{2}\,\big| \mathbf{f}\big(|\vec{p}_{\textup{rel}}|, \theta   \big)\big|^{2}+ \frac{m_{*}^{2}}{M^{2}}\big|\partial_{\mathbf{z}}\mathbf{f}\big(|\vec{p}_{\textup{rel}}|, \theta   \big)\big|^{2}   \Big) \\  =& \frac{\eta}{m_{*}}\int_{\R^{3}}d\vec{p}_{\textup{rel}} \big|\vec{p}_{\textup{rel}} \big|  \frac{ e^{-\frac{\beta}{2m_{*}} \vec{p}_{\textup{rel}}^{2}}   }{\big(2\pi m_{*} \beta^{-1}   \big)^{\frac{3}{2}}   } \Big(  \frac{\beta^{2}}{4M^{2}}\big( |\vec{p}_{\textup{rel}}|^{2}+3\frac{m}{M}\frac{m_{*}}{\beta } \big)\,\sigma_{\textup{tot}}(|\vec{p}_{\textup{rel}}| )+ \frac{m_{*}^{2}}{M^{2}}\sigma_{\mathbf{z}}(|\vec{p}_{\textup{rel}}|)   \Big) \\  =&4\eta\frac{m_{*}}{M^{2}}\big(\frac{2m_{*} }{\pi\beta }   \big)^{\frac{1}{2}}\int_{\R_{+}}d\mathbf{q} \,\mathbf{q}^{3}    e^{- \mathbf{q}^{2}}   \Big(  \frac{\beta }{4m_{*}}\big( 2\mathbf{q}^{2}+3\frac{m}{M} \big)\, \sigma_{\textup{tot}}\big((\frac{2m_{*} }{\beta} )^{\frac{1}{2}}\mathbf{q} \big)+  \sigma_{\mathbf{z}}\big((\frac{2m_{*} }{\beta} )^{\frac{1}{2}}\mathbf{q}\big)    \Big).
\end{align*}

I have shown 
\begin{align*} \mathbf{b}_{i,j} &= \big[\big\langle 1 \big|
(\nabla^{\otimes^2}\mathcal{L}_{\vec{r}})\big|_{\vec{r}=0} \nu_{\infty}\big\rangle \big]_{i,j}\\ &= \big[\big\langle 1 \big|
\big( -\frac{\ii}{\hbar} (\nabla^{\otimes^2}h_{\vec{r}})\big|_{\vec{r}=0}+   (\nabla^{\otimes^2}\mathcal{J}_{\vec{r}})\big|_{\vec{r}=0}-(\nabla^{\otimes^2} \mathcal{E}_{\vec{r}})\big|_{\vec{r}=0}\big)\nu_{\infty}\big\rangle \big]_{i,j} \\ &=-\delta_{i,j}\frac{4\eta \hbar^{2} }{3}\frac{m_{*}}{M^{2}}\big(\frac{2m_{*} }{\pi\beta }   \big)^{\frac{1}{2}}\int_{\R_{+}}d\mathbf{q} \,\mathbf{q}^{3}    e^{- \mathbf{q}^{2}}   \Big(  \frac{\beta }{4m_{*}}\big( 2\mathbf{q}^{2}+3\frac{m}{M} \big)\, \sigma_{\textup{tot}}\big((\frac{2m_{*} }{\beta} )^{\frac{1}{2}}\mathbf{q} \big)+  \sigma_{\mathbf{z}}\big((\frac{2m_{*} }{\beta} )^{\frac{1}{2}}\mathbf{q}\big)    \Big), 
\end{align*}
which is equal to $-\delta_{i,j}D_{\textup{jps}}$. 

Finally, to find an expression for $\mathbf{c}$, I again use the rotation symmetry to get
$$\mathbf{c}_{i,j}=3^{-1}\delta_{i,j}\sum_{n=1}^{3} \langle 1|  \mathcal{L}_{0}^{(1)}(n) S \mathcal{L}_{0}^{(1)}(n)   |\nu_{\infty}\rangle.    $$
For the image of the spectral projection $\textup{I}-\mathbf{P}$ corresponding to the values $\Sigma(\mathcal{L}_{0})-\{0\}$, the reduced resolvent is equal to  $S=-\int_{0}^{\infty}dt\,e^{t\mathcal{L}_{0}}$.  The image of $\textup{I}-\mathbf{P}$ is the set of elements in $L^{1}(\R^{3})$ with integral zero.  My earlier observation that $\mathbf{a}=0$ means  $\langle 1| \mathcal{L}_{0}^{(1)}(n)   |\nu_{\infty}\rangle=0 $, $n=1,2,3$ so $\mathcal{L}_{0}^{(1)}(n)  \nu_{\infty}$ is in the image of $\textup{I}-\mathbf{P}$.  Hence, the elements of $\mathbf{c}$ can be expressed as          
\begin{align} \label{Cee}
\mathbf{c}_{i,j}&= -\frac{\delta_{i,j}}{3}\sum_{n=1}^{3}\int_{0}^{\infty}dt\, \big\langle 1 \big|  \mathcal{L}_{0}^{(1)}(n) e^{t\mathcal{L}_{0} } \mathcal{L}_{0}^{(1)}(n)   \big|\nu_{\infty}\big\rangle \nonumber \\ &= \frac{\delta_{i,j}}{3}\sum_{n=1}^{3}\int_{0}^{\infty}dt\, \Big(\big\langle 1 \big|\, \frak{v}_{n} e^{t\mathcal{L}_{0} }\,\frak{v}_{n}  \big|\nu_{\infty}\big\rangle + \big\langle 1 \big|\, \frak{v}_{n} e^{t\mathcal{L}_{0} }\,\frak{w}_{n}  \big|\nu_{\infty}\big\rangle \Big)
\nonumber \\ & = \frac{\delta_{i,j}}{3}\sum_{n=1}^{3}\int_{0}^{\infty}dt\, \big\langle 1 \big|\, \frak{v}_{n} e^{t\mathcal{L}_{0} }\,\frak{v}_{n}  \big|\nu_{\infty}\big\rangle
\end{align}
where $\frak{v}_{n}$ and $\frak{w}_{n}$ are multiplication operators acting on  $L^{1}(\R^{3})$ whose corresponding functions (denoted by the same symbol) are specified below.  The functions $\frak{v}_{n},\frak{w}_{n}$ are the $n$th components  of $\frak{v},\frak{w}:\R^{3}\rightarrow \R^{3} $ for $\frak{v}$ defined as in~(\ref{VelocityFunction}) and 
$$\frak{w}(\vec{p}) := \frac{1}{\nu_{\infty}(\vec{p})} \frac{\hbar}{2}\int_{\R^{3}} d\vec{q}\,w_{\vec{q}}^{(1)}(\vec{p}-\vec{q})\nu_{\infty}(\vec{p}-\vec{q}) - \frak{v}(\vec{p}) .  $$
The second inequality in~(\ref{Cee}) follows from the relation $\big\langle 1 \big|  \mathcal{L}_{0}^{(1)}(n)= -\ii\big\langle 1 \big|  \frak{v}_{n}$ and   the definition of $\frak{w}$.
The function $\frak{w}(\vec{p})$ turns out to be zero by the detailed balance symmetry, since   
\begin{align}\label{Zizzle}
\frak{w}(\vec{p}) = & \frac{\hbar}{2} \int_{\R^{3}} d\vec{q}\,\Big( \frac{\nu_{\infty}(\vec{p}-\vec{q}) }{\nu_{\infty}(\vec{p})}w_{\vec{q}}^{(1)}(\vec{p}-\vec{q})- w_{-\vec{q}}^{(1)}(\vec{p})  \Big) \nonumber  \\ = & \frac{\hbar}{2} \int_{\R^{3}} d\vec{q}\,\frac{\nu_{\infty}^{\frac{1}{2}}(\vec{p}-\vec{q})}{\nu_{\infty}^{\frac{1}{2}}(\vec{p})}\textup{Im}\Big((\nabla_{1}\gamma_{q})(\vec{p}-\vec{q},\vec{p}-\vec{q})-(\nabla_{1} \gamma_{-q})(\vec{p},\vec{p})   \Big)\nonumber  \\ = & 0.  
\end{align}
where $\nabla_{1}$ is the gradient of the first argument $\vec{p}_{1}$ of $\gamma_{q}(\vec{p}_{1},\vec{p}_{2}) $.  The first equality in~(\ref{Zizzle}) is by the equality 
$\frak{v}(\vec{p})=\frac{\hbar}{2}\int_{\R^{3}} d\vec{q}\,w_{\vec{q}}^{(1)}(\vec{p})$.  The second equality in~(\ref{Zizzle}) follows from 
$$w_{\vec{q}}^{(1)}(\vec{p})= -((\nabla_{1}-\nabla_{2})\textup{T}_{\vec{q}})(\vec{p},\vec{p})= -2\textup{i}\textup{Im}\big[(\nabla_{1}\textup{T}_{\vec{q}})(\vec{p},\vec{p})\big] = -2\textup{i}\frac{\nu_{\infty}^{\frac{1}{2}}(\vec{p})}{\nu_{\infty}^{\frac{1}{2}}(\vec{p}+\vec{q})}\textup{Im}\big[(\nabla_{1}\gamma_{q})(\vec{p},\vec{p})\big], $$
where the second equality is $\textup{T}_{\vec{q}}(\vec{p}_{1},\vec{p}_{2})= \overline{\textup{T}_{\vec{q}}(\vec{p}_{2},\vec{p}_{1})}$ (since $\textup{T}_{\vec{q}}$ is the kernel for a positive operator), and the third is the detailed balance form~(\ref{DetailedEms}).  The last equality in~(\ref{Zizzle}) holds because $(\nabla_{1}\gamma_{q})(\vec{p}-\vec{q},\vec{p}-\vec{q})=(\nabla_{1} \gamma_{-q})(\vec{p},\vec{p})$ by~(\ref{GammaSymmetry}).

By working with the adjoint semigroup $e^{t\mathcal{L}_{0}^{*}}$, I can write the elements of $\mathbf{c}$ as
\begin{align*}
\mathbf{c}_{i,j} &=\frac{\delta_{i,j}}{3}\sum_{n=1}^{3}\int_{0}^{\infty} \int_{\R^{3}}d\vec{p}\, (e^{t\mathcal{L}_{0}^{*}}\frak{v}_{n})(\vec{p})\, \frak{v}_{n}(\vec{p})\nu_{\infty}(\vec{p})\\ &=\frac{\delta_{i,j}}{3}\int_{0}^{\infty}dt\, \mathbb{E}_{\nu_{\infty}}\big[ \frak{v}(\frak{p}_{t})\cdot \frak{v}(\frak{p}_{0})\big]\\ &=\delta_{i,j} D_{\textup{kin}},
\end{align*}
where $\frak{p}_{t}$ is a Markov process whose probability densities evolve according to the master equation~(\ref{MomDist}) starting from the stationary state $\nu_{\infty}$.

\vspace{.5cm}

\noindent Part (2):\\

This follows from the relatively bounded perturbation expansion from Part (1) of Lem.~\ref{TwoDerivatives}.

\vspace{.5cm}

\noindent Part (3):\\  

I will build up control for the Taylor expansion of $\mathbf{P}_{\vec{k}}$ through control of the resolvent  $\frac{1}{ \mathcal{L}_{\vec{k}}-z  }$ for $z$ in a certain region of $z$ bounded away form the origin.  More precisely, 
I will use that the projection $\mathbf{P}_{\vec{k}}$ can be written as  a complex integral 
\beq \label{ComplexInt}
\mathbf{P}_{\vec{k}}=-\frac{1}{2\pi \textup{i}}\int_{\gamma}\frac{1}{\mathcal{L}_{\vec{k}}-z}dz,
\eeq
for a counter-clockwise curve $\gamma$ along the circle of radius $2^{-1}\frak{g}$ around the origin.

By Part (1), I can pick $r>0$ small enough so  $|\epsilon(\vec{k})|\leq 4^{-1}\frak{g}$ for all $|\vec{k}|\leq r$ and $\Sigma(\mathcal{L}_{\vec{k}})$ contains no other values in the ball of radius $\frac{3}{4}\frak{g}$ around the origin.  Let us also pick $r$ small enough so  $\sup_{|\vec{k}|\leq r}\|(\mathcal{L}_{\vec{k}}-\mathcal{L}_{0})y\|_{1}\leq \delta\|y\|_{\mathcal{L}_{0}}$ for all $y\in L^{1}$ and some small $\delta>0$.  The resolvent expansion   
$$ \frac{1}{\mathcal{L}_{\vec{k}}-z}=  \frac{1}{\mathcal{L}_{0}-z}    \sum_{n=0}^{\infty}(-1)^{n}\Big(( \mathcal{L}_{\vec{k}}-\mathcal{L}_{0}) \frac{1}{ \mathcal{L}_{0}-z  } \Big)^{n}     $$
makes sense for $|\vec{k}|\leq r$ and $  \frac{3}{8}\frak{g}  < |z|< \frac{5}{8}\frak{g}$, since there is the following operator norm bound:
\begin{align*}
\big\|( \mathcal{L}_{z}-\mathcal{L}_{0}) \frac{1}{ \mathcal{L}_{0}-z }\big\| \leq & \delta\big\|\frac{1 }{\mathcal{L}_{0}-z   } \big\|_{\mathcal{L}_{0} } \\ = & \delta\Big(\big\| \frac{1}{ \mathcal{L}_{0}-z}  \big\|+ \big\| \frac{\mathcal{L}_{0} }{ \mathcal{L}_{0}-z}  \big\|    \Big)\\ < &  \delta\Big( 2+\frac{5}{2}C_{\frac{\frak{g}}{8}}+\frac{5}{\frak{g}}+\frac{ 4 }{\frak{g}  }C_{\frac{\frak{g}}{8}}  \Big)  \leq 2^{-1}. 
\end{align*}
The second inequality above is for some $C_{\frac{\frak{g}}{8}}>0$ by Part (4) of Prop.~\ref{RelBnds}, and the last inequality is for small enough $\delta$.

The first two derivatives for the resolvent around $\vec{k}=0$ are
\begin{eqnarray*}
 \partial_{i}\frac{1}{\mathcal{L}_{\vec{v}}-z}\big|_{\vec{v}=0} &=&-\frac{1}{\mathcal{L}_{0}-z}\mathcal{L}_{0}^{(1)}(i)\frac{1}{\mathcal{L}_{0}-z},\\ \partial_{i}\partial_{j}\frac{1}{\mathcal{L}_{\vec{v}}-z}\big|_{\vec{v}=0} &=& -\frac{1}{\mathcal{L}_{0}-z}\mathcal{L}_{0}^{(2)}(i,j)\frac{1}{\mathcal{L}_{0}-z}\\ & &+\sum_{(n_{1},n_{2})= \left\{ \substack{ (i,j),\\ (j,i)}  \right.  } \frac{1}{\mathcal{L}_{0}-z} \mathcal{L}_{0}^{(1)}(n_{1})\frac{1}{\mathcal{L}_{0}-z} \mathcal{L}_{0}^{(1)}(n_{2})\frac{1}{\mathcal{L}_{0}-z}, 
 \end{eqnarray*}
where $\mathcal{L}_{0}(i)$ and $\mathcal{L}_{0}^{(2)}(i,j)$ are defined as in Lem.~\ref{TwoDerivatives}, and the sum is counted twice with $i=j$.   The third-order error $\mathbf{E}_{\vec{k}}^{(3)}(z)$ for expanding $\frac{1}{\mathcal{L}_{\vec{k}}-z}$ around $\vec{k}=0$ has the form
\begin{align*}
\mathbf{E}_{\vec{k}}^{(3)}(z):=&\frac{1}{\mathcal{L}_{\vec{k}}-z}-\frac{1}{\mathcal{L}_{0}-z}-\vec{k}\cdot\nabla\big(\frac{1}{\mathcal{L}_{\vec{v}}-z}\big)\Big|_{\vec{v}=0}-2^{-1}\vec{k}^{\otimes^{2}}\cdot\nabla^{\otimes^{2}}\big(\frac{1}{\mathcal{L}_{\vec{v}}-z}\big)\Big|_{\vec{v}=0} \\ =& -\frac{1}{\mathcal{L}_{0}-z}\hat{R}_{\vec{k}}^{(3)}\frac{1}{\mathcal{L}_{0}-z}+  \frac{1}{\mathcal{L}_{0}-z}\hat{R}_{\vec{k}}^{(2)}\frac{1}{\mathcal{L}_{0}-z}\hat{R}_{\vec{k}}^{(1)}\frac{1}{\mathcal{L}_{0}-z}\\ &+ \frac{1}{\mathcal{L}_{0}-z}\hat{R}_{\vec{k}}^{(1)}\frac{1}{\mathcal{L}_{0}-z}\hat{R}_{\vec{k}}^{(2)}\frac{1}{\mathcal{L}_{0}-z}\\ & - \frac{1}{\mathcal{L}_{0}-z}\hat{R}_{\vec{k}}^{(2)}\frac{1}{\mathcal{L}_{0}-z}\hat{R}_{\vec{k}}^{(2)}\frac{1}{\mathcal{L}_{0}-z}+
 \frac{1}{\mathcal{L}_{0}-z}    \sum_{n=3}^{\infty}(-1)^{n}\Big(( \mathcal{L}_{\vec{k}}-\mathcal{L}_{0}) \frac{1}{ \mathcal{L}_{0}-z  } \Big)^{n}, 
\end{align*}
where  $\hat{R}_{\vec{k}}^{(n)}$ is the error for the $(n-1)$th order Taylor expansion of $\mathcal{L}_{\vec{k}}$ around zero:      
$$ \hat{R}_{\vec{k}}^{(n)}= \mathcal{L}_{\vec{k}}-\mathcal{L}_{0}-\sum_{j=1}^{n-1}\frac{1}{j!}\vec{k}^{\otimes^{j}}\cdot\nabla^{\otimes^j}\big(\frac{1}{\mathcal{L}_{\vec{v}}-z}\big)\Big|_{\vec{v}=0}.    $$

The technique for bounding the second derivative and the error $\mathbf{E}_{\vec{k}}^{(3)}(z)$ will be clear from the analysis of the first derivatives $\partial_{i}\frac{1}{\mathcal{L}_{\vec{v}}-z}\big|_{\vec{v}=0}$. 
   For $|z|=2^{-1}\frak{g}$, I have the following string of inequalities. 
\begin{align*} \| \big( \partial_{i}\frac{1}{\mathcal{L}_{\vec{v}}-z}\big|_{\vec{v}=0}\big)y \|_{\mathcal{L}_{0}^{n}} & < 4(1+\frac{1}{\frak{g}} )(1+c_{\frac{\frak{g}}{4} })\|\mathcal{L}_{0}^{(1)}(i)\frac{1}{\mathcal{L}_{0}-z}y\|_{\mathcal{L}_{0}^{n-1}}\\ &\leq 4(1+\frac{1}{\frak{g}} )(1+c_{\frac{\frak{g}}{4} })C_{n}\|\frac{1}{\mathcal{L}_{0}-z}y\|_{\mathcal{L}_{0}^{n}}\\ & < 16(1+\frac{1}{\frak{g}} )^{2}(1+c_{\frac{\frak{g}}{4} })^{2} C_{n}\|y\|_{\mathcal{L}_{0}^{n-1}}.\end{align*}
The first and last inequalities hold for some $c_{\frac{\frak{g}}{4} }$ by the bounds on $\|\frac{1}{\mathcal{L}_{0}-z}\|$ and $\|\frac{\mathcal{L}_{0}}{\mathcal{L}_{0}-z}\|$ from  Part (4) of Prop.~\ref{RelBnds}.  The second equality above is by Part (2) of Lem.~\ref{TwoDerivatives}.  Since the integral~(\ref{ComplexInt}) is over a compact curve and all the estimates are uniform, I obtain the necessary bound on the value of $\|(\partial_{i}\mathbf{P}_{\vec{v}}|_{\vec{v}=0}) y\|_{\mathcal{L}_{0}^{n}}$.

\vspace{.5cm}

\noindent Part (4):\\

For any ball $B_{r}$ of radius $r$ around the origin, there are constants $a,b>0$ such that 
$$\hspace{3cm} \sup_{|\vec{k}|\leq r}\sum_{j=1,2,3}\|\partial_{j}\mathcal{L}_{\vec{k}}y\|_{1}\leq a\|\mathcal{L}_{0}y\|_{1}+b\|y\|_{1}, \hspace{2cm} y\in L^{1}(\R^{3}).  $$
Consequently, for $|\vec{k}|\leq r$, the operators   $\mathcal{L}_{\vec{k}}-\mathcal{L}_{0}$ are relatively bounded to $\mathcal{L}_{0}$ with uniform constants given by
\begin{align*}
\|(\mathcal{L}_{\vec{k}}-\mathcal{L}_{0})y\|_{1} &\leq \int_{0}^{1}dv \big\| \vec{k} \nabla \mathcal{L}_{v\vec{k}}y \big\|_{1}\\ &\leq  3a r\|\mathcal{L}_{0}y\|_{1}+3br\|y\|_{1}.  
\end{align*}
I can pick $r$ to make the coefficients $3ar$ and $3br$ arbitrarily small.  In particular, I can pick $r$ so that $$\Sigma(\mathcal{L}_{\vec{k}})\subset \Sigma(\mathcal{L}_{0})+B_{\delta},  $$
and the non-degenerate eigenvalue $\epsilon(\vec{k})$ is the only spectral value within a radius $\delta$ of $0$.  It follows the rest of the spectrum has real part less than $-\frak{g}+\delta$, and I can take  $\delta<2^{-1}\frak{g}$.

\vspace{.5cm}

\noindent Part (5):\\

 By similar reasoning as for Part (3) of Prop.~\ref{RelBnds},
there is $c>0$ such that for all $t\in \R_{+}$ and $|\vec{k}|\leq r$ for $r>0$ from Part (4),
\begin{align}\label{Trim}
  \big\| e^{t\mathcal{L}_{\vec{k}} }\big( \textup{I}-\mathbf{P}_{\vec{k}}\big)  \big\| \leq ce^{-t\frac{\frak{g}}{2} }.  
\end{align}
 
The difference between $\mathcal{L}_{\vec{k}}$ and $\mathcal{L}_{\vec{k}}'$ is bounded, and thus $\mathcal{L}_{\vec{k}}'$ is a closed operator with domain $D(\mathcal{L}_{\vec{k}})$. The semigroup $e^{t\mathcal{L}_{\vec{k}}'}$ can be expressed as
$$ e^{t\mathcal{L}_{\vec{k}}'}= e^{t\frac{\frak{g}}{2} } e^{t\mathcal{L}_{\vec{k}}}\big(\textup{I}-\mathbf{P}_{\vec{k}}\big)+e^{t\epsilon(\vec{k})}\mathbf{P}_{\vec{k}}.$$ 
By the triangle inequality and~(\ref{Trim}),
\begin{align*}
\| e^{t\mathcal{L}_{\vec{k}}'}\big\| \leq e^{t\frac{\frak{g}}{2} }\big\| e^{t\mathcal{L}_{\vec{k}} }(\textup{I}-\mathbf{P}_{\vec{k}})\|+e^{t\epsilon(\vec{k})}\|\mathbf{P}_{\vec{k}}\|  \leq c+1.
\end{align*}

\end{proof}

\section{Proof of main result}\label{SecProofMain}

With the results from previous sections, the proof of the convergence for the rescaled position distribution in Thm.~\ref{MainThm} has the pattern found in~\cite{Diff}.  The convergence of the moments requires more effort due to the unbounded operators in this case.

\begin{proof}[Proof of Thm.~\ref{MainThm}]
Parts (i) and (ii) below treat respectively  the distributional convergence of the measure $t^{\frac{3}{2}}\mu_{t}(t^{\frac{1}{2}}\vec{r} )$ as $t\rightarrow \infty$ and the convergence of its second moments. The notation $\partial_{j}$, $j=1,2,3 $ and $\nabla$ will refer to partial derivatives and the gradient, respectively,  of the fiber variable $\vec{k  }\in \R^{3}$. \vspace{.5cm}

\noindent (i).\hspace{.1cm}  To prove the distributional convergence of the probability densities $t^{\frac{3}{2}}\mu_{t}(t^{\frac{1}{2}}\vec{r} )$
to a Gaussian, I will prove the pointwise convergence of the characteristic functions $$\varphi_{t}(\vec{s})= \int_{\R^{3}}d\vec{r}\,t^{\frac{3}{2}}\mu_{t}(t^{\frac{1}{2}}\vec{r} )  e^{\ii\vec{r}\vec{s}}$$ to $e^{-\frac{D}{2}\vec{s}^{2}}  $.  The characteristic function for $\mu_{t}$ can be written
\begin{align}\label{Alt0}
\tilde{\mu}_{t}(\vec{k})& =\int_{\R^{3}} d\vec{x}\,\mu_{t}(\vec{x})\,e^{\ii\vec{k}\vec{x}}=\Tr\big[\rho_{t}e^{\ii\vec{k}\vec{X} }\big]\nonumber  \\  &=\int_{\R^{3}} d\vec{p}\,\rho_{t}(\vec{p},\vec{p}+\hbar\vec{k}) =\int_{\R^{3}}d\vec{p}\,[\rho_{t}]_{\vec{k}}(\vec{p})\nonumber \\ &=\int_{\R^{3}}d\vec{p}\,( e^{t\mathcal{L}_{\vec{k}}} [\rho]_{\vec{k}})(\vec{p})=\langle 1\big|\, e^{t\mathcal{L}_{\vec{k} }}[\rho]_{\vec{k} }\rangle, 
\end{align}
where I have used that $\rho_{t}e^{i\vec{k}\vec{X} }$ is a trace class with integral vernal $\mathcal{K}(p_{1},p_{2})=\rho_{t}(p_{1},p_{2}+\hbar\vec{k})$ for the third equality,  a change of integration variable $\vec{p}+2^{-1}\hbar \vec{k}\rightarrow \vec{p}$ for the fourth, and the fiber decomposition for the fifth. The last expression is the evaluation of the vector $e^{t\mathcal{L}_{\vec{k} }}[\rho]_{\vec{k} }\in L^{1}(\R^{3}) $ for the linear functional determined by integration against the constant function $1_{\R^{3}}\in L^{\infty}(\R^{3})$.  Fix $\vec{s}\in \R^{3}$ and let $\vec{k}= t^{-\frac{1}{2}}\vec{s}$. By the change of variable $t^{\frac{1}{2}}\vec{r}\rightarrow \vec{r}$, 
\begin{align}\label{Alt}
\varphi_{t}(\vec{s})= \int_{\R^{3}}d\vec{x}\,\mu_{t}(\vec{x} )  e^{\ii t^{-\frac{1}{2}}\vec{x}\vec{s}}= \tilde{\mu}_{t}(\vec{k})  =\langle 1\big|\, e^{t\mathcal{L}_{\vec{k} }}[\rho]_{\vec{k} }\rangle, 
\end{align}
where the last equality is from~(\ref{Alt0}).

Now I study the large $t$ asymptotics for  $e^{t\mathcal{L}_{\vec{k} }}[\rho]_{\vec{k} }\in L^{1}(\R^{3})$. By Prop.~\ref{Perturbation}, for small enough $\vec{k}$ there is an isolated non-degenerate eigenvalue $\epsilon(\vec{k})$ with projection 
$\mathbf{P}_{\vec{k}}$ such that the remainder of the spectral values have real part falling below $-2^{-1}\frak{g}$.  The operator    $e^{t\mathcal{L}_{\vec{k}} }$ can be written as
\begin{align}\label{Zoo}
e^{t\mathcal{L}_{\vec{k}} }= e^{t\epsilon(\vec{k}) }\mathbf{P}_{\vec{k}}+ e^{t\mathcal{L}_{\vec{k} }}(\textup{I}-\mathbf{P}_{\vec{k}}),    
\end{align}  
where the operator norm of $e^{t\mathcal{L}_{ \vec{k} }}(\textup{I}-\mathbf{P}_{\vec{k}})$ decays with the exponential rate $e^{-\frac{\frak{g}}{2} t}$ by Part (5) of Prop.~\ref{Perturbation}. 
Thus, by Parts (1), (2), and (5) of Prop.~\ref{Perturbation}, 
\begin{align*}
 e^{t\mathcal{L}_{ \vec{k} }}[\rho]_{ \vec{k} }=   e^{-\frac{D}{2}\vec{s}^{2}   }\mathbf{P} [\rho]_{ \vec{k} }    +\mathit{O}(t^{-\frac{1}{2}}),
 \end{align*}
where the error is with respect to the $\|\cdot\|_{1}$ norm.  By  Prop.~\ref{FiberStuff}, the condition  $X_{j}\rho$, $j=1,2,3$ (and equivalently $\rho X_{j}$) is trace class implies  $[\rho]_{\vec{k}}$ has a bounded derivative in $\vec{k}$ in the $L^{1}(\R^{3})$ norm.  In particular, $[\rho]_{\vec{k} }=[\rho]_{0}+\mathit{O}(t^{-\frac{1}{2}})$, and it follows  
\begin{align}\label{Exp}
 e^{t\mathcal{L}_{\vec{k} }}[\rho]_{ \vec{k} }=   e^{-\frac{D}{2}\vec{s}^{2}   }\nu_{\infty}    +\mathit{O}(t^{-\frac{1}{2}}),
 \end{align}
since $\langle 1| [\rho]_{0}\rangle=\int_{\R}d\vec{p}[\rho]_{0}(\vec{p})=1$.  Plugging~(\ref{Exp}) in to~(\ref{Alt}), then I conclude  
$$ \varphi_{t}(\vec{s})= e^{-\frac{D}{2}\vec{s}^{2}   }   +\mathit{O}(t^{-\frac{1}{2}}).  $$
This proves the convergence of $t^{\frac{3}{2}}\mu_{t}(t^{\frac{1}{2}}\vec{r} )$ to a normal distribution with variance $D$.  

\vspace{.5cm}
\noindent (ii).\hspace{.1cm} To show  the second moments converge, I can rewrite  the variance in terms of $e^{t\mathcal{L}_{\vec{k} }}[\rho]_{\vec{k} }$ as
\begin{align}\label{Chong}
 t^{-1}\Tr\big[X_{i}X_{j}\rho_{t}]  = & t^{-1}\int_{\R^{3}}d\vec{x}\,x_{i}x_{j}\mu_{t}(\vec{x})\nonumber \\  =&-t^{-1}\partial_{i}\partial_{j}\langle 1\big|\, e^{t\mathcal{L}_{\vec{k} }}[\rho]_{\vec{k} }\rangle\big|_{\vec{k}=0}\nonumber  \\  = & t^{-1}\partial_{i}\partial_{j}\big\langle 1\big|\, e^{t\epsilon(\vec{k}) }\mathbf{P}_{\vec{k}}[\rho]_{\vec{k} }\big\rangle\big|_{\vec{k}=0}\nonumber \\ &+t^{-1}e^{-t\frac{\frak{g}}{2} }\partial_{i}\partial_{j}\big\langle 1\big|\, e^{t(\mathcal{L}_{\vec{k} }+\frac{\frak{g}}{2})   }(\textup{I}-\mathbf{P}_{\vec{k}})[\rho]_{\vec{k} }\big \rangle\big|_{\vec{k}=0}, 
 \end{align}
where the second equality follows from differentiating~(\ref{Alt0}), and the third is~(\ref{Zoo}).  However, the differentiability of $[\rho]_{\vec{k}}$, $\epsilon(\vec{k})$, and $\mathbf{P}_{\vec{k}}$ gives 
$$t^{-1}\partial_{i}\partial_{j}\big\langle 1\big|\, e^{t\epsilon(\vec{k}) }\mathbf{P}_{\vec{k}}[\rho]_{\vec{k} }\big\rangle\big|_{\vec{k}=0}=\delta_{i,j}D+\mathit{O}(t^{-1}),$$   
since $\nabla \epsilon(\vec{k})=0$.

The remainder of the proof is concerned with bounding the last term in~(\ref{Chong}).  Define $\mathcal{L}_{\vec{k}}^{\prime}=  \mathcal{L}_{\vec{k} }+2^{-1}\frak{g}(\textup{I}-\mathbf{P}_{\vec{k}})$ as in Part (5) of Prop.~\ref{Perturbation}.  Then similarly to~(\ref{Zoo}),  $e^{t\mathcal{L}_{\vec{k}}^{\prime}}=e^{t\epsilon(\vec{k})}\mathbf{P}_{\vec{k}}+e^{t(\mathcal{L}_{\vec{k}}+ \frac{\frak{g}}{2})}(\textup{I}-\mathbf{P}_{\vec{k}}) $.
   By the triangle inequality 
\begin{align*}
\big|\partial_{i}\partial_{j}\big\langle 1\big|\, e^{t(\mathcal{L}_{\vec{k} }+\frac{\frak{g}}{2} )   }(\textup{I}-\mathbf{P}_{\vec{k}})[\rho]_{\vec{k} }\big \rangle\big|_{\vec{k}=0}\big| & = \big|\partial_{i}\partial_{j}\big\langle 1\big|\, \big(e^{t\mathcal{L}_{\vec{k}}^{\prime}   } -e^{t\epsilon(\vec{k})}\mathbf{P}_{\vec{k}}  \big)[\rho]_{\vec{k} } \big \rangle\big|_{\vec{k}=0}\big|\\ 
 &\leq \big|\partial_{i}\partial_{j}\big\langle 1\big|\, e^{t\mathcal{L}_{\vec{k}}^{\prime}   } [\rho]_{\vec{k} }\big \rangle\big|_{\vec{k}=0}\big|+\big|\partial_{i}\partial_{j}\big\langle 1\big|\,  e^{t\epsilon(\vec{k}) }\mathbf{P}_{\vec{k}}  [\rho]_{\vec{k} }\big \rangle\big|_{\vec{k}=0}\big|  \\ & \leq \|\partial_{i}\partial_{j} \big(e^{t\mathcal{L}_{\vec{k}}^{\prime}   } [\rho]_{\vec{k} }\big)|_{\vec{k}=0}\|_{1}+\|\partial_{i}\partial_{j}\big(e^{t\epsilon(\vec{k})}\mathbf{P}_{\vec{k}}  [\rho]_{\vec{k} }\big)|_{\vec{k}=0} \|_{1}.
\end{align*}
The rightmost term is finite by the comments above.  I claim  $ \|\partial_{i}\partial_{j} \big(e^{t\mathcal{L}_{\vec{k}}^{\prime}   } [\rho]_{\vec{k} }\big)|_{\vec{k}=0}\|_{1}\leq c(1+t^{2}) $ for some $c>0$.  This would complete the proof, since the  last term in~(\ref{Chong}) has the exponentially decaying factor $e^{-t\frac{\frak{g}}{2} }$. 

 By the product rule, there are four terms to deal with:
\begin{align} \label{OldProd}
\partial_{i}\partial_{j}\big( e^{t\mathcal{L}_{\vec{k}}^{\prime}   } [\rho]_{\vec{k} }\big)|_{\vec{k}=0}\nonumber  = & \big(\partial_{i}\partial_{j}e^{t\mathcal{L}_{\vec{k}}^{\prime}   }|_{\vec{k}=0}\big) [\rho]_{0}+\big( \partial_{i} e^{t\mathcal{L}_{\vec{k}}^{\prime}   }|_{\vec{k=0}}\big)\big( \partial_{j}[\rho]_{\vec{k}}|_{\vec{k}=0 }\big) \\ &+\big( \partial_{j} e^{t\mathcal{L}_{\vec{k} }^{\prime} }|_{\vec{k}=0} \big)\big( \partial_{i}[\rho]_{\vec{k}}|_{\vec{k}=0 }\big)+   e^{t\mathcal{L}_{0}^{\prime}   } \big(\partial_{i}\partial_{j}[\rho]_{\vec{k}}|_{\vec{k}=0 }\big),
\end{align}
where the middle two terms are similar.  
 The formal expressions for $\partial_{i} e^{t\mathcal{L}_{\vec{k}}^{\prime}   }|_{\vec{k}=0}:\mathcal{T}_{1}\rightarrow L^{1}(\R^{3})$ and  $\partial_{i}\partial_{j}e^{t\mathcal{L}_{\vec{k} }^{\prime}   }|_{\vec{k}=0}:\mathcal{T}_{2}\rightarrow L^{1}(\R^{3})$ are given by~(\ref{Duhamel1}) and~(\ref{Duhamel2}).  Part of the analysis involves showing   $\|\partial_{i} e^{t\mathcal{L}_{\vec{k}}^{\prime}   }|_{\vec{k}=0}y\|_{1}$ and $\|\partial_{i}\partial_{j}e^{t\mathcal{L}_{\vec{k} }^{\prime}   }|_{\vec{k}=0}y\|_{1}$ are bounded by the multiples of the norms $\|y\|_{\frak{M}}$ and $\|y\|_{\frak{M}^{2}}$, respectively.  The existence of the first and second derivatives of $e^{t\mathcal{L}_{\vec{k}}^{\prime}   }$ is established with the third-order error~(\ref{Error1}) analysis at the end of the proof.  

Now, I bound the terms on the right side of~(\ref{OldProd}).  By Part (5) of Prop.~\ref{Perturbation}, for all $\vec{k}\in \R^{3}$ in a small enough neighborhood around the origin and all $t\geq 0$, there is a $C>0$ such that 
\begin{align}\label{Bibi}
  \big\|e^{t\mathcal{L}_{\vec{k}}^{\prime}   }\big\|\leq C.
\end{align}
In words, $ e^{t\mathcal{L}_{\vec{k}}^{\prime}   }$ is a bounded semigroup.  The last term in~(\ref{OldProd}) is $\mathit{O}(1)$,  by~(\ref{Bibi}) and the differentiability of $[\rho]_{\vec{k}}$.  To bound $\big(\partial_{i} e^{t\mathcal{L}_{\vec{k}}^{\prime}   }|_{\vec{k}=0}\big) \big(\partial_{j} [\rho]_{\vec{k} }|_{\vec{k}=0}\big)$ in $L^{1}(\R^{3})$, note  
the derivative of $e^{t\mathcal{L}_{\vec{k}}^{\prime}}$ at $\vec{k}=0$ has the form
\beq \label{Duhamel1}
\partial_{i} e^{t\mathcal{L}_{\vec{k}}^{\prime}   }|_{\vec{k}=0}= \int_{0}^{t}dt_{1}e^{(t-t_{1})\mathcal{L}_{0}^{\prime} }(\partial_{i} \mathcal{L}^{\prime}_{\vec{k}}|_{\vec{k}=0} ) e^{t_{1}\mathcal{L}_{0}^{\prime}   },   
\eeq
where $\partial_{i} \mathcal{L}^{\prime}_{\vec{k}} |_{\vec{k}=0}=\partial_{i} \mathcal{L}_{\vec{k}}|_{\vec{k}=0}- 2^{-1}\frak{g}\,\partial_{i}\mathbf{P}_{\vec{k}}|_{\vec{k}=0} $.

For $y\in L^{1}(\R^{3})$,
\begin{align}\label{Bambam}
\| \partial_{i} e^{t\mathcal{L}_{\vec{k}}^{\prime}   }|_{\vec{k}=0}y \|_{1}&\leq C\int_{0}^{t}dt_{1}\big\| \partial_{i} \mathcal{L}^{\prime}_{\vec{k}} |_{\vec{k}=0} e^{t_{1}\mathcal{L}_{0}^{\prime}   }y\big\|_{1}\nonumber \\ &\leq C\int_{0}^{t}dt_{1}\big\| \partial_{i} \mathcal{L}_{\vec{k}} |_{\vec{k}=0} e^{t_{1}\mathcal{L}_{0}^{\prime}   }y\big\|_{1}+2^{-1}C\frak{g}\int_{0}^{t}dt_{1}\big\| \partial_{i}\mathbf{P}_{\vec{k}}|_{\vec{k}=0} e^{t_{1}\mathcal{L}_{0}^{\prime}   }y\big\|_{1}\nonumber  \\ &\leq  CC_{0}\int_{0}^{t}dt_{1}\big\|  e^{t_{1}\mathcal{L}_{0}^{\prime}   }y\big\|_{\mathcal{L}_{0}}+2^{-1}Cc_{0}\frak{g}\int_{0}^{t}dt_{1}\big\|  e^{t_{1}\mathcal{L}_{0}^{\prime}   }y\big\|_{1}\nonumber \\ &\leq tC^{2}\big(C_{0}+ 2^{-1}c_{0}\frak{g}  )\| y \|_{\mathcal{L}_{0}}\leq t\mathbf{c} C^{2}\big( C_{0}+ 2^{-1}c_{0}\frak{g}   )\|  y \|_{\frak{M}}.     
\end{align}
The first inequality above uses the bound for $e^{t\mathcal{L}_{0}   }$ above, and the second is the triangle inequality.   For the third inequality, I use that there some $C_{0}>0$ such that  $\|\partial_{i} \mathcal{L}_{\vec{k}} |_{\vec{k}=0}y\|_{1}\leq C_{0}\|y\|_{\mathcal{L}_{0}}$ for all $y$ by Part (1) of Lem.~\ref{TwoDerivatives} and  $\max_{i}\| \partial_{i}\mathbf{P}_{\vec{k}}|_{\vec{k}=0}\|_{\infty}=c_{0}<\infty $ by Part (3) of Prop.~\ref{Perturbation}.  The fourth inequality follows because $e^{t\mathcal{L}_{0}^{\prime}   }$ and $\mathcal{L}_{0}$ commute and then by~(\ref{Bibi}) again.  The last inequality is $\|\cdot\|_{\mathcal{L}_{0}}\leq \mathbf{c}\|\cdot\|_{\frak{M}}$ where $\mathbf{c}$ exists by the equivalence of the norms $\|\cdot\|_{\mathcal{L}_{0}}$ and $\|\cdot\|_{\frak{M}}$ from Part (5) of Prop.~\ref{RelBnds}.  With the above, 
\begin{align*}
 \| \big(\partial_{i}e^{t\mathcal{L}_{\vec{k}}^{\prime}   }|_{\vec{k}=0} \big)\big(\partial_{j} [\rho]_{\vec{k}}|_{\vec{k}=0} \big) \|_{1}& \leq t\mathbf{c} C^{2}\big( C_{0}+ 2^{-1}c_{0}\frak{g}   )\|\big(\partial_{j} [\rho]_{\vec{k}}|_{\vec{k}=0} \big) \|_{\frak{M}}\\ &= t\mathbf{c} C^{2}\big( C_{0}+ 2^{-1}c_{0}\frak{g}   )(\| \big(\partial_{j} [\rho]_{\vec{k}}|_{\vec{k}=0} \big)\|_{1}+\| |\vec{p}|\big(\partial_{j} [\rho]_{\vec{k}}|_{\vec{k}=0} \big)\|_{1})\\ & \leq t\mathbf{c} C^{2}\big( C_{0}+ 2^{-1}c_{0}\frak{g}   )\big(2^{-1}\|\{X_{j},\rho   \}\|_{\mathbf{1}}+  2^{-2}\sum_{n=1}^{3}\|\{ P_{n},\{X_{j},\rho   \}\}\|_{\mathbf{1}}\big).  
 \end{align*}
The second inequality is by Prop.~\ref{FiberStuff}, and the right side is finite by the assumptions on $\rho$.

For the $\big( \partial_{i}\partial_{j}e^{t\mathcal{L}_{\vec{k} }^{\prime}   }|_{\vec{k}=0}\big) [\rho]_{0} $ term,
\begin{align}\label{Duhamel2}
 \partial_{i}\partial_{j}e^{t\mathcal{L}_{\vec{k} }^{\prime}   }|_{\vec{k}=0}  = &   \int_{0}^{t}dt_{1}e^{(t-t_{1})\mathcal{L}_{0}^{\prime} }(\partial_{i}\partial_{j}\mathcal{L}^{\prime}|_{\vec{k}=0}) e^{t_{1}\mathcal{L}_{0}^{\prime}   }\nonumber \\ &+ \sum_{(n_{1},n_{2})=\left\{ \substack{(i,j)\\ (j,i) } \right. }\int_{0}^{t}dt_{2}\int_{0}^{t_{2}}dt_{1} e^{(t-t_{2})\mathcal{L}_{0}^{\prime} }(\partial_{n_{1} }\mathcal{L}^{\prime}|_{\vec{k}=0}) e^{(t_{2}-t_{1})\mathcal{L}_{0}^{\prime} }(\partial_{n_{2}}\mathcal{L} ^{\prime}|_{\vec{k}=0})e^{t_{1}\mathcal{L}_{0}^{\prime}},
 \end{align}
where the sum over $(n_{1},n_{2})$ counts twice when $i=j$.  The first term on the right side is relatively bounded to $\frak{M}$ by the same argument as for~(\ref{Duhamel1}).  Also by the steps in~(\ref{Bambam}) up to the next to last inequality,  then a single term from the sum in~(\ref{Duhamel2}) can be bounded by 
\begin{align*}
\Big\| \int_{0}^{t}dt_{2}\int_{0}^{t_{2}}& dt_{1}  e^{(t-t_{2})\mathcal{L}_{0}' }(\partial_{i }\mathcal{L}'_{\vec{k}}|_{\vec{k}=0} ) e^{(t_{2}-t_{1})\mathcal{L}_{0}' }(\partial_{j }\mathcal{L}'_{\vec{k}}|_{\vec{k}=0}) e^{t_{1}\mathcal{L}_{0}'}y \Big\|_{1}
\\ &\leq C^{2}\big(C_{0}+ 2^{-1}c_{0}\frak{g}  ) 
\int_{0}^{t}dt_{2}\int_{0}^{t_{2}}dt_{1} \big\|  (\partial_{j }\mathcal{L}'_{\vec{k}}|_{\vec{k}=0}) e^{t_{1}\mathcal{L}_{0}^{\prime}}y \big\|_{\mathcal{L}_{0}}\\ &\leq C^{2}(C_{0}+ 2^{-1}c_{0}\frak{g}  ) 
\int_{0}^{t}dt_{2}\int_{0}^{t_{2}}dt_{1}\big(C_{0} \big\|  e^{t_{1}\mathcal{L}_{0}^{\prime}}y \big\|_{\mathcal{L}_{0}^{2}}+ 2^{-1}c_{0}\frak{g}\|e^{t_{1}\mathcal{L}_{0}^{\prime}} y\|_{1}  \big)\\ &\leq 2^{-1}\mathbf{c}C^{3}t^{2}(C_{0}+ 2^{-1}c_{0}\frak{g}  )^{2}\|y\|_{\frak{M}^{2}} .
\end{align*}
The second inequality follows by using the triangle inequality with $\partial_{j }\mathcal{L}'_{\vec{k}}|_{\vec{k}=0}=\partial_{j }\mathcal{L}_{\vec{k}}|_{\vec{k}=0}-2^{-1}\frak{g}\partial_{j}\mathbf{P}_{\vec{k}}|_{\vec{k}=0}$  and applying respectively Part (2) of Lem.~\ref{TwoDerivatives} and Part (3) of Prop.~\ref{Perturbation}, as before.

Since I have used the formal expressions~(\ref{Duhamel1}) and~(\ref{Duhamel2}) for the derivatives of $e^{t\mathcal{L}^{\prime}_{\vec{k}}}$ at $\vec{k}=0$, I should control the resulting error in using those terms in a second-order Taylor expansion.  Using a rearranged third-order Duhamel equation,  
\begin{align}\label{Error1}
\mathbf{E}_{\vec{k}}^{(3)}:=&e^{t\mathcal{L}_{\vec{k}}^{\prime}   }-e^{t\mathcal{L}_{0}   }-   \vec{k}\cdot(\nabla e^{t\mathcal{L}_{\vec{r} }^{\prime}   }|_{\vec{r}=0})-2^{-1}\vec{k}^{\otimes^{2}}\cdot(\nabla^{\otimes^{2}} e^{t\mathcal{L}_{\vec{r} }^{\prime}   }|_{\vec{r}=0}) \nonumber       \\ =&\int_{0}^{t}dt_{1}e^{(t-t_{1})\mathcal{L}_{\vec{k}}^{\prime} } \hat{R}_{\vec{k}}^{(3)}e^{t_{1}\mathcal{L}_{0} ^{\prime}  }\nonumber  \\ &+ \int_{0}^{t}dt_{2}\int_{0}^{t_{2}}dt_{1}\, e^{(t-t_{2})\mathcal{L}_{\vec{k} }^{\prime} }\hat{R}_{\vec{k}}^{(2)}e^{(t_{2}-t_{1})\mathcal{L}_{0}^{\prime} }\hat{R}_{\vec{k}}^{(1)}e^{t_{1}\mathcal{L}_{0}^{\prime}} \nonumber \\ & +\int_{0}^{t}dt_{2}\int_{0}^{t_{2}}dt_{1} e^{(t-t_{2})\mathcal{L}_{\vec{k}}^{\prime} }\hat{R}_{\vec{k}}^{(1)}e^{(t_{2}-t_{1})\mathcal{L}_{0}^{\prime}}\hat{R}_{\vec{k}}^{(2)}e^{t_{1}\mathcal{L}_{0}^{\prime}} \nonumber \\ &-\int_{0}^{t}dt_{2}\int_{0}^{t_{2}}dt_{1}\, e^{(t-t_{2})\mathcal{L}_{\vec{k}}^{\prime} }\hat{R}_{\vec{k}}^{(2)}e^{(t_{2}-t_{1})\mathcal{L}_{0}^{\prime}}\hat{R}_{\vec{k}}^{(2)}e^{t_{1}\mathcal{L}_{0}^{\prime}} \nonumber \\   &  +\int_{0}^{t}dt_{3}\int_{0}^{t_{3}}dt_{2}\int_{0}^{t_{2}}dt_{1} e^{(t-t_{3})\mathcal{L}_{\vec{k}}^{\prime} }\hat{R}_{\vec{k}}^{(1)}e^{(t_{3}-t_{2})\mathcal{L}_{0}^{\prime}}\hat{R}_{\vec{k}}^{(1)}e^{(t_{2}-t_{1})\mathcal{L}_{0}}\hat{R}_{\vec{k}}^{(1)}e^{t_{1}\mathcal{L}_{0}^{\prime}},
\end{align}
where $\hat{R}_{\vec{k}}^{(n)}$ is the error the $n$th order Taylor expansion of $\mathcal{L}_{\vec{k}}^{\prime}$ around $\vec{k}=0$:   
\beq \label{Error2}
  \hspace{2.5cm}\hat{R}_{\vec{k}}^{(n)} =\mathcal{L}_{\vec{k}}^{\prime}-\mathcal{L}_{0}^{\prime}-\sum_{m=1}^{n-1} \frac{1}{m!}\vec{k}^{\otimes^{m}}\cdot(\nabla^{\otimes^{m}} \mathcal{L}_{\vec{r}}^{\prime}|_{\vec{r}=0}), \quad \quad n\geq 0.        
  \eeq
Since $\mathcal{L}_{\vec{k}}^{\prime}=\mathcal{L}_{\vec{k}}+2^{-1}\frak{g}(\textup{I}-\mathbf{P}_{\vec{k}})$, the error $\hat{R}_{\vec{k}}^{(n)}$ is equal to 
 $ \hat{R}_{\vec{k},1}^{(n)}-2^{-1}\frak{g}\hat{R}_{\vec{k},2}^{(n)}$ where $\hat{R}_{\vec{k},1}^{(n)}$ and $\hat{R}_{\vec{k},2}^{(n)}$ are the analogously defined errors for $\mathcal{L}_{\vec{k}}$ and $\mathbf{P}_{\vec{k}}$, respectively.   
  
Using Lem.~\ref{TwoDerivatives} and Prop.~\ref{Perturbation} along with the same techniques as  above for the first two derivatives, then it can be shown  $\| \mathbf{E}_{\vec{k}}^{(3)} y \|_{1}$ is bounded by a constant multiple of $|\vec{k}|^{3}\|y\|_{\frak{M}^{3}}$ for all $y\in \mathcal{T}_{3}$.  By Prop.~\ref{FiberStuff},  
\begin{align*}
\|[\rho]_{\vec{k}}\|_{\frak{M}^{3}}& \leq \|[\rho]_{\vec{k}}\|_{1}+\sum_{i_{1},i_{2},i_{3}\in \{1,2,3\}  }\|\,p_{i_{1}}p_{i_{2}}p_{i_{3}}[\rho]_{\vec{k}}   \|_{1}\\ & \leq   \|\rho\|_{\mathbf{1}}+\sum_{i_{j}\in \{1,2,3\}  }2^{-3}\| \{ P_{i_{1}},\{P_{i_{2}},\{P_{i_{3}},\rho \}\}\}  \|_{\mathbf{1}}   .
\end{align*}   
By the triangle inequality and the self-adjointness of $\rho$, the term on the right is bounded by a sum of terms $\|P_{i_{1}}P_{i_{2}}\rho P_{i_{3}}\|_{\mathbf{1}}$ and $\|P_{i_{1}}P_{i_{2}}P_{i_{3}}\rho \|_{\mathbf{1}}$.  However,
$$\|P_{i_{1}}P_{i_{2}}\rho P_{i_{3}}\|_{\mathbf{1}}\leq \|\,|\vec{P}|^{2}\rho P_{i_{3}}\|_{\mathbf{1}}= \| P_{i_{3}}\rho |\vec{P}|^{2}\|_{\mathbf{1}}\leq \|\, |\vec{P}|\rho |\vec{P}|^{2}\|_{\mathbf{1}}<\infty.$$
 The first inequality holds since    $ ( P_{i_{3}} \rho P_{i_{2}}|P_{i_{1}}|^{2}P_{i_{2}}\rho P_{i_{3}})^{\frac{1}{2}}\leq ( P_{i_{3}} \rho |\vec{P}|^{4}\rho P_{i_{3}})^{\frac{1}{2}}$ by the operator monotonicity of the square root function.  The equality uses that the adjoint operation preserves trace norm, and second inequality is by the same reason as for the first.  Finally the right side is finite by my assumptions on $\rho$.

\end{proof}

\section*{Acknowledgments}
I thank Wojciech De Roeck for many helpful insights and  observations regarding the topic of this article.  I also thank Klaus Hornberger for kindly sending me an offprint of his article (joint with Bassano Vacchini) \emph{Quantum linear Boltzmann equation}. This work is supported by the Belgian Interuniversity Attraction Pole P6/02 and the European Research Council grant No. 227772.  Part of this work was completed at the Mittag-Leffler institute.

\begin{appendix}

\section{A Fourier transform relating $\Psi$ and the component maps $\mathbf{T}_{\vec{q}}$ }\label{FourFormal}
In this section, I point out a canonical relation between the completely positive map $\Psi$ and the maps $\mathbf{T}_{\vec{q}}$ in the form~(\ref{Jeez}) through a Fourier transform.  The statement and the ``proof" of Prop.~\ref{Rough} are heuristic.  I think the basic mechanics for the relations may interest the reader, but it is not necessary to clarify the analytic details for this article.

Note that  the  covariance properties~(\ref{BiVar}) below is equivalent to saying  the maps $\mathbf{T}_{\vec{q}} :\mathcal{B}_{1}\big(L^{2}(\R^{d})\big)   $ act as  multiplications operator in the momentum representation as in~(\ref{Ems}).

\begin{proposition}\label{Rough} \text{ }\\
\begin{enumerate}
\item Let $\Psi: \mathcal{B}_{1}\big(L^{2}(\R^{d})\big) $ be a map of the form    
\beq \label{ToMap}
\Psi(\rho)=\int_{\R^{d}}d\vec{q}\, e^{\ii \frac{\vec{q}}{\hbar}\vec{X}}\mathbf{T}_{\vec{q}}(\rho )e^{-\ii \frac{\vec{q}}{\hbar}\vec{X}}, 
\eeq
where the maps $\mathbf{T}_{\vec{p}}$ are completely positive and for all $\vec{x},\vec{q}\in \R^{d}$ satisfy 
\beq \hspace{4cm} \label{BiVar}\mathbf{T}_{\vec{q}}(\tau_{\vec{x}}\rho)=\tau_{\vec{x}}\mathbf{T}_{\vec{q}}(\rho)\quad \text{and} \quad \mathbf{T}_{\vec{q}}(\rho\tau_{\vec{x}})=\mathbf{T}_{\vec{q}}(\rho)\tau_{\vec{x}}, 
\eeq
in which  $\tau_{\vec{x}}=e^{\ii\frac{\vec{x}}{\hbar}\vec{P}}$.   Then, $\Psi$ is translation covariant and completely positive.  Moreover, the map $\Psi$ is related to the maps $\mathbf{T}_{\vec{q}}$ through the Fourier transform
\beq \label{ToComp}
\mathbf{T}_{\vec{q}}(\rho)=e^{-\ii\frac{\vec{q}}{\hbar}\vec{X}}\Big(\frac{1}{(2\pi \hbar)^{d}}\int_{\R^{d}}d\vec{x}\,e^{\ii\frac{\vec{q}\vec{x}}{\hbar} }  \tau_{\vec{x}}^{*}
\Psi(\tau_{x}\rho ) \Big) e^{\ii\frac{\vec{q}}{\hbar}\vec{X}}.  
\eeq

\item  Conversely, if $\Psi$ is completely positive and maps $\mathbf{T}_{\vec{q}}$ are defined by~(\ref{ToComp}), then the maps $\mathbf{T}_{\vec{q}}$ are completely positive and satisfy~(\ref{BiVar}).  Also, $\Psi$ can be written as~(\ref{ToMap}).

\end{enumerate}

\end{proposition}

\begin{proof}\text{  }\\
\noindent Part (1):\\

  The complete positivity of $\Psi$ follows trivially from the form~(\ref{ToMap}), and the translation covariance of $\Psi$ follows from the the left and right covariance properties~(\ref{BiVar}) and the Weyl intertwining  relation 
\beq  \label{Weyl}e^{\ii\frac{\vec{x}}{\hbar}\vec{P}}e^{i\frac{\vec{q}}{\hbar}\vec{X}}=e^{\ii\frac{\vec{x}\vec{q}}{\hbar}} e^{\ii\frac{\vec{q}}{\hbar}\vec{X}}e^{i\frac{\vec{x}}{\hbar}\vec{P}}.
\eeq
Also due to the intertwining relations~(\ref{BiVar}) and~(\ref{Weyl}),
\begin{align}\label{FourTrans}
e^{-\ii\frac{\vec{q}}{\hbar}\vec{X}} \Big(\frac{1}{(2\pi\hbar)^{d}}\int_{\R^{d}}& d\vec{x}\, e^{\ii\frac{\vec{q}\vec{x}}{\hbar} }\tau_{\vec{x}}^{*}\Psi(\tau_{\vec{x}}\rho )\Big)e^{i\frac{\vec{q}}{\hbar}\vec{X} }\nonumber  \\ &= e^{-\ii\frac{\vec{q}}{\hbar}\vec{X}} \Big(\frac{1}{(2\pi\hbar)^{d}}\int_{\R^{d}}d\vec{x}\, e^{\ii\frac{\vec{q}\vec{x}}{\hbar} }\tau_{\vec{x}}^{*}\int_{\R^{d}}d\vec{p}\, e^{\ii \frac{\vec{p}}{\hbar}\vec{X}}\mathbf{T}_{\vec{p}}(\tau_{\vec{x}}\rho )e^{-\ii \frac{\vec{p}}{\hbar}\vec{X}}\Big)e^{i\frac{\vec{q}}{\hbar}\vec{X} }\nonumber \\ & \nonumber =\int_{\R^{d}} d\vec{p} \Big(\frac{1}{(2\pi \hbar)^{d}} \int_{\R^{d}}d\vec{x}\,e^{\ii\frac{\vec{x}}{\hbar}(\vec{p}-\vec{q})}\Big)\Big(e^{\ii\frac{\vec{p}-\vec{q}}{\hbar}\vec{X}}\mathbf{T}_{\vec{p}}(\rho ) e^{-\ii\frac{\vec{p}-\vec{q}}{\hbar}\vec{X}}\Big)\\ &=\mathbf{T}_{\vec{q}}(\rho).
\end{align}
The above comes through switching the order of integration and the formal identity $\delta(\vec{p}-\vec{q})=\frac{1}{(2\pi \hbar)^{d}}\int_{\R^{d}}d\vec{x}\,e^{\ii\frac{\vec{x}}{\hbar}(\vec{q}-\vec{p})}$. 

\vspace{.5cm}

\noindent Part (2):\\

First, I show  $\mathbf{T}_{\vec{q}}$ satisfies~(\ref{BiVar}).  By the Weyl intertwining relation~(\ref{Weyl}),  
\begin{align}
 \tau_{\vec{y}}\mathbf{T}_{\vec{q}}(\rho)& =e^{-\ii\frac{\vec{q}}{\hbar}\vec{X}}\Big(\frac{1}{(2\pi \hbar)^{d}}\int_{\R^{d}}d\vec{x}\,e^{\ii\frac{\vec{q}(\vec{x}-\vec{y}) }{\hbar} }  \tau_{\vec{x}-\vec{y} }^{*}
\Psi(\tau_{x-y}\tau_{y}\rho ) \Big) e^{\ii\frac{\vec{q}}{\hbar}\vec{X}}
\nonumber \\
&=e^{-\ii\frac{\vec{q}}{\hbar}\vec{X}}\Big(\frac{1}{(2\pi \hbar)^{d}}\int_{\R^{d}}d\vec{x}\,e^{\ii\frac{\vec{q}\vec{x} }{\hbar} }  \tau_{\vec{x} }^{*}
\Psi(\tau_{x}\tau_{y}\rho ) \Big) e^{\ii\frac{\vec{q}}{\hbar}\vec{X}}\nonumber \\ &=\mathbf{T}_{\vec{q}}(\tau_{\vec{y}}\rho),\end{align}
where the second equality is a change of integration $\vec{x}-\vec{y}\rightarrow \vec{x}$, and the third is the definition of $\mathbf{T}_{\vec{q}}$.  The  same argument can be applied to show   
$\mathbf{T}_{\vec{q}}(\rho\tau_{\vec{y}})=\mathbf{T}_{\vec{q}}(\rho )\tau_{\vec{y}}$ by using  the expression $\tau_{\vec{x} }^{*}
\Psi(\tau_{x}\rho )$ in the integral~(\ref{ToComp}) can be replaced by 
$\Psi(\rho \tau_{x}^{*})\tau_{\vec{x} } $ due to  the translation covariance of $\Psi$.

Now, I show  $\mathbf{T}_{\vec{q}}$ maps positive operators to positive operators if $\Psi$ is completely positive. The argument is based on Bochner's theorem.  First notice  the translation covariance of $\Psi$ implies $\tau_{\vec{x}}\Psi(\tau_{\vec{x}}^{*}\rho\tau_{y})\tau_{y}^{*}\in \mathcal{B}_{1}\big(L^{2}(\R^{d})\big) $ is a function of $\vec{x}-\vec{y}$.  Let $\phi\in L^{2}(\R^{d})$, $\rho\in \mathcal{B}_{1}\big(L^{2}(\R^{d})\big)$ be positive,  and $F_{\phi,\rho}:\R^{d}\rightarrow \C$ be defined as    
$$F_{\phi,\rho}(\vec{x}-\vec{y}) := \langle \phi \big|  \tau_{\vec{x}}\Psi(\tau_{\vec{x}}^{*}\rho\tau_{\vec{y}})\tau_{\vec{y}}^{*} \phi \rangle.  $$           
By the complete positivity of $\Psi$, it follows  $F_{\phi,\rho}(\vec{x}-\vec{y})$ is the integral kernel for a positive operator on $L^{2}(\R^{d})$ when $\rho$ is a positive operator.  By Bochner's theorem the Fourier transform  of $F_{\phi,\rho}$ is positive-valued, so   
  \begin{align*}0\leq \int_{\R^{d}}d\vec{x}\, e^{\ii\frac{\vec{q}\vec{x}}{\hbar} } F_{\phi,\rho}(\vec{x}) &=  \Big\langle \phi \Big| \Big(\int_{\R^{d}}d\vec{x}\, e^{\ii\frac{\vec{q}\vec{x}}{\hbar} }   \tau_{\vec{x} }^{*}\Psi(\tau_{\vec{x}}\rho)\Big) \phi \Big\rangle\\ &= \big\langle e^{-\ii\frac{\vec{q}}{\hbar}\vec{X}   }\phi \big| \mathbf{T}_{\vec{q}}(\rho) e^{-\ii\frac{\vec{q}}{\hbar}\vec{X}   } \phi \big\rangle. 
  \end{align*}
Since $\phi\in L^{2}(\R^{d})$ is arbitrary, it follows that $\mathbf{T}_{\vec{q}}(\rho)$ is a positive operator.  
 The argument above extends trivially to show that  $\mathbf{T}_{\vec{q}}$ is completely positive by replacing $\Psi$, $\mathbf{T}_{\vec{q}}$ with    $\Psi\otimes I_{\textup{ext}} $, $\mathbf{T}_{\vec{q}}\otimes I_{\textup{ext}}$ operating on the extended Banach space  $\mathcal{B}_{1}\big(L^{2}(\R^{d})\otimes \mathcal{H}_{\textup{ext}} \big)$ for some Hilbert space $\mathcal{H}_{\textup{ext}}$.

A computation similar to~(\ref{FourTrans}) shows that plugging  $\mathbf{T}_{\vec{q}}$ in to the right side of~(\ref{ToMap}) yields $\Psi$.

\end{proof}

\section{Properties of the fiber maps}

The following elementary proposition gives sufficient conditions for the smoothness of the fiber elements $ [\rho]_{\vec{k}}$, $\rho\in \mathcal{B}_{1}\big(L^{2}(\R^{d})\big) $ as a function of $\vec{k}\in \R^{d}$ in the $L^{1}(\R^{d})$-norm in terms of weighted tracial semi-norms of $\rho$.   The proposition also gives bounds for the moments of $[\rho]_{\vec{k}}$.  The partial derivatives $\partial_{j}$, $j\in [1,d]$  in the proposition below refer to the fiber variable $\vec{k}$.

\begin{proposition}\label{FiberStuff}
Let $\rho\in \mathcal{B}_{1}\big(L^{2}(\R^{d})\big)$, then

\begin{enumerate}
\item 
$\| \partial_{j_{1}}\cdots \partial_{j_{n}} [\rho]_{\vec{k}}  \|_{1}\leq \frac{1}{2^{n}}\| \{X_{j_{1}},\cdots,\{X_{j_{n}} ,  \rho\}\cdots \}   \|_{\mathbf{1}},   $  

\item  $\|p_{j_{1}}\cdots p_{j_{n}}[\rho]_{\vec{k}} \|_{1}\leq \frac{1}{2^{n}}\| \{P_{j_{1}},\cdots,\{P_{j_{n}} ,  \rho\}\cdots \}   \|_{\mathbf{1}}  $,

\item
 $\| p_{i}\partial_{j}[\rho]_{\vec{k}}\|_{1}\leq \frac{1}{4}\|  \{P_{i},  \{X_{j}, \rho\}\} \|_{\mathbf{1}} $.

\end{enumerate}

\end{proposition}

\begin{proof}
By the definition of the fiber maps $[\cdot]_{\vec{k}}:\mathcal{B}_{1}\big(L^{2}(\R^{d})\big)\rightarrow L^{1}(\R^{d})$,  they satisfy
\begin{align*}
\int_{\R^{d}}d\vec{p} \,g(\vec{p})\,[\rho]_{\vec{k}}(\vec{p})&=\Tr\big[  e^{\ii \vec{k}\vec{X}}\,g(\vec{P}+2^{-1}\hbar \vec{k} ) \rho \big]\\ &=\Tr\big[  g(\vec{P})\,e^{\ii\frac{\vec{k}}{2}\vec{X}} \rho \,e^{\ii\frac{\vec{k}}{2}\vec{X}}\big]   
\end{align*}
for all $  g\in L^{\infty}(\R^{d})$. Differentiating both sides by $\partial_{j}$ and taking the absolute value,
\begin{align*}
\big|\int_{\R^{d}}d\vec{p}\,g(\vec{p})\partial_{j} [\rho]_{\vec{k}}(\vec{p})\big|& \leq 2^{-1}\big|\Tr\big[  g(\vec{P})\,e^{\ii\frac{\vec{k}}{2}\vec{X}}\{X_{j},\rho\} \,e^{\ii\frac{\vec{k}}{2}\vec{X}}\big] \big|\\ &\leq 2^{-1}\|g\|_{\infty}\|\{X_{j},\rho\}\|_{\mathbf{1}}.     
\end{align*}
Supremizing over all $\|g\|_{\infty}=1$,  the left side is equal to $\|\partial_{j}[\rho]_{\vec{k}}\|_{1}$.  Higher derivatives work by the same argument.  In order to bound the moments in Parts (2) and (3), recall  $e^{i\vec{k}\vec{X}}\rho e^{i\vec{k}\vec{X}}=e^{i\vec{k}\{\vec{X},\cdot\} }(\rho)$ and  $\{X_{j},\cdot \}$ and $\{P_{i},\cdot\}$ commute for all $1\leq i,j\leq d$.

\end{proof}

\section{Some relative bounds}

\begin{lemma}\label{Leftover}
The operators $H_{f}$ and $L_{\vec{q},\vec{v}}$ are relatively bounded to $|\vec{P}|$.  
\end{lemma}

\begin{proof}

In the expression~(\ref{ElementKraus}) for $L_{\vec{q},\vec{v}}$, the only unbounded factor for fixed $\vec{q}\in\R^{3}$, $\vec{v}\in (\vec{q})_{\perp}$ is the scattering amplitude  
\begin{multline*}
f\Big(\frac{m_{*}}{m}\vec{v}-\frac{m_{*}}{M}\vec{P}_{\perp \vec{q}}-2^{-1}\vec{q} , \,    \frac{m_{*}}{m}\vec{v}-\frac{m_{*}}{M}\vec{P}_{\perp \vec{q}}+2^{-1}\vec{q}\Big)\\ = \mathbf{f}\Big(\big|\frac{m_{*}}{m}\vec{v}-\frac{m_{*}}{M}\vec{P}_{\perp \vec{q}}+2^{-1}\vec{q} \big|, \,2\tan^{-1}\big( \frac{ 2^{-1}\big|\vec{q} \big|}{ \big|  \frac{m_{*}}{m}\vec{v}-\frac{m_{*}}{M}\vec{P}_{\perp \vec{q}}\big|}   \big)\Big).
\end{multline*}
In particular, the unboundedness of $\mathbf{f}(\mathbf{p},\theta)$ occurs as  $\mathbf{p}\rightarrow \infty$  for $\theta$ near zero.  However, there is a $c>0$ such that    
\beq \label{Turnip}
 | \mathbf{f}(\mathbf{p},\theta)|\leq c\big( a^{-1}\hbar+\mathbf{p   }\big).
\eeq 
If I show this, then $L_{\vec{q},\vec{v}}$ is relatively bounded to $|\vec{P}|$, since  
$$\Big|f\Big(\frac{m_{*}}{m}\vec{v}-\frac{m_{*}}{M}\vec{P}_{\perp \vec{q}}-2^{-1}\vec{q} , \,    \frac{m_{*}}{m}\vec{v}-\frac{m_{*}}{M}\vec{P}_{\perp \vec{q}}+2^{-1}\vec{q}\Big)\Big|\leq c\big( a^{-1}\hbar+\frac{m_{*}}{m}|\vec{v}|+2^{-1}|\vec{q}|+ \frac{m_{*}}{M}|\vec{P}|\big). $$

Showing~(\ref{Turnip}) requires a similar form argument as in Lem.~\ref{ScatAmpl}.  By the partial-wave expansion~(\ref{PartialWave}) and the fact that $|P_{\ell}(\cos(\theta))|\leq 1$, then
$$
|\mathbf{f}(\mathbf{p},\theta)|\leq\frac{\hbar}{2\mathbf{p}}\sum_{\ell =0}^{\infty}(2\ell+1)\big| S_{\ell}(\frac{a}{\hbar}\mathbf{p})  -1\big|.$$
However, for $\frac{a}{\hbar}\mathbf{p}\gg 1$, this sum is bounded by a constant multiple of $\mathbf{p}$ by the analysis in the proof of Lem.~\ref{ScatAmpl}.

For the Hamiltonian $H_{f}$,  
\begin{align*}
H_{f}(\vec{p})& \leq\frac{2\pi\hbar^{2}\eta}{ m_{*}}\int_{\R^{3}}d\vec{q}\,r(\vec{q})\,\big|\mathbf{f}\big(\big|\frac{m_{*}}{m}\vec{q}-\frac{m_{*}}{M}\vec{p}\big|,\, 0\big) \big| \\ &\leq \frac{2\pi c \hbar^{2}\eta}{ m_{*}}\int_{\R^{3}}d\vec{q}\,r(\vec{q})\,\big(a^{-1}\hbar+\big|\frac{m_{*}}{m}\vec{q}-\frac{m_{*}}{M}\vec{p} \big| \big) \\ &\leq \frac{2\pi c \hbar^{2}\eta}{ m_{*}}\Big(a^{-1}\hbar+2 \frac{m_{*}}{m}\big(\frac{2m}{\pi \beta}   \big)^{\frac{1}{2}}+\frac{m_{*}}{M}|\vec{p}|    \Big),
\end{align*}
where I have used inequality~(\ref{Turnip}), an explicit integration of the Gaussian $r(\vec{q})$, and the triangle inequality.  The above gives a linear bound in $|\vec{P}|$ for $H_{f}(\vec{P})$.

\end{proof}

\end{appendix}

\end{document}